\documentclass[11pt,reqno]{amsart}
\usepackage{amssymb}

\usepackage{amsthm,amsfonts,amssymb,euscript,hyperref,color}
\usepackage{comment}
\usepackage{mathtools}
\usepackage{stmaryrd}
\usepackage{mathrsfs}
\usepackage{pifont}
\usepackage{amsfonts,amsmath,amssymb,amsthm,amscd,cancel}
\usepackage{graphicx}
\usepackage[]{geometry}
\usepackage{verbatim}
\usepackage{mathrsfs}
\usepackage{fancyhdr}
\usepackage{footnpag}

\newtheorem{theorem}{Theorem}
\newtheorem{lemma}{Lemma}
\newtheorem{proposition}{Proposition}
\newtheorem{corollary}{Corollary}

\newtheorem{remark}{Remark}

\allowdisplaybreaks[4]

\DeclareMathAlphabet{\mathsfsl}{OT1}{cmss}{m}{sl}
\numberwithin{equation}{section}

\newcommand{\D}{\mathrm{d}}

\newcommand{\tr}{\mathrm{tr}}
\newcommand{\cir}[1]{\overset{\circ}{#1}}

\renewcommand{\iint}{\int\!\!\!\!\!\int}

\def\alphab{\underline{\alpha}}
\def\betab{\underline{\beta}}
\def\chib{\underline{\chi}}
\def\chibh{\widehat{\underline{\chi}}}
\def\chih{\widehat{\chi}}
\def\etab{\underline{\eta}}

\def\Lb{\underline{L}}
\def\mub{\underline{\mu}}
\def\kappab{\underline{\kappa}}
\def\tr{\mathrm{tr}}
\def\omegab{\underline{\omega}}

\def\tensor{\widehat{\otimes}}

\def\ub{\underline{u}}
\def\Cb{\underline{C}}

\def\Lb{\underline{L}}

\newcommand{\Db}{\underline{D}}
\newcommand{\Dh}{\widehat{D}}
\newcommand{\Dbh}{\widehat{\underline{D}}}

\def\nablas{\mbox{$\nabla \mkern -13mu /$ }}
\def\Deltas{\mbox{$\Delta \mkern -13mu /$ }}

\def\divs{\mbox{$\mathrm{div} \mkern -13mu /$ }}
\def\curls{\mbox{$\mathrm{curl} \mkern -13mu /$ }}
\def\ds{\mbox{$\mathrm{d} \mkern -9mu /$}}
\def\gs{\mbox{$g \mkern -9mu /$}}
\def\epsilons{\mbox{$\epsilon \mkern -9mu /$}}

\begin{document}

\title{On the Local Extension of the Future Null Infinity}

\author[Junbin Li]{Junbin Li}
\address{Department of Mathematics, Sun Yat-sen University\\ Guangzhou, China}
\email{mc04ljb@mail2.sysu.edu.cn}

\author[Xi-Ping Zhu]{Xi-Ping Zhu}
\address{Department of Mathematics, Sun Yat-sen University\\ Guangzhou, China}
\email{stszxp@mail.sysu.edu.cn}

\date{}

\maketitle

\begin{abstract}
We consider a characteristic problem of the vacuum Einstein equations with part of the initial data given on a future complete null cone with suitable decay, and show that the solution exists uniformly around the null cone for general such initial data. We can then define a segment of the future null infinity. The initial data are not required to be small and the decaying condition inherits from the works of \cite{Ch-K} and \cite{K-N}.

\end{abstract}

\tableofcontents

\setcounter{tocdepth}{1}

\parskip=\baselineskip

\section{Introduction}

\subsection{Introduction} It is well known that the weak cosmic censorship, one of the major problem in mathematical relativity, states that the maximal developments of generic asymptotically flat initial data possess a complete future null infinity. There are many mathematical works towards this problem. 

In 1993, Christodoulou and Klainerman \cite{Ch-K} proved the nonlinear stability of Minkowski space-time in vacuum. They showed that, beginning with any strongly asymptotically flat, maximal Cauchy initial data sets which are sufficiently close to the space-like slices in Minkowski space-time, the maximal future developments are future geodesically complete and approaching to Minkowski space-time in suitable sense. In particular, the weak cosmic censorship holds in this class.
Klainerman and Nicol\`o \cite{K-N} has proved the global existence in the domain of dependence outside a compact set in the initial data set. They constructed the space-time in a neighbourhood of space-like infinity, and does not care about what happens in the interior region of the initial data. 

According to \cite{Chr90} and \cite{K-N}, a corollary of the above works is, for any strongly asymptotically flat Cauchy data, there exists a region $\Omega_0$ with compact closure such that the boundary of the causal future $\partial J^+(\Omega_0)$ of $\Omega_0$ in the maximal development consists of future complete null generators. Christodoulou then defined the phase ``possessing a complete future null infinity'' in \cite{Chr90} as if for every large $A>0$, we can find some $\Omega$ containing $\Omega_0$ such that the boundary of the domain of dependence $\partial D^+(\Omega)$ of $\Omega$ has the property that, each of its futre null geodesic generators has an affine length measured from $\partial D^+(\Omega)\bigcap\partial J^+(\Omega_0)$, of at least $A$.

Therefore, it is natural to consider the characteristic-Cauchy mixed initial data problem, that the initial data are given on a three-dimensional disk $\Omega_0$, and the complete outgoing null cone rooted at $\partial\Omega_0$, as an approach to the weak cosmic censorship. In practice, we should consider the case that $\Omega_0$ is not necessarily a large enough region, that is, the initial data are no longer small. As a first step, we should answer whether the solution to the vacuum Einstein equations exists around the initial hypersurface in a uniform way, such that we can define a segment of the future null infinity. This is a local existence problem, where ``local'' is measured from infinity. In addition, if we firstly solve the vacuum Einstein equations on $\Omega_0$, the boundary of the domain of dependence of $\Omega_0$ can be served as a new initial null hypersurface, and we are facing a double characteristic problem with initial data given on two intersecting null cones.

The characteristic setting is quite useful because the constraints are basically ODEs along the null generators. In 1980-90s, Christodoulou considered in a series of papers the spherical symmetric solutions of Einstein equations coupled with scalar fields and gave an affirmative answer to the weak cosmic censorship in this class, see \cite{Chr99}. The initial data of this problem consists a function $\alpha_0=\partial(r\phi)/\partial s|_{C_0}$ where $s$ is the affine parameter of the null generators of $C_0$, and $C_0$ is assumed to be a complete null cone from a point towards to the future null infinity. 

In the light of Christodoulou's work on spherical symmetric Einstein-scalar field equations, we find that the characteristic setting may be a more reasonable approach to the problem. In fact, in the case that the black hole eventually forms, a global existence result for the characteristic problem leads to the event horizon in a natural way, because the event horizon is also a complete null cone, and serves as the ``final'' piece of the solution. We remark that Christodoulou considered the single characteristic problem with initial data given on a null cone from a point, but not the characteristic-Cauchy mixed or double characteristic cases. The approaches do not differ too much because the null cone which extends to infinity plays the essential role in the problem.

The local existence of the double characteristic problem for general initial data was considered in \cite{Ren} and \cite{Luk}. The initial data are given on two null cones $C_0$ and $\Cb_0$ which intersect at a two-sphere $S_{0,0}$. Rendall \cite{Ren} has showed that for smooth initial data given on $C_0$ and $\Cb_0$, the solution of vacuum Einstein equations exists in a neighbourhood of their intersection $S_{0,0}$ to the future of two null cones. Luk \cite{Luk} extended this result. He showed that if the initial data are given on $C_0$ for $0\le\ub\le I_1$ and $\Cb_0$ for $0\le u\le I_2$, the solution in fact exists in a full neighbourhood of two initial null cones to the future, i.e., $(\ub,u)\in\left([0,I_1]\times[0,\varepsilon]\right)\bigcup\left([0,\varepsilon]\times[0,I_2]\right)$, where $\varepsilon>0$ is a small parameter depending on the size of the initial data. One can directly apply this result to the case when $C_0$ is complete. The solution then exists on a neighbourhood of $C_0$ but the ``thickness'' of the solution is not uniform. Caciotta and Nicol\`o \cite{C-N05} \cite{C-N10} also consider the initial outgoing null cones to be complete towords to the future. The data they prescribe are small and suitably decay. In this case, the solution exists globally to the whole future of the two initial null cones, say $(\ub,u)\in[0,+\infty)\times[0,I]$ where the constant $I$ is the affine length of the null generators of $\Cb_0$. 

There are some other works using characteristic setting to capture interesting mathematical and physical phenomena, by specifying some carefully designed initial data, such as in the recent breakthrough on the formation of black holes \cite{Chr}\cite{K-R-09}, and on the impulsive gravitational waves in \cite{L-R1}\cite{L-R2}.

The main result of this paper is the following existence theorem (in a rough form).
\begin{theorem}Given general initial data on $C_0$ and $\Cb_0$, which are two null cone intersecting at a sphere $S_{0,0}$, assume that the outgoing null cone $C_0$ is complete towards to the future and the data on $C_0$ suitably decay (but not necessarily small). Then the solution of vacuum Einstein equations exists in a uniform neighbourhood of $C_0$, that is, the solution contains a family of outgoing complete null cone such that we can define a small segment of future null infinity.
\end{theorem} 

The precise form of the theorem is in Theorem \ref{maintheorem} in Section \ref{maintheoremsection}. One direct corollary is 
\begin{corollary}
Consider the characteristic-Cauchy mixed initial data problem of the vaccum Einstein equations, where the initial data are given on a three-dimensional disk and the complete outgoing null cone rooted at the boundary of the disk. Then for general such initial data with suitable decay along the null cone, the solution exists in a uniform neighbourhood of the initial hypersurfaces.
\end{corollary}

As mentioned above, this corollary serves as the first step towards to the weak cosmic censorship. Physically, it says that the solution exists locally in retarted time, which can be viewed as the affine parameter of the null generators of the future null infinity. We also remark that the decaying condition we impose inherits from the work of Christodoulou-Klainerman \cite{Ch-K} and Klainerman-Nicol\`o \cite{K-N}.

\subsection{Comments on the Proof} 

The technique of energy estimates for curvature was developed by Christodoulou and Klainerman in \cite{Ch-K}, using the Bel-Ronbinson tensor for the Weyl curvature of the space-time based on the second Bianchi identities. This technique was also employed in \cite{K-N} and in the recent breakthrough on the formation of black holes by Christodoulou \cite{Chr} and soon an improvement by Klainerman and Rodnianski \cite{K-R-09}. In \cite{Ch-K}, \cite{K-N} and \cite{Chr}, this technique was applied in an infinite region which extends to (future or past) null infinity. They used some specific approximately Killing and conformal Killing vectorfields as multipliers and commutators in order to capture the decay of different geometric quantities towards to the null infinity. The more recent works such as \cite{Luk}, \cite{L-Y}, \cite{L-R1}, \cite{L-R2}, \cite{D-H-R} used instead the null Bianchi equations which are written in components and integration by parts. When dealing with the null infinity, such as in \cite{L-Y} and \cite{D-H-R} and the current work, since the decay of the geometry is well understood after the work of \cite{Ch-K} and \cite{K-N}, we can use suitable weights to generate the weighted energy estimate using integration by parts. The advantage of not using Bel-Robinson tensor is that one do not need to use specific vectorfields to be multipliers and commutators. The existence of a preferable vectorfield is not ensured in general space-times. One another advantage is that, it is much easier to do renormalization in order to compare two different space-time or explorer some hidden cancellations of the equations as in \cite{L-R1} and \cite{L-R2}.

In our proof, the space-time is foliated by two optical functions $\ub$ and $u$ which have null gradiants. Such a foliation is usually called a double null foliation. Almost all works mentioned above are in this framework. The double null foliation is particular suitable for the problems dealing with null infinity, since the (future or past) null infinity can be viewed as the limit of the (incoming or outgoing) null cones. For example, in \cite{Chr}, the initial data are given on a family of outgoing null cones which are tending to the past null infinity, and in \cite{K-N}, the solutions are solved up to the future null infinity. 

When constructing the solution up to the future null infinity, the construction of the global optical function which represents the retarded time should be done using a last slice argument, e.g. in \cite{Ch-K}, \cite{K-N}, \cite{C-N10}. The level sets of the global optical function is not constructed by simply extend the outgoing null cones rooted on the initial hypersurface (space-like or null). The level sets of the global optical function should be constructed initiated from the sections of the future null infinity, and the sections satisfy a specific differential equation. In practice, this is done by a bootstrap and limiting argument. In the current work, though the geometric quantities are not small, since we only construct the solution in a small retarded time, it turns out the canonical double null foliation first used in \cite{K-N} also works in our case.

One another ingrediant is that some special \textit{reductive} structures in Einstein equations are used to ellimilate the high nonlinearity of the equations. The work \textit{reductive} has different meanings in different works, but in general, it means that for some local or semi-global problem, when no uniform  smallness is imposed, we can always adjust the small parameter, such that the highly coupled nonlinear terms are absorbed by the small parameter. The coupled nature is completely destroyed. An important example to this philosophy is \cite{Chr}. Even though the data prescribed is protentially very large for some components, but it turns out that these components appear together with some other small components. The small parameter is still able to absorb the nonlinearity. Other examples are in \cite{Luk}, \cite{L-R1}, \cite{L-R2}. In these cases, the data prescribed are bounded (in suitable norms) and the small parameter comes from the smallness of the existence region, which is the same in our current work. We also remark that the reductive struture is related to the null structure of Einstein equations, which corresponds to the null condition in nonlinear wave equations and ensures the decay estimates are strong enough to make the space-time integral converges. The two structures are more or less the same thing in some cases. In the current work, they are considered simultaneously.

We make some comments on some relaxation of the decaying condition. Firstly, notice that the work of \cite{Chr} serves as a good example of our theorem. If we in addition assume the initial data given in $\ub\in[0,\delta)$ is compactly supported and extend the initial data trivially (that is $\chih\equiv0$), then we obtain the initial data given on a  complete null cone $C_{u_0}$. Notice that the incoming null cone $\Cb_\delta$ serves as the initial incoming cone in the new problem. According to \cite{Chr}, if we do not appeal to a higher order estimate, the curvature component $\beta$ is of size $\delta^{-1}$ on $\Cb_{\delta}$. Consider now the characteristic initial data problem on $C_{u_0}$ and $\Cb_\delta$, one should ask whether we can solve the vacuum Einstein equations independent of $\delta$.

At this point, we find the renormalization in \cite{L-R1} is quite useful. They define two new quantities
\begin{align*}
\check{\rho}=\rho-\frac{1}{2}(\chih,\chibh),\quad \check{\sigma}=\sigma+\frac{1}{2}\chih\wedge\chibh.
\end{align*}
This renormalization ellimilates the component $\alpha$ in the null Bianchi equtions $D\rho$, $D\sigma$, and then we can do energy estimate without knowing information about $\alpha$. We find that this renormalization also works when considering such a semi-global problem. In usual energy estimate, the estimate of $\alpha$ on outgoing null cone appears together with the estimate of $\beta$ on incoming null cone. Although $\beta$ is of size $\delta^{-1}$ on $\Cb_{\delta}$ in \cite{Chr}\footnote{This is not directly pointed out in  \cite{Chr} because only the estimate of $\beta$ on outgoing null cone is used and of size $\delta^{-1/2}$.}, we find that no other components on $\Cb_{\delta}$ are of size $\delta$ to some negative power. Therefore we can solve the vacuum Einstein equations independent of $\delta$. In addition, we find that our main theorem may indicate that the decaying condition of the component $\alpha$ seems not to be quite relevant to the weak cosmic censorship.

\section{Preliminary}

\subsection{Basic Geometric Setup} We follow the geometric setup and notations in \cite{Chr}.  We use $M$
 %= M(u_0, \delta + 1)$ ($u_0 \leq -2$ is a fixed negative number and $\delta$ is a small positive number which will be determined later)
 to denote the underlying space-time (which will be the solution) and use $g$ to denote the background 3+1 dimensional Lorentzian metric. We use $\nabla$ to denote the Levi-Civita connection of the metric $g$.

Let $\ub$ and $u$ be two optical functions on $M$, that is
\begin{equation*}
g(\nabla\ub,\nabla \ub)= g(\nabla u,\nabla u)=0.
\end{equation*}
The space-time $M$ is foliated by the level sets of $\ub$ and $u$ respectively. Since the gradients of $u$ and $\ub$ are null, we call the the these two foliations together a double null foliation. We require the functions $u$ and $\ub$ increase towards the future. We use $C_u$ to denote the outgoing null hypersurfaces which are the level sets of $u$ and use ${\Cb}_{\ub}$ to denote the incoming null hypersurfaces which are the level sets of $\ub$. We denote the intersection $S_{\ub,u}=\Cb_{\ub} \cap C_u$, which is a  space-like two-sphere.

We define a positive function $\Omega$ by the formula $ \Omega^{-2}=-2g(\nabla\ub,\nabla u)$.  We then define the normalized null pair $(e_3, e_4)$ by $e_3=-2\Omega\nabla\ub$ and $e_4=-2\Omega\nabla u$, and define one another null pair $\Lb=\Omega e_3$ and $L=\Omega e_4$. We remark that the flows generated by $\Lb$ and $L$ preserve the double null foliation. On a given two sphere $S_{\ub, u}$ we choose a local orthonormal frame $(e_1,e_2)$. We call $(e_1, e_2, e_3, e_4)$ a \emph{null frame}.  As a convention, throughout the paper, we use capital  Latin letters $A, B, C, \cdots$ to denote an index from $1$ to $2$, e.g. $e_A$ denotes either $e_1$ or $e_2$. 

We define $\phi$ to be a tangential tensorfield if $\phi$ is \textit{a priori} a tensorfield defined on the space-time $M$ and all the possible contractions of $\phi$ with either $e_3$ or $e_4$ are zeros. We use $D\phi$ and $\Db\phi$ to denote the projection to $S_{\ub,u}$ of usual Lie derivatives $\mathcal{L}_L\phi$ and $\mathcal{L}_{\Lb}\phi$. The space-time metric $g$ induces a Riemannian metric $\gs$ on $S_{\ub,u}$ and $\epsilons$ is the volume form of $\gs$ on $S_{\ub,u}$. We use $\ds$ and $\nablas$ to denote the exterior differential and covariant derivative (with respect to $\gs$) on $S_{\ub,u}$.

\begin{comment}Let $(\theta^A)_{A=1,2}$ be a local coordinate system on the two sphere $S_{0,0}$. We can extend $\theta^A$'s  to the entire $M$ by first setting $L(\theta^A)=0$ on $C_{0}$, and then setting $\Lb(\theta^A)=0$ on $M$. Therefore, we obtain a local coordinate system $(\ub,u,\theta^A)$ on $M$. In such a coordinate system, the Lorentzian metric $g$ takes the following form
\begin{align*}
g=-2\Omega^2(\D\ub\otimes\D u+\D u\otimes\D \ub)+\gs_{AB}(\D\theta^A-b^A\D\ub)\otimes(\D\theta^B-b^B\D\ub).
\end{align*}
The null vectors $\Lb$ and $L$ can be computed as $\Lb=\partial_u$ and $L=\partial_{\ub}+b^A\partial_{\theta^A}$. By construction, we have $b^A(\ub,u_0,\theta)=0$. 
\end{comment}

We recall the definitions of null connection coefficients. Roughly speaking, the following quantities are Christoffel symbols of $\nabla$ according to the null frame $(e_1,e_2,e_3,e_4)$:
\begin{align*}
\chi_{AB}&=g(\nabla_Ae_4,e_B),\quad \eta_A=-\frac{1}{2}g(\nabla_3e_A,e_4),\quad \omega=\frac{1}{2}\Omega g(\nabla_4e_3,e_4),\\
\chib_{AB}&=g(\nabla_Ae_3,e_B), \quad\etab_A=-\frac{1}{2}g(\nabla_4e_A,e_3), \quad\omegab=\frac{1}{2}\Omega g(\nabla_3e_4,e_3).
\end{align*}
They are all tangential tensorfields. We also define $\chi'=\Omega^{-1}\chi$, $\chib'=\Omega^{-1}\chi$ and $\zeta=\frac{1}{2}(\eta-\etab)$. The trace of $\chi$ and $\chib$ will play an important role in Einstein field equations and they are defined by $\tr\chi = \gs^{AB}\chi_{AB}$ and $\tr\chib = \gs^{AB}\chib_{AB}$. By definition, we can check directly the following identities $\ds\log\Omega=\frac{1}{2}(\eta+\etab)$, $D\log\Omega=\omega$, $\Db\log\Omega=\omegab$.

We can also define the null components of the curvature tensor
{\bf R}:
\begin{align*}
\alpha_{AB}&=\mathbf{R}(e_A,e_4,e_B,e_4),\quad\beta_A=\frac{1}{2}\mathbf{R}(e_A,e_4,e_3,e_4),\quad\rho=\frac{1}{4}\mathbf{R}(e_3,e_4,e_3,e_4),\\
\alphab_{AB}&=\mathbf{R}(e_A,e_3,e_B,e_3),\quad\betab_A=\frac{1}{2}\mathbf{R}(e_A,e_3,e_3,e_4),\quad\sigma=\frac{1}{4}\mathbf{R}(e_3,e_4,e_A,e_B)\epsilons^{AB}.
\end{align*}

We then define several kinds of contraction of the tangential tensorfields, which are used in deriving the equations. For a  symmetric tangential 2-tensorfield $\theta$, we use $\widehat{\theta}$ and $\tr\theta$ to denote the trace-free part and trace of $\theta$ (with respect to $\gs$). If $\theta$ is trace-free, $\Dh\theta$ and $\Dbh\theta$ refer to the trace-free part of $D\theta$ and $\Db\theta$. Let $\xi$ be a tangential $1$-form. We define some products and operators for later use. For the products, we define $(\theta_1,\theta_2)=\gs^{AC}\gs^{BD}(\theta_1)_{AB}(\theta_2)_{CD}$ and $\ (\xi_1,\xi_2)=\gs^{AB}(\xi_1)_A(\xi_2)_B$. This also leads to the following norms $|\theta|^2=(\theta,\theta)$ and $|\xi|^2=(\xi,\xi)$. We then define the contractions $(\theta\cdot\xi)_A=\theta_A{}^B\xi_B$, $(\theta_1\cdot \theta_2)_{AB}=(\theta_1)_A{}^C(\theta_2)_{CB}$, $\theta_1 \wedge\theta_2=\epsilons^{AC}\gs^{BD} (\theta_1)_{AB}(\theta_2)_{CD}$ and $\xi_1\tensor \xi_2=\xi_1\otimes\xi_2+\xi_2\otimes\xi_1-(\xi_1,\xi_2)\gs$. The Hodge dual for $\xi$ is defined by $\prescript{*}{}\xi_A=\epsilons_A{}^C\xi_C$. For the operators, we define $\divs\xi=\nablas^A\xi_A$, $\curls\xi_A=\epsilons^{AB}\nablas_A\xi_B$ and $(\divs\theta)_A=\nablas^B\theta_{AB}$. We finally define a traceless operator $(\nablas\tensor\xi)_{AB}=(\nablas\xi)_{AB}+(\nablas\xi)_{BA}-\divs\xi \,\gs_{AB}$.

%For the sake of simplicity, we will use abbreviations $\Gamma$ and $R$ to denote an arbitrary connection coefficient and an arbitrary null curvature component. We introduce a schematic way to write products. Let $\phi$ and $\psi$ be arbitrary tangential tensorfields, we also use $\phi\cdot\psi$ to denote an arbitrary contraction of $\phi$ and $\psi$ by $\gs$ and $\epsilons$. This schematic notation only captures the quadratic nature of the product, and it will be good enough for most of the cases when we derive estimates. 

\subsection{Equations}

The following is the first structure equations in the space-time written in a null frame (where  $K$ is the Gauss curvature of $S_{\ub,u}$):\footnote{See Chapter 1 of \cite{Chr} for the derivation of these equations.}
\begin{align*}
\Dh \chih'=-\alpha,&\quad
D\tr\chi'=-\frac{1}{2}\Omega^2(\tr\chi')^2-\Omega^2|\chih'|^2,\\
\Dbh \chibh'=-\alphab,&\quad
\Db\tr\chib'=-\frac{1}{2}\Omega^2(\tr\chib')^2-\Omega^2|\chibh'|^2,\\
D\eta &= \Omega(\chi \cdot\etab-\beta),\\
\Db\etab &= \Omega(\chib \cdot\eta+\betab),\\
D  \omegab &=\Omega^2(2(\eta,\etab)-|\eta|^2-\rho),\\
\Db  \omega &=\Omega^2(2(\eta,\etab)-|\etab|^2-\rho),\\
K&=-\frac{1}{4}\tr \chi\tr\chib+\frac{1}{2}(\chih,\chibh)-\rho,\\
\divs \chih'&=\frac{1}{2}\ds \tr \chi'-\chih'\cdot\eta+\frac{1}{2}\tr \chi'\eta-\Omega^{-1}\beta,\\
\divs \chibh'&=\frac{1}{2}\ds \tr \chib'-\chibh'\cdot\etab+\frac{1}{2}\tr \chib'\etab-\Omega^{-1}\betab,\\
\curls \eta&=\sigma-\frac{1}{2}\chih \wedge\chibh,\\
\curls \etab&= -\sigma + \frac{1}{2}\chih \wedge\chibh,\\
\Dh(\Omega\chibh)&=\Omega^2(\nablas \tensor \etab + \etab \tensor \etab +\frac{1}{2}\tr\chi\chibh-\frac{1}{2}\tr\chib \chih),\\
D(\Omega\tr\chib)&=\Omega^2(2\divs\etab+2|\etab|^2-(\chih,\chibh)-\frac{1}{2}\tr\chi\tr\chib+2\rho),\\
\Dbh(\Omega\chih)&=\Omega^2(\nablas \tensor \eta + \eta \tensor \eta +\frac{1}{2}\tr\chib\chih-\frac{1}{2}\tr\chi \chibh),\\
\Db(\Omega\tr\chi)&=\Omega^2(2\divs\eta+2|\eta|^2-(\chih,\chibh)-\frac{1}{2}\tr\chi\tr\chib+2\rho),\\
\Db\eta&=-\Omega(\chib\cdot\eta+\betab)+2\ds\omegab,\\
D\etab&=-\Omega(\chi\cdot\etab-\beta)+2\ds\omega.
\end{align*}

We also use the null frame to decompose the contracted second Bianchi identity $\nabla^\alpha \mathbf{R}_{\alpha\beta\gamma\delta} = 0$ into components. This leads the following null Bianchi equations:\footnote{See Proposition 1.2 of \cite{Chr}.}
\begin{gather*}
\Dbh\alpha-\frac{1}{2}\Omega\tr\chib \alpha+2\omegab\alpha+\Omega\{-\nablas\tensor\beta -(4\eta+\zeta)\tensor \beta+3\chih \rho+3{}^*\chih \sigma\}=0,\\
\Dh\alphab-\frac{1}{2}\Omega\tr\chi \alphab+2\omega\alphab+\Omega\{\nablas\tensor\betab +(4\etab-\zeta)\tensor \betab+3\chibh \rho-3{}^*\chibh \sigma\}=0,\\
D\beta+\frac{3}{2}\Omega\tr\chi\beta-\Omega\chih\cdot\beta-\omega\beta-\Omega\{\divs\alpha+(\etab+2\zeta)\cdot\alpha\}=0,\\
\Db\betab+\frac{3}{2}\Omega\tr\chib\betab-\Omega\chibh\cdot\betab-\omegab\betab+\Omega\{\divs\alphab+(\eta-2\zeta)\cdot\alphab\}=0,\\
\Db\beta+\frac{1}{2}\Omega\tr\chib\beta-\Omega\chibh \cdot \beta+\omegab \beta-\Omega\{\ds \rho+{}^*\ds \sigma+3\eta\rho+3{}^*\eta\sigma+2\chih\cdot\betab\}=0,\\
D\betab+\frac{1}{2}\Omega\tr\chi\betab-\Omega\chih \cdot \betab+\omega \betab+\Omega\{\ds \rho-{}^*\ds \sigma+3\etab\rho-3{}^*\etab\sigma-2\chibh\cdot\beta\}=0,\\
D\rho+\frac{3}{2}\Omega\tr\chi \rho-\Omega\{\divs \beta+(2\etab+\zeta,\beta)-\frac{1}{2}(\chibh,\alpha)\}=0,\\
\Db\rho+\frac{3}{2}\Omega\tr\chib \rho+\Omega\{\divs \betab+(2\eta-\zeta,\betab)+\frac{1}{2}(\chih,\alphab)\}=0,\\
D\sigma+\frac{3}{2}\Omega\tr\chi\sigma+\Omega\{\curls\beta+(2\etab+\zeta,{}^*\beta)-\frac{1}{2}\chibh\wedge\alpha\}=0,\\
\Db\sigma+\frac{3}{2}\Omega\tr\chib\sigma+\Omega\{\curls\betab+(2\etab-\zeta,{}^*\betab)+\frac{1}{2}\chih\wedge\alphab\}=0.
\end{gather*}

\subsection{Hodge systems and Commutation Formulas}

We will introduce the Hodge systems satisfied by the connection coefficients. We first define $\mu,\mub$ to be
\begin{align*}
\mu=K+\frac{1}{4}\tr\chi\tr\chib-\divs\eta,\\
\mub=K+\frac{1}{4}\tr\chi\tr\chib-\divs\etab.
\end{align*}
And we also define $\kappa,\kappab$ to be
\begin{align*}
\kappa=\Deltas\omega+\divs(\Omega\beta),\\
\kappab=\Deltas\omegab -\divs(\Omega\betab).
\end{align*}

Then we have\footnote{See Chapter 6 of \cite{Chr}.}
\begin{lemma}\label{hodgesystems}
$(\Omega\chih,\Omega\tr\chi)$ satisfies
\begin{align*}
D(\Omega\tr\chi)&=-\frac{1}{2}(\Omega\tr\chi)^2-|\Omega\chih|^2+2\omega(\Omega\tr\chi),\\
\divs (\Omega\chih)&=\frac{1}{2}\ds (\Omega\tr \chi)-\Omega\chih\cdot\eta+\frac{1}{2}\Omega\tr \chi\eta-\Omega\beta.
\end{align*}

$(\Omega\chibh,\Omega\tr\chib)$ satisfies
\begin{align*}
\Db(\Omega\tr\chib)&=-\frac{1}{2}(\Omega\tr\chib)^2-|\Omega\chibh|^2+2\omega(\Omega\tr\chib),\\
\divs (\Omega\chibh)&=\frac{1}{2}\ds (\Omega\tr \chib)-\Omega\chibh\cdot\etab+\frac{1}{2}\Omega\tr \chib\etab+\Omega\betab.
\end{align*}

$(\eta,\mu)$ satisfies
\begin{align*}
%\label{NSE_div_curl_eta}
&\begin{cases}\divs\eta=-\rho+\frac{1}{2}(\chih,\chibh)-\mu,\\
\curls\eta=\sigma-\frac{1}{2}\chih\wedge\chibh,\end{cases}\\
%\label{NSE_D_mu}
D\mu=&-\Omega\tr\chi\mu-\frac{1}{2}\Omega\tr\chi\mub-\frac{1}{4}\Omega\tr\chib|\chih|^2+\frac{1}{2}\Omega\tr\chi|\etab|^2\\&+\divs(2\Omega\chih\cdot\eta-\Omega\tr\chi\etab).
\end{align*}

$(\etab,\mub)$ satisfies
\begin{align*}
%\label{NSE_div_curl_etab}
&\begin{cases}\divs\etab=-\rho+\frac{1}{2}(\chih,\chibh)-\mub,\\
 \curls\etab=-\sigma+\frac{1}{2}\chih\wedge\chibh,\end{cases}\\
%& \label{NSE_Db_mub}
\Db\mub=&-\Omega\tr\chib\mub-\frac{1}{2}\Omega\tr\chib\mu-\frac{1}{4}\Omega\tr\chi|\chibh|^2+\frac{1}{2}\Omega\tr\chib|\eta|^2\\&+\divs(2\Omega\chibh\cdot\etab-\Omega\tr\chib\eta).
\end{align*}

$(\omega,\kappa)$ satisfies
\begin{align*}
%\label{NSE_Deltas_omega}
\Deltas\omega&=\kappa-\divs(\Omega\beta),\\
%\label{NSE_Db_kappa}
\Db\kappa+\Omega\tr\chib\kappa&=-2(\Omega\chibh,\nablas^2\omega)+m,
\end{align*}
where
\begin{align*}
m=&-2(\divs(\Omega\chibh),\ds\omega)+\frac{1}{2}\divs(\Omega\tr\chib\cdot\Omega\beta)-(\ds(\Omega^2),\ds\rho)+(\ds(\Omega^2),{}^*\ds\sigma)-\rho\Deltas(\Omega^2)\\
&+\Deltas(\Omega^2(2(\eta,\etab)-|\etab|^2))+\divs(\Omega^2(-\chibh\cdot\beta+2\chih\cdot\betab+3\eta\rho+3{}^*\eta\sigma)).
\end{align*}

$(\omegab,\kappab)$ satisfies
\begin{align*}
%\label{NSE_Delta_omegab}
\Deltas\omegab& =\kappab+\divs(\Omega\betab),\\
%\label{NSE_D_kappab}
D\kappab+\Omega\tr\chi\kappab&=-2(\Omega\chih,\nablas^2\omegab)+\underline{m},
\end{align*}
where
\begin{align*}
\underline{m}=&-2(\divs(\Omega\chih),\ds\omegab)+\frac{1}{2}\divs(\Omega\tr\chi\cdot\Omega\betab)-(\ds(\Omega^2),\ds\rho)-(\ds(\Omega^2),{}^*\ds\sigma)-\rho\Deltas(\Omega^2)\\
&+\Deltas(\Omega^2(2(\eta,\etab)-|\eta|^2))+\divs(\Omega^2(\chih\cdot\betab-2\chibh\cdot\beta+3\etab\rho-3{}^*\etab\sigma)),
\end{align*}
\end{lemma}

We denote the first order elliptic operators (or Hodge operators) $\mathcal{D}_1,\mathcal{D}_2$, by\footnote{See \cite{K-R-09}.}
\begin{align*}
&\mathcal{D}_1:\text{tangential one-form $\xi\mapsto$ a pair of functions $(\divs\xi,\curls\xi)$};\\
&\mathcal{D}_2:\text{tangential symmetric trace-free $(0,2)$ type tensorfield $\theta\mapsto$ tangential one-form $\divs\theta$}.
\end{align*}
It is easy to calculate the formal $L^2$ adjoint
\begin{align*}
&^*\mathcal{D}_1:\text{a pair of functions $(f,g)\mapsto$ tangential one-form $-\ds f+{}^*\ds g$};\\
&^*\mathcal{D}_2:\text{tangential one-form $\xi\mapsto$ tangential symmetric trace-free $(0,2)$ type tensorfield $-\frac{1}{2}\nablas\tensor\xi$}.
\end{align*}
We will denote any one of the above elliptic operators (or Hodge operators) and their formal $L^2$ adjoint by $\mathcal{D},{}^*\mathcal{D}$.

We also need the following commutation formula for the estimate of derivatives\footnote{See Chapter 4 of \cite{Chr} for the first group. The second group can be derived directly by the definition of curvature.}.
\begin{lemma}\label{commutator}
Given integer $i$ and tangential tensorfield $\phi$. we have
\begin{align*}
[D,\nablas^i]\phi&=\sum_{j=1}^i\nablas^j(\Omega\chi)\cdot\nablas^{i-j}\phi,\\
[\Db,\nablas^i]\phi&=\sum_{j=1}^i\nablas^j(\Omega\chib)\cdot\nablas^{i-j}\phi,
\end{align*}
and
\begin{align*}
[\mathcal{D},\nablas^i]\phi&=\sum_{j=1}^i\nablas^{j-1}K\cdot\nablas^{i-j}\phi,\\
[^*\mathcal{D},\nablas^i]\phi&=\sum_{j=1}^i\nablas^{j-1}K\cdot\nablas^{i-j}\phi.
\end{align*}
Here we use ``$\cdot$'' to represent an arbitrary contraction with the coefficients by $\gs$ or $\epsilons$. In addition, if $\phi$ is a function, then when $i=1$, all commutators above are zero; when $i\ge2$, all $i$'s are replaced by $i-1$'s in above formulas.
\end{lemma}

\subsection{Basic Inequalities}\label{basicinequalities}

We will introduce some frequently used basic inequalities,, such as Sobolev inequalities, Poincar\'e inequalities, Gronwall type inequalities, and estimates for Hodge systems.

For Sobolev inequalities, we start from isoperimetric inequality: Given a function $f\in W^{1,1}(S_{\ub,u})$ and denoting by $\bar{f}$ the average of $f$ on $S_{\ub,u}$, we have
\begin{align*}
\int_{S_{\ub,u}}(f-\bar{f})^2\D\mu_{\gs}\le I(S_{\ub,u})\left(\int_{S_{\ub,u}}|\ds f|\D\mu_{\gs}\right)^2.
\end{align*}
where $I(S_{\ub,u})$ is the isoperimetric constant. Based on the isoperimetric inequality, we have

\begin{lemma}[Sobolev inequalities, see Section 5.2 of \cite{Chr}]\label{Sobolev}
Given a tangential tensorfield $\phi$, we have for $q\in(2,+\infty)$,
\begin{align*}
\|\phi\|_{L^q(S_{\ub,u})}\le C\sqrt{\max\{I(S_{\ub,u}),1\}}\sum_{i=0}^1 r^{-1+1/q}\|(r\nablas)^i\phi\|_{L^2(S_{\ub,u})},\\
\|\phi\|_{L^\infty(S_{\ub,u})}\le C\sqrt{\max\{I(S_{\ub,u}),1\}}\sum_{i=0}^2 r^{-1}\|(r\nablas)^i\phi\|_{L^2(S_{\ub,u})},
\end{align*}
where $r=r(\ub,u)$ satisfies $4\pi r^2=Area(S_{\ub,u})$ and $C$ is a universal constant.
\end{lemma}

Applying H\"older inequality on the right hand side of the isoperimetric inequality, we have:
\begin{lemma}[Poincar\'e inequality]\label{poincare}
Given a function $f$, we have
\begin{align*}
\|f-\bar{f}\|_{L^2(S_{\ub,u})}\le C\sqrt{I(S_{\ub,u})}\|(r\nablas)f\|_{L^2(S_{\ub,u})}.
\end{align*}
\end{lemma}

We also need the following Gronwall type estimates:
\begin{lemma}[Gronwall type estimates, see Chapter 4 of \cite{Chr} or Chapter 4 of \cite{K-N}]\label{evolution}
Assume that on the outgoing null cone $C_u$, the parameter $\ub\in[\ub_0,\ub_1]\subset[0,+\infty)$, $C^{-1}(\ub+1)\le r\le C(\ub+1)$ for some universal $C$, and
\begin{align*}
C^{-1}r^{-3/2}\le\Omega|\chih|+\left|\Omega\tr\chi-\overline{\Omega\tr\chi}\right|\le Cr^{-3/2},
\end{align*}
Then, for a~$(0,s)$ type tengential tensorfield $\phi$, $2\le q\le+\infty$, and any real $\nu$, we have
\begin{align*}
&\|r^{s-\nu-2/q}\phi\|_{L^q(S_{\ub,u})}\\
\le& C_{q,\nu,s}\left(\|r^{s-\nu-2/q}\phi\|_{L^q(S_{\ub_0,u})}+\int_{\ub_0}^{\ub}\|r^{s-\nu-2/q}(D\phi-\frac{\nu}{2}\Omega\tr\chi\phi)\|_{L^q(S_{\ub',u})}\D\ub'\right).
\end{align*}
It is also true if we replace the first term of the right hand side by taking value on $S_{\ub_1,u}$ and the integral of the second term by integrating from $\ub_1$ to $\ub$.

Assume that on the incoming null cone $\Cb_{\ub}$ with the parameter $u\in[0,\varepsilon]$, we have
\begin{align*}
C^{-1}\le\Omega|\chibh|+\left|\Omega\tr\chib\right|\le C.
\end{align*}
Then for an arbitrary tangential tensorfield $\phi$, $2\le q\le+\infty$, we have
\begin{align*}
\|\phi\|_{L^q(S_{\ub,u})}
\le C_{q}\left(\|\phi\|_{L^q(S_{\ub,0})}+\int_{0}^{u}\|\Db\phi\|_{L^q(S_{\ub,u'})}\D u'\right).
\end{align*}
\end{lemma}

Finally, we need the following elliptic estimates:
\begin{lemma}[Elliptic estimates for Hodge systems, see Chapter 7 of \cite{Chr}]\label{elliptic}Assume that $\theta$ is a tangential symmetric trace-free $(0,2)$ type tensorfield with
$$\divs\theta=f,$$
where $f$ is a tangential one-form. Then for $i\ge1$, we have
\begin{align*}
\|(r\nablas)^{i}\theta\|_{L^2(S_{\ub,u})}\le C_K\left(\sum_{j=0}^{i-1}\|(r\nablas)^j(rf)\|_{L^2(S_{\ub,u})}+\|\theta\|_{L^2(S_{\ub,u})}\right).
\end{align*}
Given a tangential one-form $\xi$ with
\begin{equation*}
\divs\xi=f, \quad \curls\xi=g,
\end{equation*}
we have for $i\ge1$, 
\begin{align*}
\|(r\nablas)^i\xi\|_{L^2(S_{\ub,u})}\le C_K\left(\sum_{j=0}^{i-1}\|(r\nablas)^j(rf)\|_{L^2(S_{\ub,u})}+\|(r\nablas)^j(rg)\|_{L^2(S_{\ub,u})}+\|\xi\|_{L^2(S_{\ub,u})}\right).
\end{align*}
Here $C_K$ depends on $i$, $r\|(r\nablas)^{\le \max\{i-2,0\}} K\|_{L^2(S_{\ub,u})}$ and the  Sobolev constant (which depends on the isoperimetric constant $I(S_{\ub,u})$). 
\end{lemma}

\section{Main Theorem and Structure of the Proof}

\subsection{Definition of Various Quantities and the Statement of Main Theorem}\label{maintheoremsection}

We first define several norms of the connection coefficients and curvature components adapted to our problem. In many cases, we use $\Gamma$ to denote one of the connection coefficients, say $\chih$, $\tr\chi$, $\chibh$, $\tr\chib$, $\eta$, $\etab$, $\omega$ and $\omegab$ up to a multiple by $\Omega$. In fact, $\Gamma$ will share all estimates with $\Omega\Gamma$ because $\Omega$ will be close to $1$ up to derivatives. We also use $\Gamma$ to denote $\Omega\tr\chi-\overline{\Omega\tr\chi}$ in some cases. We emphasize that $\Omega\tr\chi$ and $\Omega\tr\chi-\overline{\Omega\tr\chi}$ should be treated in a completely different way. We also use $R$ to denote one of the curvature components of $\beta$, $\rho$, $\sigma$, $\betab$ and $\underline{R}$ to denote one of the curvature components of $\rho$, $\sigma$, $\betab$ and $\alphab$. For the same component, that we use the underline or not depends on different roles the component playing in the equations. We assign a number $p$ to every connection and curvature components to denote their expected decay rate about $r$:
\begin{equation}\label{definitionp}
\begin{split}
&p(\chih, \Omega\tr\chi-\overline{\Omega\tr\chi}, \eta, \etab)=2,\ p(\chibh, \tr\chi, \tr\chib, \omegab)=1,\ p(\omega)=3,\\
&p(\beta)=4,\ p(\rho,\sigma)=3,\ p(\betab)=2,\ p(\alphab)=1.
\end{split}
\end{equation}
We will use $\Gamma_p$, $R_p$ and $\underline{R}_p$ to denote one of the quantities with  $p$ as the number assigned. The above definitions also valid up to a multimle by $\Omega$.  

Let $r$ be the area radius defined by $4\pi r(\ub,u)^2=Area(S_{\ub,u})$. For the null connection coefficients, we denote
\begin{align*}
\mathcal{O}^{0,\infty}[\Gamma_p]&=\sup_{\ub,u}\|r^p\Gamma_p\|_{L^\infty(S_{\ub,u})},\\
\mathcal{O}^{i,4}[\Gamma_p]&=\sup_{\ub,u}\|r^{p-1/2}(r\nablas)^i\Gamma_p\|_{L^4(S_{\ub,u})} \text{ for } i\le1,\\
\mathcal{O}^{i,2}[\Gamma_p]&=\sup_{\ub,u}\|r^{p-1}(r\nablas)^i\Gamma_p\|_{L^2(S_{\ub,u})} \text{ for } i\le2.
\end{align*}
We denote $\mathcal{O}[\Gamma]$ be the sum over all norms of $\Gamma$. We also use $\mathcal{O}^{0,\infty}$, $\mathcal{O}^{i,4}$, $\mathcal{O}^{i,2}$ and $\mathcal{O}$ to denote the corresponding sum over all $\Gamma$. For curvature components, we denote
\begin{align*}
\mathcal{R}[R_p]=\sup_{u}\sum_{i=0}^2\|r^{p-2}(r\nablas)^iR_p\|_{L^2(C_{u})},\\
\underline{\mathcal{R}}[\underline{R}_p]=\sup_{\ub}\sum_{i=0}^2\|r^{p-1}(r\nablas)^i\underline{R}_p\|_{L^2(\Cb_{\ub})}
\end{align*}
and $\mathcal{R}$, $\underline{\mathcal{R}}$ denote the sum over all components $R$ or $\underline{R}$.

As discussed in the introduction, we consider double characteristic initial data problem of vacuum Einstein equations, where the initial data are given on two null cones $C_0$ and $\Cb_0$ intersecting at a sphere $S_{0,0}$, where $C_0$ is complete towards to the future. Precisely, We will prove the following main theorem in this paper:
\begin{theorem}[Main Theorem]\label{maintheorem}
Suppose that we have two intersection null cones $C_0$ and $\Cb_0$ where $C_0$ is an outgoing null cone extended to infinity and $\Cb_0$ an incoming null cone, and $S_0=\Cb_0\bigcap C_0$ is a two-sphere. Suppose also $C_0$ and $\Cb_0$ are foliated by affine sections which are labelled by two functions $s$ and $\underline{s}$. Let $\Lambda(s)$ and $\lambda(s)$ be the larger and smaller eigenvalue of $r(s)^2\gs|_{S_{s,0}}$ with respect to $r(0)^2\gs|_{S_{0,0}}$, and $r(s)$ be the area radius defined by $4\pi r(s)^2=Area(S_{s,0})$. If the initial data given on $\Cb_0\bigcup C_0$ satisfy the following:
\begin{align*}
&C_{\lambda,\Lambda}^{-1}\le\lambda(s)\le\Lambda(s)\le C_{\lambda,\Lambda},\quad C_r^{-1}(1+s)\le r(s)\le C_r(1+s),\\
\mathcal{O}_0\triangleq&\sup_s\sum_{\Gamma_p\neq\omegab}\left(\|r^p\Gamma_p\|_{L^\infty(S_{s,0})}+\sum_{i=0}^1\|r^{p-1/2}(r\nablas)^i\Gamma_p\|_{L^4(S_{s,0})}+\sum_{i=0}^2\|r^{p-1}(r\nablas)^i\Gamma_p\|_{L^2(S_{s,0})}\right)\\
&+\|(r\nablas)^3(\etab,\chih)\|_{L^2(C_0)}+\sup_s\|(r\nablas)^3\tr\chib\|_{L^2(S_{s,0})}+C_{\lambda,\Lambda}+C_r<\infty\\
\mathcal{R}_0\triangleq&\sum_{R_p}\sum_{i=0}^2\|r^{p-2}(r\nablas)^iR_p\|_{L^2(C_0)}+\|\Db^2\betab\|_{L^2(C_0)}\\&+
\sum_{i=0}^1\left(\|(r\nablas)^i\Db\betab\|_{L^2(C_0)}+\sup_s\|r^{p-1}(r\nablas)^iR_p\|_{L^2(S_{s,0})}\right)\\
&+\sum_{i=0}^1\left(\sup_s\|r^{1/2}\Db\alphab, (r\nablas)^i(r^{1/2}\alphab,r^{3/2}\betab)\|_{L^4(S_{s,0})}\right)<\infty,\\ \underline{\mathcal{R}}_0\triangleq&\sum_{\underline{R}_p}\sum_{i=0}^2\|r^{p-1}(r\nablas)^i\underline{R}_p\|_{L^2(\Cb_0)}+\sum_{i=0}^1\|(r\nablas)^i\Db\alphab\|_{L^2(\Cb_0)}+\|\Db^2\alphab\|_{L^2(\Cb_0)}<\infty.
\end{align*} 
Then there exists an $\varepsilon>0$ depends on $\mathcal{O}_0$, $\mathcal{R}_0$, $\underline{\mathcal{R}}_0$ and a global optical function $u$ such that the solution of vacuum Einstein equations exists in a global double null foliation
$0\le\ub<+\infty$, $0\le u\le\varepsilon$. In addition, $\Omega\to1$ uniformly on every $C_u$ when $\ub\to+\infty$, which means that the function $u$ on $\Cb_{\ub}$ tends to the affine parameter of the null generators of $\Cb_{\ub}$ when $\ub\to+\infty$.
\end{theorem}
\begin{remark}
The last statement of the above theorem shows that the global optical function $u$ can be viewed as the retarted time defined on the future null infinity.
\end{remark}
\begin{remark}
Notice that the initial quantities are written in the sections $s$ and $\underline{s}$. Because we are solving Einstein equations to the future null infinity, the global optical function $u$ is not coincide with the function $\underline{s}$ on the initial incoming null cone $\Cb_0$ in general. In constructing the space-time, the difference between two different foliations on $\Cb_0$ can be controlled. This implies that, although the optical function $\ub$ is constructed such that it coincides with $s$ on $C_0$, but the lapse $\Omega$ may change. In other words, the vectorfield $L$ restricted on $C_0$ and $\Lb'$ restricted on $\Cb_0$ is invariant in our whole construction, but $L'$ restricted on $C_0$ and $\Lb$ restricted on $\Cb_0$ is not.
\end{remark}

\subsection{The Characteristic Initial Data}

Recall from \cite{Ch-K} and \cite{K-N} that for any strongly asymptotically flat Cauchy data, there exists a region $\Omega_0$ with compact closure such that the boundary of the causal future of $\Omega_0$ in the maximal development consists of complete null generators. The boundary of the causal future of $\Omega_0$ satisfies the assumptions on the outgoing null cone of our Theorem \ref{maintheorem}. 

In this subsection, we will also give a brief review on how to specify the characteristic initial data on two intersected null cones and show a trivial extension of the initial data in \cite{Chr} satisfies the assumption of Theorem \ref{maintheorem} and serves as a simple example to our setting. One can see the full details of specifying arbitrary characteristic initial data in for example \cite{Chr}, \cite{Ren}, \cite{Luk}, or \cite{C-N05} for the infinity case.

Geometrically, the initial data on the null cone $C_0$ consist of the conformal geometry of $C_0$, which means that given a family of spherical sections on $C_0$ (usually the affine sections) which are parameterized by a function $s$, one needs to specify a family of metrics $\widehat{\gs}(s)$ on the sections and then the actual metrics on the sections of $C_0$ are given by $\gs(s)=\phi^2(s)\widehat{\gs}(s)$ where the conformal factor $\phi(s)$ is determined by $\widehat{\gs}(s)$, see for example \cite{Chr}. In practice, we will usually impose initial condition on the shear $\chih(s)$, which is the derivative of the conformal geometry, see for example \cite{K-R-09}. Similarly, the initial data on the null cone $\Cb_0$ also consist of the conformal geometry of $\Cb_0$ and we usually specify the shear $\chibh$. To ensure the well-posedness, we also need to specify the ``full geometry'' on the intersection of two null cones $S_{0,0}=C_0\bigcap\Cb_0$. The full geometry consists of the metric $\gs$ induced on $S_{0,0}$ (but not only the conformal metric), the torsion $\zeta$, and both expansions $\tr\chi$ and $\tr\chib$.
\begin{remark}
Specifying initial data in such a way will cause the lost of derivatives. This is the nature of characteristic problem. For example, we need to impose the condition on more than third derivatives of the shear $\chih$ or $\chibh$, which are not maintained along the evolution. The assumptions of Theorem \ref{maintheorem} will maintain along the evolution. 
\end{remark}

Now we turn to \cite{Chr}. Remember the initial data are given on a null cone $C_{u_0}$ from a point $o$ where $u_0<-1$ is a fixed number. Let $\ub$ be the affine parameter on $C_{u_0}$ with its value being $u_0$ at the point $o$. We assume that for $u_0\le\ub\le0$, the initial data given on $C_{u_0}$ is trivial, which means that the geometry is precisely the geometry of a null cone from a point in Minkowski space-time. Let $\delta>0$ be a small parameter, the initial data given in $0\le\ub\le\delta$ are the so-called short pulse data, of which the shear satisfies
$$|\delta^{n}|u_0|^m\nablas^m D^n\chih|\le \delta^{-1/2}|u_0|^{-2}C_{m,n}.$$
The norm is taken by $\gs$ in $L^{\infty}$.

We drop out the weight $u_0$ because we are not considering the problem from past null infinity. Remember the relation derived in Chapter 2 in \cite{Chr} as follows:
\begin{align*}
&\chih\sim\delta^{-1/2}, \tr\chi\sim1,\ \zeta\sim\delta^{1/2},\ \chibh\sim\delta^{1/2},\ \tr\chib\sim1,\\
&\alpha\sim\delta^{-3/2},\ \beta\sim\delta^{-1/2},\ \rho,\sigma\sim1,\ \betab\sim\delta,\ \alphab, \Db\alphab\sim\delta^{3/2}
\end{align*}
The above notation $\psi\sim\delta^r$ means that for $0\le\ub\le\delta$,
$$|\delta^{n}\nablas^m D^n\psi|\le \delta^{r}C_{m,n}.$$
In particular, the above relation holds for $\ub=\delta$. We can also choose $\chih$ suitably such that $\tr\chi>0$ at $\ub=\delta$. We also remember the induced metrics on the spherical sections $\gs$ are expressed in the form
$$\gs(\ub)=\phi^2(\ub)\widehat{\gs}(\ub)$$
where $\widehat{\gs}(\ub)$ has the same volumn form for all $\ub$. We know that for $0\le\ub\le\delta$,
$$\phi-1\sim\delta,\ |u_0|^{-2}\widehat{\gs}-\cir{\gs}\sim\delta^{1/2},\ |u_0|^{-2}\gs-\cir{\gs}\sim\delta^{1/2}$$
where $\cir{\gs}$ is the standard metric on unit sphere. The above relation in particular holds for $\ub=\delta$. 

We may assume that the shear imposed on $0\le\ub<\delta$ has compact support such that we can smoothly extend $\chih$ to $\ub>\delta$ trivially, i.e., $\chih\equiv0$\footnote{One can also do as in the last several pages in \cite{Chr} without the compactly supported condition, but which is essentially the same.}. We then prove the following:
\begin{proposition}
Extending the initial data on $C_{u_0}$ as described above. Then the assumption of Theorem \ref{maintheorem} on $\mathcal{O}_0$, $\mathcal{R}$ holds on the truncated cone $C_{u_0}$ for $\ub\ge\delta$.
\end{proposition}
\begin{remark}
The incoming null cone $\Cb_\delta$ in the work of Christodoulou \cite{Chr} will serve as the incoming null cone $\Cb_0$ in our Theorem \ref{maintheorem}. We do not care what happens for $0\le\ub\le\delta$.
\end{remark}
\begin{proof}
The triviality of $\chih$ imples that $\widehat{\gs}(\ub)\equiv\widehat{\gs}(\delta)$ for $\ub>\delta$ and $\phi$ satisfies $D^2\phi=0$. Recall the relation $D\phi=\phi\tr\chi/2$ and $\tr\chi(\delta)>0$, we can solve $\phi$ as
$$\phi(\ub)=\frac{\phi(\delta)\tr\chi(\delta)}{2}(\ub-\delta)+\phi(\delta)$$
and $\phi>0$ for all $\ub>\delta$. Also, $\tr\chi$ satisfies $D\tr\chi=-\frac{1}{2}(\tr\chi)^2$ and can be solved as
$$\tr\chi(\ub)=\frac{2}{\ub-\delta+\frac{2}{\tr\chi(\delta)}}$$
which is positive for all $\ub>\delta$. The above analyse shows that $C_{u_0}$ is a complete null cone without conjugate points and cut points along the null generators.

We then estimate $r$ by using the equation $Dr=r\overline{\tr\chi}/2$. This equation follows from the variation of area:
$$D\int_{S_{\ub,u}}\D\mu_{\gs}=\int_{S_{\ub,u}}\Omega\tr\chi\D\mu_{\gs}.$$
We denote $a=\min_{S_{\delta,u_0}}\tr\chi(\delta)$, $A=\max_{S_{\delta,u_0}}\tr\chi(\delta)$, and $A\ge a>0$. By the formula of $\tr\chi$, we have $2/(\ub+2/a)\le \tr\chi(\ub)\le 2/(\ub+2/A)$ and this also holds for $\overline{\tr\chi}$. Therefore, we have
$$\log(\ub-\delta+2/a)-\log(2/a)\le\log r-\log r(\delta)\le\log(\ub+2/A)-\log(2/A)$$
which implies the condition on $r$ with the constant depending on $a$ and $A$.

It is not hard to estimate $\Lambda(\ub)$ and $\lambda(\ub)$ (the eigenvalues of $\gs(\ub)$ with respect to $\gs(\delta)$) because along the null generators, the conformal geometry $\widehat{\gs}$ does not change, then $\Lambda(\ub)=\lambda(\ub)=\phi(\ub)/\phi(\delta)$, whose bound depends only on $a$, $A$ and the maximal and minimal values of $\phi(\delta)$ on $S_{\delta,u_0}$.

We begin to consider the connection coefficients and curvature components. For $\tr\chi$, we have
$$|\tr\chi|\le Cr^{-1},\ |\tr\chi-\overline{\tr\chi}|\le Cr^{-2}$$
by the relation between $\ub$ and $r$, and $C$ depends on $a$ and $A$. For derivatives of $\tr\chi$, we argue by induction on the order $i$ of derivatives. For $i=1$, we consider the equation $D\ds\tr\chi=-\tr\chi\ds\tr\chi$ and apply Gronwall type estimates:
$$|r^3\ds\tr\chi|(\ub)\le C|r^3\ds\tr\chi|_{\ub=\delta}\le C.$$
Notice that the weight $r^3$ is related to the structure of the equation in an essential way. Now we assume $|(r\nablas)^{i-1}(r\tr\chi)|\le C_{i-1}$, and commute $\nablas^{i-1}$ with the equation $D\ds\tr\chi=-\tr\chi\ds\tr\chi$ to obtain
$$D\nablas^i\tr\chi=-\tr\chi\nablas^i\tr\chi+\sum_{j=1}^{i-1}\nablas^j\tr\chi\cdot\nablas^{i-j}\tr\chi$$
The second term on the right hand side comes from the commutator $[D,\nablas^{i-1}]$ (see Lemma \ref{commutator}, and remember $\chih\equiv0$).
Again by Gronwall type estimates, we can conclude
$$|(r\nablas)^{i}(r\tr\chi)|(\ub)\le |(r\nablas)^{i}(r\tr\chi)|_{S_{\delta,u_0}}\le C_i.$$

We can then turn to $\zeta=\eta=-\eta$ on $C_{u_0}$. Combining the null Codazzi equations $\divs\chih=\cdots$ and the equation for $D\eta$, we deduce that (see Chapter 2 of \cite{Chr})
$$D\zeta+\tr\chi\zeta=-\frac{1}{2}\ds\tr\chi.$$
Applying the Gronwall type estimates, we have
\begin{align*}|r^3\zeta|(\ub)\le &C|r^3\zeta|_{S_{\delta,u_0}}+C\int_\delta^{\ub}|r^3\ds\tr\chi|(\ub')\D\ub',\\
\le& C|r^3\zeta|_{S_{\delta,u_0}}+Cr(\ub)
\end{align*}
Dividing $r(\ub)$ on both sides leads to the desired estimate $|r^2\zeta|(\ub)\le C$. The derivatives of $\zeta$ are estimated in similar way.

We consider the equation for the Gauss curvature $K$, (see Chapter 5 of \cite{Chr})
$$DK+\tr\chi K=-\frac{1}{2}\Deltas\tr\chi.$$
We can deduce that $|(r\nablas)^i(r^2K)|\le C_i$ by commuting derivatives and applying Gronwall type estimates.

We turn to $\tr\chib$ and $\chibh$. We combine the Gauss equation and the equation for $D\tr\chib$ and obtain the following (see Chapter 2 of \cite{Chr})
$$D\tr\chib+\tr\chi\tr\chib=-2K-2\divs\zeta+2|\zeta|^2.$$
We then deduce that
\begin{align*}
|r^2\tr\chib|(\ub)\le& C|r^2\tr\chib|_{S_{\delta,u_0}}+C\int_{\delta}^{\ub}(1+r^{-1}+r^{-2})\D\ub'\\
\le& C|r^2\tr\chib|_{S_{\delta,u_0}}+Cr(\ub)
\end{align*}
and therefore $|r\tr\chib|(\ub)\le C$. A similar argument gives the desired estimate on derivatives of $\tr\chib$. For $\chibh$, the equation we consider is (see Chapter 5 of \cite{Chr})
$$D\chibh-\frac{1}{2}\tr\chi\chibh=-\nablas\tensor\zeta+\zeta\tensor\zeta.$$
We apply Gronwall type estimates to obtain $|(r\nablas)^i(r\chibh)|(\ub)\le C_i$.

We remark that we do not need $\omegab$ here. In fact, if we consider the equation for $D\omegab$, we find that $\omegab$ need not to decay. The decaying condition on $\omegab$ is ensured by our construction of canonical foliation. 

We finally turn the curvature components, which are easier to derive. We rely on the null Bianchi equations for $D\betab$, $D\rho$, $D\sigma$, $D\betab$, $D\alphab$ and $D\Db\alphab$. The last one comes from commuting $\Db$ with the equation for $D\alphab$. Remember $\chih\equiv0$ and then $\alpha\equiv0$. We can commute derivatives and apply Gronwall type estimate, to obtain
$$|(r\nablas)^i(r^4\beta,r^3\rho,r^3\sigma,r^2\betab,r\alphab,r\Db\alphab)|\le C_i.$$
We should remark here that because $\chih=0$ when $\ub=\delta$, then by null Codazzi equation we can obtain an improved estimate $\beta\sim1$ when $\ub=\delta$ (originally $\beta\sim\delta^{-1/2}$). We remark that the trivial extension of the initial data on $C_{u_0}$ to finite length is also used in \cite{L-Y} and such improvements are essential in the work \cite{L-Y}. In fact, we have obtained an even better improvement under an addition condition.

We can directly check that the assumptions in Theorem \ref{maintheorem} about $\mathcal{O}_0$ and $\mathcal{R}_0$.
We should also remark that the above estimates are done in $C^\infty$ and with stronger decay as compared to the assumptions in Theorem \ref{maintheorem}.

\end{proof}

\subsection{Structure of the Proof} 

We will prove the above Main Theorem in the remaining part of this paper. First of all, we give the strucutre of the proof.

We begin the proof by defining
$$\mathcal{A}_{\varepsilon,\Delta}=\{c\ge0: \text{$c$ satisfies the following two properties for small $\varepsilon>0$ and large $\Delta>0$}\},$$
where $\varepsilon$ is a small positive parameter and $\Delta$ is a positive large constant. They will be suitably chosen in the context depending only on $\mathcal{O}_0,\mathcal{R}_0,\underline{\mathcal{R}}_0$.
\begin{itemize}
\item[(1)]The solution of vacuum Einstein equations $g$ exists in a double null foliation given by $(\ub,u)$ for $0\le\ub\le c$, $0\le u\le\varepsilon$, where $\ub$ coincides with the affine function $s$ on $C_0$ and $u|_{\Cb_{\ub_*}}$ is canonical, which means that the following equation holds on $\Cb_{\ub_*}$:
\begin{align}\label{lastslice}
\overline{\log\Omega}=0,\quad \Deltas\log\Omega=\frac{1}{2}\divs\etab+\frac{1}{2}\left(\frac{1}{2}((\chih,\chibh)-\overline{(\chih,\chibh)})-(\rho-\overline{\rho})\right).
\end{align} 
\item[(2)]Written in the double null foliation given by $(\ub,u)$, $\mathcal{R}, \underline{\mathcal{R}}\le\Delta$.
\end{itemize}
We will prove, for $\varepsilon>0$ sufficiently small and $\Delta$ sufficiently large, $\ub_*=\sup\mathcal{A}_{\varepsilon,\Delta}=+\infty$. 

The proof is divided into following steps:
\begin{itemize}
\item[\bf Step 1.] We first construct on $\Cb_0$ a new function $u$ which is canonical and such that the section $u=0$ is the simply $S_{0,0}=C_0\bigcap\Cb_0$ and $u$ varies in $[0,\varepsilon]$ where $\varepsilon$ depends on $\mathcal{O}_0$ and $\underline{\mathcal{R}}_0$. This shows that $\mathcal{A}_{\varepsilon,\Delta}$ is not empty. We are going to argue by contradiction. We assume $\ub_*=\sup\mathcal{A}_{\varepsilon,\Delta}<+\infty$. %The proof is basically as in Section \ref{SectionLastSlice}. In this case, the condition $W(s,\theta)\le\varepsilon+\delta'$ is not needed in the definition of the function space $\mathcal{K}$, and the proof is slightly easier. We then proceed as in Step 4, showing that $\mathcal{A}$ is not empty. We then assume $\ub_*=\sup\mathcal{A}<+\infty$.
\item[\bf Step 2.] In this and the next step, we work on the space-time region $M_{\ub_*,\varepsilon}$ which corresponding to $0\le\ub\le\ub_*$, $0\le u\le\varepsilon$. It is not hard to see $\ub_*\in\mathcal{A}_{\varepsilon,\Delta}$. We will prove that, if $\varepsilon>0$ is sufficiently small,
 \begin{align*}
\mathcal{O}[\chibh,\tr\chib,\etab,\omega,\chih,\Omega\tr\chi-\overline{\Omega\tr\chi},\eta]\le C(\mathcal{O}_0,\mathcal{R}_0), \ \mathcal{O}[\omegab]\le C(\mathcal{O}_0,\mathcal{R}).
\end{align*}
This is done in Section \ref{SectionConnection}, Proposition \ref{connection3}.
\item[\bf Step 3.]Under the conclusion of Step 2, we prove that, for $\varepsilon>0$ sufficiently small, 
$$\mathcal{R},\underline{\mathcal{R}}\le C(\mathcal{O}_0,\mathcal{R}_0,\underline{\mathcal{R}}_0).$$
This is done in Section \ref{SectionCurvature}, Proposition \ref{curvaturecompletenullcone}.
\item[\bf Step 4.]We extend the solution $g$ to $0\le\ub\le\ub_*+\delta$, $0\le u\le \varepsilon+\delta'$ for $\delta,\delta'$ sufficiently small, such that it holds again $\mathcal{O},\mathcal{R},\underline{\mathcal{R}}\le 2C(\mathcal{O}_0,\mathcal{R}_0,\underline{\mathcal{R}}_0)$. The extension of the space-time follows by \cite{Luk}. We then construct on $\Cb_{\ub_*+\delta}$ a new function $u_\delta$ for $0\le u_\delta\le\varepsilon$ such that $u_\delta|_{\Cb_{\ub_*+\delta}}$ is canonical, provided that $\varepsilon>0$ and $\delta>0$ are sufficiently small ($\delta$ may depend on $\varepsilon$ and $\delta'$). Then, by continuity, if $\delta>0$ is sufficiently small,  the new function $u_\delta$ can be extended inside up to $\Cb_0$ as an optical function, and we have the new double null foliation $0\le\ub\le\ub_*+\delta$, $0\le u_\delta\le\varepsilon$. In particular, the norms $\mathcal{R}_\delta,\underline{\mathcal{R}}_\delta$ expressed in the new foliation are bounded by some constant depending on $\mathcal{O}_0,\mathcal{R}_0,\underline{\mathcal{R}}_0$.
\end{itemize}
Now we can choose $\Delta$ sufficiently large such that $\ub_*+\delta\in\mathcal{A}_{\varepsilon,\Delta}$ which leads to a contradition to that $\ub_*=\sup\mathcal{A}_{\varepsilon,\Delta}<+\infty$. Finally, we will show that we can construct a global retarded time function $u$ and complete the proof.

\section{Proof of the Main Theorem}

\subsection{Canonical Foliation on Initial Null Cone}\label{sectioninitialslice}

We shall carry out Step 1 of the proof to construct canonical foliation on $\Cb_0$ in this subsection.

The construction of a canonical foliation on the last slice is carried out both in \cite{Ch-K} and \cite{K-N} (or \cite{C-N10}). The case for space-like hypersurface is considered in \cite{Ch-K} and the case for incoming null hypersurface is considered in \cite{K-N}. The full detail of the local existence of the canonical foliation on a null cone is given in \cite{N}. For the sake of completeness, and because the setting is not exactly the same in our case, we give a proof here but do not carry out some detail computations. The reader can also refer to Chapter 3 of \cite{Sau} for the comparsion between two different foliation and Chapter 7 for another argument.

We first define the following addition quantities which are needed in the construction of canonical foliation:
\begin{align*}
\mathcal{R}[\Db\betab]=\sup_{u}\sum_{i=0}^1\|(r\nablas)^i\Db\betab\|_{L^2(C_u)},&\\
\underline{\mathcal{R}}[\Db\alphab]=\sup_{\ub}\sum_{i=0}^1\|(r\nablas)^i\Db\alphab\|_{L^2(\Cb_{\ub})}&,\quad
\underline{\mathcal{R}}[\Db^2\alphab]=\sup_{\ub}\|\Db^2\alphab\|_{L^2(\Cb_{\ub})},\\
\mathcal{O}^2[\Db\omegab]=\sup_{\ub,u}\sum_{i=0}^1\|(r\nablas)^i\Db\omegab\|_{L^2(S_{\ub,u})}&,\quad \mathcal{O}[\Db^2\omegab]=\sup_{\ub}\|\Db^2\omegab\|_{L^2(\Cb_{\ub})}.
\end{align*}

Now we work on an arbitrary incoming null cone $\Cb_{\ub}$ with a background foliation given by a function $u$. We assume that on this null cone, we have:
\begin{align}\label{assumption0}
\mathcal{O},\underline{\mathcal{R}},\underline{\mathcal{R}}[\Db\alphab,\Db^2\alphab],\mathcal{O}^2[\Db\omegab],\mathcal{O}[\Db^2\omegab] \le C
\end{align}
for some $C$. The norms appear above should be understood as the norm taken only on one single null cone $\Cb_{\ub}$.

To define a new foliation, we use a function $W$ defined on $[0,\varepsilon]\times S_{\ub,0}$ to represent a new foliation in the way that, the new foliation function ${}^{(W)}u$  is defined by the relation ${}^{(W)}u(W(s,\theta),\theta)=s$\footnote{Here $s$ is the parameter varying in $[0,\varepsilon]$, but not the affine parameter on $C_0$.}. Under the new foliation, we have also the new ``lapse'' function ${}^{(W)}\Omega$. $W$ represents a foliation iff
\begin{align*}
{}^{(W)}a(s,\theta)\triangleq\frac{\partial W}{\partial s}(s,\theta)={}^{(W)}\Omega^2(W(s,\theta),\theta)\Omega^{-2}(W(s,\theta),\theta)>0.
\end{align*}

Under the new foliation, the following vectorfields
\begin{equation}\label{nullframerelation}
\begin{aligned}
{}^{(W)}\Lb&={}^{(W)}a\Lb,\\
{}^{(W)}L'&=\frac{1}{{}^{(W)}a}(L'+|{}^{(W)}\nablas W|^2\Lb+2{}^{(W)}\nablas^AWE_A),\\
{}^{(W)}E_A&=E_A+{}^{(W)}\nablas_AW\Lb
\end{aligned}
\end{equation}
also form a null frame with $g({}^{(W)}\Lb,{}^{(W)}\Lb)=g({}^{(W)}L',{}^{(W)}L')=0$, $g({}^{(W)}\Lb,{}^{(W)}L')=-2$, and $[{}^{(W)}\Lb,{}^{(W)}E_A]=0$ if $[\Lb,E_A]=0$. In addition,  ${}^{(W)}\Lb{}^{(W)}u=1$. We can also compute the relation of the connection coefficients and curvature components in different foliations:
\begin{align*}
{}^{(W)}\Omega{}^{(W)}\chi=&\Omega^2[\chi'+2\eta\otimes{}^{(W)}\nablas W+2{}^{(W)}\nablas W\otimes\eta\\
&-2\omegab{}^{(W)}\nablas W\otimes{}^{(W)}\nablas W-|{}^{(W)}\nablas W|^2\Omega\chib+2{}^{(W)}\nablas^2W],\\
{}^{(W)}\Omega^{-1}{}^{(W)}\chib=&\Omega^{-1}\chib,\\
{}^{(W)}\etab=&\etab+\Omega\chib\cdot{}^{(W)}\nablas W,\\
%{}^{(W)}\eta=&\eta+2{}^{(W)}\ds\log({}^{(W)}\Omega\Omega^{-1})+2\omegab{}^{(W)}\nablas W-\Omega\chib\cdot{}^{(W)}\nablas W,\\
{}^{(W)}\rho=&\rho+2({}^{(W)}\nablas W,\betab)+\alphab({}^{(W)}\nablas W,{}^{(W)}\nablas W),\\
{}^{(W)}\betab=&{}^{(W)}a(\betab-\alphab\cdot{}^{(W)}\nablas W),\\
{}^{(W)}\alphab=&{}^{(W)}a^2\alphab.
\end{align*}

Suppose now that $W$ represents a foliation. We consider a map $\mathcal{A}$, which is defined by $$\mathcal{A}(W)(s,\theta)=\int_0^s{}^{(W_\mathcal{A})}\Omega^2(W(s',\theta),\theta)\Omega^{-2}(W(s',\theta),\theta)\D s'$$ where ${}^{(W_\mathcal{A})}\Omega$ is the solution of the equation
\begin{align*}
{}^{(W)}\Deltas\log{}^{(W_\mathcal{A})}\Omega(W(s,\theta),\theta)&={}^{(W)}G(W(s,\theta),\theta),\\
{}^{{}^{(W)}}\overline{\log{}^{(\mathcal{A}(W))}\Omega}(s)&=0,
\end{align*}
where
\begin{align*}
{}^{(W)}G\triangleq\frac{1}{2}{}^{(W)}\divs{}^{(W)}\etab+\frac{1}{2}\left(\frac{1}{2}(({}^{(W)}\chih,{}^{(W)}\chibh)-{}^{{}^{(W)}}\overline{({}^{(W)}\chih,{}^{(W)}\chibh)})-({}^{(W)}\rho-{}^{{}^{(W)}}\overline{{}^{(W)}\rho})\right).
\end{align*}
\begin{comment}
Since ${}^{(W)}G$ only depends on the sections of $S_{\ub,{}^{(W)}u}$, not on the lapse, we can express ${}^{(W)}G$ in terms of the the connection coeffients and curvature components in background foliation and $\nablas W$, $\nablas^2 W$, where $\nablas$ is the conncection related to background foliation. Therefore, the above equations should be understood as a family of equations (with parameter $s$) written on $S_{\ub,0}$, with pull back metric from $S_{\ub,{}^{(W)}u}$.
\end{comment}

We can express $G$ in terms of geometric quantities related to background foliation and derivatives of $W$. We can compute (in a rough form):
\begin{align*}
{}^{(W)}\divs{}^{(W)}\etab=&\divs\etab+{}^{(W)}\nablas W\cdot\Omega\chib\cdot\etab+{}^{(W)}\nablas W\cdot\Db\etab\\
&+{}^{(W)}\nablas W\cdot\nablas(\Omega\chib)+\Omega\chib\cdot{}^{(W)}\nablas^2W+{}^{(W)}\nablas W\cdot\Omega\chib\cdot\Omega\chib\\
&+{}^{(W)}\nablas W\cdot{}^{(W)}\nablas W\cdot\Db(\Omega\chib),\\
({}^{(W)}\chih,{}^{(W)}\chibh)=&(\chih,\chibh)-\Omega^2|{}^{(W)}\nablas W|^2|\chibh|^2\\
&+\Omega\chibh\cdot(\eta\cdot{}^{(W)}\nablas W+\omegab{}^{(W)}\nablas W\cdot{}^{(W)}\nablas W+{}^{(W)}\nablas^2 W),\\
{}^{(W)}\rho=&\rho+2{}^{(W)}\nablas W\cdot\betab+\alphab\cdot{}^{(W)}\nablas W\cdot{}^{(W)}\nablas W.
\end{align*}

For $W$ which does not represent a foliation (but only a family of sections), the right hand sides of the above inequalities still make sense. Therefore, we can extend the map $\mathcal{A}$ to the case that $W$ does not represent a foliation. Notice that if $W$ is a fixed point of the map $\mathcal{A}$, then we can estimate $\log\Omega$ directly and conclude that $W$ represents a foliation, then the function ${}^{(W)}u$ given by $W$ is canonical (see the equation \eqref{lastslice}). 

We first carry out some calculation and then consider the behavior of $\mathcal{A}$ on suitable chosen function spaces.

For any $W$, we denote $S_{W(s)}$ or simply $S_W$ to be the sections. We first compare two different family of sections $W_1$ and $W_2$. For fixed $s$, We consider a family of sections $W_t\triangleq tW_1(s,\theta)+(1-t)W_2(s,\theta)$ where $t\in[0,1]$. We will compute
\begin{align*}
\frac{\D}{\D t}\left({}^{(W_t)}\Deltas\log{}^{((W_t)_\mathcal{A})}\Omega\right)=\frac{\D}{\D t}{}^{(W_t)}G.
\end{align*}

We first compute $\frac{\D}{\D t}{}^{(W_t)}G$. We denote $\widetilde{W}=W_1-W_2$, then by the expression of ${}^{(W_t)}\rho$,
\begin{align*}
\frac{\D}{\D t}{}^{(W_t)}\rho=&\widetilde{W}\Db\rho+2{}^{(W_t)}\nablas \widetilde{W}\cdot\betab+2\widetilde{W}{}^{(W_t)}\nablas W_t\cdot\Db\beta^\sharp\\
&+2\alphab\cdot{}^{(W_t)}\nablas\widetilde{W}\cdot{}^{(W_t)}\nablas W_t+\widetilde{W}\Db\alphab^\sharp\cdot{}^{(W_t)}\nablas W_t\cdot{}^{(W_t)}\nablas W_t.
\end{align*}
Notice that we use $\betab^\sharp$ and $\alphab^\sharp$ to denote the corresponding contravariant tensorfields to $\betab$ and $\alphab$. The derivative $\Db$ also applies to the metric $\gs$. The addition term $\Omega\chib$ arises, but since $\Omega\chib$ is in $L^\infty$, we do not need to care about it.

We also compute
\begin{align*}
\frac{\D}{\D t}{}^{{}^(W_t)}\overline{{}^{(W_t)}\rho}=-{}^{{}^(W_t)}\overline{\widetilde{W}\Omega\tr\chib}{}^{{}^(W_t)}\overline{{}^{(W_t)}\rho}+{}^{{}^(W_t)}\overline{\frac{\D}{\D t}{}^{(W_t)}\rho+\widetilde{W}\Omega\tr\chib{}^{(W_t)}\rho}.
\end{align*}
The other terms of $\frac{\D}{\D t}{}^{(W_t)}G$ are computed in a similar way.

We also compute
\begin{align*}
&\frac{\D}{\D t}\left({}^{(W_t)}\Deltas\log{}^{((W_t)_\mathcal{A})}\Omega\right)\\
=&{}^{(W_t)}\Deltas\left(\frac{\D}{\D t}\log{}^{((W_t)_\mathcal{A})}\Omega\right)-\widetilde{W}\Omega\tr\chib{}^{(W_t)}\Deltas\log{}^{((W_t)_\mathcal{A})}\Omega\\
&-2\widetilde{W}\Omega\chibh\cdot{}^{(W_t)}\nablas^2\log{}^{((W_t)_\mathcal{A})}\Omega-2{}^{(W_t)}\divs(\widetilde{W}\Omega\chibh)\cdot{}^{(W_t)}\nablas\log{}^{((W_t)_\mathcal{A})}\Omega\\
\triangleq&{}^{(W_t)}\Deltas\left(\frac{\D}{\D t}\log{}^{((W_t)_\mathcal{A})}\Omega\right)-{}^{(W_t)}G_1.
\end{align*}
In the above expression, ${}^{(W_t)}\Deltas\log{}^{((W_t)_\mathcal{A})}\Omega={}^{(W_t)}G$, and ${}^{(W_t)}\nablas^2\log{}^{((W_t)_\mathcal{A})}\Omega$ is controlled in $L^2(S_{W_t})$, by $\|{}^{(W_t)}G\|_{L^2(S_{W_t})}$, if we assume that for all $W$,
\begin{align}\label{assumption1}
{}^{(W)}r\|{}^{(W)}K\|_{L^2(S_{W_t})}\le C%(\mathcal{O}_0,\mathcal{R}_0,\underline{\mathcal{R}}_0)
\end{align}
and then the elliptic estimate applies. The expression ${}^{(W_t)}\divs(\widetilde{W}\Omega\chibh)$ is of the form ${}^{(W_t)}\nablas \widetilde{W}\cdot\Omega\chibh+\widetilde{W}\cdot(\nablas(\Omega\chibh)+{}^{(W_t)}\nablas W_t\cdot\Db(\Omega\chibh))$, and ${}^{(W_t)}\nablas\log{}^{((W_t)_\mathcal{A})}\Omega$ is also controlled in $L^p(S_{W_t})$ for all $p\in[2,+\infty)$, by $\|{}^{(W_t)}G\|_{L^2(S_{W_t})}$.

Also, we compute, by ${}^{{}^(W_t)}\overline{\log{}^{((W_t)_\mathcal{A})}\Omega}=0$,
\begin{align*}
0=\frac{\D}{\D t}{}^{{}^(W_t)}\overline{\log{}^{((W_t)_\mathcal{A})}\Omega}={}^{{}^(W_t)}\overline{\frac{\D}{\D t}\log{}^{((W_t)_\mathcal{A})}\Omega+\widetilde{W}\Omega\tr\chib\log{}^{((W_t)_\mathcal{A})}\Omega}
\end{align*}

Therefore,
\begin{align*}
{}^{(W_t)}\Deltas\left(\frac{\D}{\D t}\log{}^{((W_t)_\mathcal{A})}\Omega\right)=&\frac{\D}{\D t}{}^{(W_t)}G+{}^{(W_t)}G_1,\\ {}^{{}^(W_t)}\overline{\frac{\D}{\D t}\log{}^{((W_t)_\mathcal{A})}\Omega}=&-{}^{{}^(W_t)}\overline{\widetilde{W}\Omega\tr\chib\log{}^{((W_t)_\mathcal{A})}\Omega}
\end{align*}
and we conclude by elliptic estimate that
\begin{align*}
&\left\|{}^{(W_t)}(r\nablas)^{\le2}\left(\frac{\D}{\D t}\log{}^{((W_t)_\mathcal{A})}\Omega\right)\right\|_{L^2(S_{W_t})}\\
\lesssim&\|{}^{(W_t)}r^2(\frac{\D}{\D t}{}^{(W_t)}G+{}^{(W_t)}G_1)\|_{L^2(S_{W_t})}+\|{}^{(W_t)}r^2{}^{{}^{(W_t)}}\overline{\widetilde{W}\Omega\tr\chib\log{}^{((W_t)_\mathcal{A})}\Omega}\|_{L^2(S_{W_t})}
\end{align*}

The notation $A\lesssim B$ here and in this subsection is understood as $A\le C'B$ where $C'$ depends on the constants in the assumptions \eqref{assumption0}, \eqref{assumption1} and the following: for $i=0,1,2$, any $W$ and any function $f$, we assume
\begin{align}\label{assumption2}
C^{-1}\|(r\nablas)^if(W(t,\cdot),\cdot)\|_{L^2(S_{\ub,0})}\le \|{}^{(W)}(r\nablas)^i f\|_{L^2(S_{W_t})}\le C\|{}^{(W)}(r\nablas)^if(W(t,\cdot),\cdot)\|_{L^2(S_{\ub,0})}.
\end{align}
So, to compare two different foliations, 
\begin{align*}
&\|(r\nablas)^{\le2}(\log{}^{((W_1)_{\mathcal{A}})}\Omega(W_1(s,\cdot),\cdot)-\log{}^{((W_2)_{\mathcal{A}})}\Omega(W_2(s,\cdot),\cdot))\|_{L^2(S_{\ub,0})}\\
\lesssim&\left\|(r\nablas)^{\le2}\left(\int_0^1\frac{\D}{\D t}\log{}^{((W_t)_\mathcal{A})}\Omega(W_t(\cdot),\cdot)\D t\right)\right\|_{L^2(S_{\ub,0})}\\
\lesssim&\int_0^1\left\|{}^{(W_t)}(r\nablas)^{\le2}\left(\frac{\D}{\D t}\log{}^{((W_t)_\mathcal{A})}\Omega(W_t(\cdot),\cdot)\right)\right\|_{L^2(S_{W_t})}\D t\\
\lesssim&\int_0^1\left(\|{}^{(W_t)}r^2(\frac{\D}{\D t}{}^{(W_t)}G+{}^{(W_t)}G_1)\|_{L^2(S_{W_t})}+\|{}^{(W_t)}r^2{}^{{}^{(W_t)}}\overline{\widetilde{W}\Omega\tr\chib\log{}^{((W_t)_\mathcal{A})}\Omega}\|_{L^2(S_{W_t})}\right)\D t.
\end{align*}

We estimate
\begin{align}
&\nonumber\|{}^{(W_t)}r^2\frac{\D}{\D t}{}^{(W_t)}G\|_{L^2(S_{W_t})}\\\nonumber\lesssim&|{}^{(W_t)}r\widetilde{W}|\left(\|r\Db(\divs\etab,(\chih,\chibh),\rho)\|_{L^2(S_{(W_t)})}+\|r\Omega\tr\chib(\divs\etab,(\chih,\chibh),\rho)\|_{L^2(S_{(W_t)})}\right.\\\nonumber
&+\|{}^{(W_t)}r^{1/2}{}^{(W_t)}\nablas W_t\|_{L^4(S_{(W_t)})}\\\nonumber&\quad\times\|r^{3/2}(\Db(\Omega\chib\cdot\eta,\Db\etab,\nablas(\Omega\chib),\Omega\chib\cdot\Omega\chib,\Omega\chib\cdot\eta,\betab)^\sharp,r^{3/2}\nablas(\Omega\chib))\|_{L^4(S_{(W_t)})}\\\nonumber
&+\|{}^{(W_t)}r^{3/2}|{}^{(W_t)}\nablas W_t|^2\|_{L^4(S_{(W_t)})}\|r^{1/2}\Db(\Db(\Omega\chib),|\Omega\chibh|^2,\Omega\chib\omegab,\alphab,\Omega\chib)^\sharp\|_{L^4(S_{(W_t)})}\\\nonumber
&\left.+\|{}^{(W_t)}r{}^{(W_t)}\nablas^2 W_t\|_{L^2(S_{(W_t)})}\|r\Db(\Omega\chib)^\sharp\|_{L^\infty(S_{(W_t)})}\right)\\\nonumber
&+\|{}^{(W_t)}r^{3/2}{}^{(W_t)}\nablas \widetilde{W}\|_{L^4(S_{(W_t)})}\|r^{1/2}(\Omega\chib\cdot\eta,\Db\etab,\nablas(\Omega\chib),\Omega\chib\cdot\eta,\betab)\|_{L^4(S_{(W_t)})}\\\nonumber
&+\|{}^{(W_t)}r^{7/4}{}^{(W_t)}\nablas \widetilde{W}\|_{L^8(S_{(W_t)})}\|{}^{(W_t)}r^{3/4}{}^{(W_t)}\nablas W_t\|_{L^8(S_{(W_t)})}\\\nonumber
&\quad\times\|r^{1/2}(\Db(\Omega\chib),|\Omega\chibh|^2,\Omega\chib\omegab,\alphab,\Omega\chib)\|_{L^4(S_{(W_t)})}\\
&\label{dtG}+\|{}^{(W_t)}r^2{}^{(W_t)}\nablas^2 \widetilde{W}\|_{L^2(S_{(W_t)})}\|\Omega\chib\|_{L^\infty(S_{(W_t)})}.
\end{align}
We need to estimate the norms of the background quantities on the sphere $S_{(W_t)}$. For arbitrary section given by a function $Y=Y(\theta)\ge0$, we consider a family of sections $Y_t(\theta)=tY(\theta)$ which connects the sphere $S_Y$ and $S_{\ub,0}$. For arbitrary tensorfields $\psi$, we compute
\begin{align*}
\frac{\D}{\D t}\int_{S_{Y_t}}|\psi|^4\D\mu_t&=\int_{S_{Y_t}}Y(\Db(|\psi|^4)+\Omega\tr\chib|\psi|^4)\D\mu_t\\
&\lesssim \int_{S_{Y_t}}Y|\psi|^4\D\mu_t+\int_{S_{Y_t}}Y|\psi|^3|\Db\psi|\D\mu_t.
\end{align*}
We integrate the above inequality over $[0,1]$. Now, if we assume
\begin{align}\label{assumption3}
W(s,\theta)\le \varepsilon_W
\end{align}
then we can choose $\varepsilon_W$ small enough so that
\begin{align*}
\int_{S_Y}|\psi|^4\D\mu\lesssim&\int_{S_{\ub,0}}|\psi|^4\D\mu_0+\left(\int_0^1\int_{S_{Y_t}}Y|\psi|^6\D\mu_t\D t\right)^{1/2}\left(\int_0^1\int_{S_{Y_t}}Y|\Db\psi|^2\D\mu_t\D t\right)^{1/2}\\
\lesssim&\|\psi\|_{L^4(S_{\ub,0})}^4+\|\psi\|_{L^6(\Cb_{\ub})}^3\|\Db\psi\|_{L^2(\Cb_{\ub})}.
\end{align*}
The second inequality holds because the volumn element $Y\D\mu_t\D t=\D\mu_{S_{\ub,u}}\D u$. Therefore, by Sobolev inequality, we have
\begin{equation}\label{SobolevCb}
\begin{split}
\|\psi\|_{L^4(S_Y)}
\lesssim\|\psi\|_{L^4(S_{\ub,0})}+\left(\|r^{-1/2}(r\nablas)\psi\|_{L^2(\Cb_{\ub})}+\|r^{-1/2}\psi\|_{L^2(\Cb_{\ub})}\right)^{3/4}\|\Db\psi\|_{L^2(\Cb_{\ub})}^{1/4}
\end{split}
\end{equation}

Using the above inequality, we can relate the norms of the background quantities on the sphere $S_{(W_t)}$ to the norms on $S_{\ub,0}$ and the norms of $\Db$ derivative on $L^2(\Cb_{\ub})$. One can check that the norms in \eqref{assumption0} are enough to control the norms appearing in \eqref{dtG}. Again if $\varepsilon$ is small depending on $C_1$, we have
\begin{align*}
\int_0^1\|{}^{(W_t)}r^2\frac{\D}{\D t}{}^{(W_t)}G\|_{L^2(S_{W_t})}\D t
\le C' \|\nablas^{\le2}(W_1(s,\cdot)-W_2(s,\cdot))\|_{L^2(S_{\ub,0})}.
\end{align*}
Here $C'$ depends on the constants in the assumptions \eqref{assumption0}, \eqref{assumption1}, \eqref{assumption2} and the bound $\sup_s\|(r\nablas)^{\le2} W_t(s,\cdot)\|_{L^2(S_{\ub,0})}$. The other terms are estimated in a similar way and we have
\begin{equation}\label{contraction}
\begin{split}
&\|(r\nablas)^{\le2}(\log{}^{((W_1)_{\mathcal{A}})}\Omega(W_1(s,\cdot),\cdot)-\log{}^{((W_2)_{\mathcal{A}})}\Omega(W_2(s,\cdot),\cdot))\|_{L^2(S_{\ub,0})}\\
\le& C'\|(r\nablas)^{\le2}(W_1(s,\cdot)-W_2(s,\cdot))\|_{L^2(S_{\ub,0})}.
\end{split}
\end{equation}

Finally,  we have the following:
\begin{proposition}\label{canonicalslice}
Suppose that on incoming null cone $\Cb_{\ub}$ with background foliation given by $u$, we have the assumptions \eqref{assumption0}, \eqref{assumption1}, \eqref{assumption2}, \eqref{assumption3} hold. Then the inequality \eqref{contraction} holds for any two families of secions $W_1$ and $W_2$, if $\varepsilon_W$ is sufficiently small.
\end{proposition}
\begin{remark}
It is not hard to see  \eqref{assumption1}, \eqref{assumption2} hold for some $C$ depending on the bound in the assumption \eqref{assumption0} and $\sup_s\|(r\nablas)^{\le2} W_t(s,\cdot)\|_{L^2(S_{\ub,0})}$ if $\varepsilon$ is sufficiently small.
\end{remark}

We apply the above proposition on $\Cb_0$ to prove the existence of the canonical foliation on $\Cb_0$ if $\varepsilon$ is sufficiently small. Given $\varepsilon>0$, we define the function space $\mathcal{K}_{0,\varepsilon}\subset C([0,\varepsilon],H^2(S_{0,0}))$ to be the collection of functions $W$ which satisfy
\begin{align*}
\begin{cases}W(0,\theta)=0,W\ge0\\\sup_s\|(r\nablas)^{\le2}W(s,\cdot)\|_{L^2(S_{0,0})}\le C_{\mathcal{K}_{0,\varepsilon}}.\end{cases}
\end{align*}

We will choose $C_{\mathcal{K}_{0,\varepsilon}}$ sufficiently small such that the assumption \eqref{assumption3} holds for suitable $\varepsilon_W$. Notice that now the constant $C'$ above depends on $C(\mathcal{O}_0,\mathcal{R}_0,\underline{\mathcal{R}}_0)$. We simply take first $W_1=W$ and $W_2=0$. Then by \eqref{contraction},
\begin{align*}
&\|(r\nablas)^{\le2}W(s,\cdot)\|_{L^2(S_{0,0})}\\
\lesssim&\int_0^s\|(r\nablas)^{\le2}{}^{((W)_{\mathcal{A}})}({}^{(W_\mathcal{A})}\Omega^2(W(s',\cdot),\cdot)\Omega^{-2}(W(s',\cdot),\cdot))\|_{L^2(S_{0,0})}\D s\\
\lesssim& \varepsilon C(\mathcal{O}_0,\mathcal{R}_0,\underline{\mathcal{R}}_0,C_{\mathcal{K}_{0,\varepsilon}})\sup_s\|(r\nablas)^{\le2}W(s,\cdot)\|_{L^2(S_{0,0})}\\
\le&\frac{1}{2}C_{\mathcal{K}_{0,\varepsilon}}.
\end{align*}
We remark that we should use \eqref{SobolevCb} again here because we compare the bound of background lapse $\Omega$ on the sphere $S_{0,s}$ to the bound on $S_{W(s)}$. Now we have verified $\mathcal{A}(\mathcal{K}_{0,\varepsilon})\subset \mathcal{K}_{0,\varepsilon}$. To see $\mathcal{A}:\mathcal{K}_{0,\varepsilon}\to\mathcal{K}_{0,\varepsilon}$ is a contraction, we only need to apply \eqref{contraction} again for any two $W_1$ and $W_2$.

We then complete the Step 1 of the proof.

%We remark that once the fixed point of $\mathcal{A}$, i.e., the solution of canonical foliation is found, we can gain one more derivative directly from elliptic esimate via the equation \eqref{lastslice}. In particular, we find the bound of $\log\Omega$ which implies the solution $W$ represents a foliation.

\subsection{Estimates for Connection}\label{SectionConnection}

In this subsection, we derive estimates for connection coefficients in terms of curvature components and complete the Step 2 of the proof. Remamber that we work on the space-time with the double null coordinate $(\ub,u)$ with $(\ub,u)\in[0,\ub_*]\times[0,\varepsilon]$, and $u$ is canonical on $\Cb_{\ub_*}$. Precisely, we establish the following proposition. 
\begin{proposition}\label{connection3}
Assume that $\mathcal{R},\underline{\mathcal{R}}<\infty$. Then we have
\begin{align*}
\mathcal{O}[\chibh,\tr\chib,\etab,\omega,\chih,\tr\chi,\Omega\tr\chi-\overline{\Omega\tr\chi},\eta]\le C(\mathcal{O}_0,\mathcal{R}_0), \ \mathcal{O}[\omegab]\le C(\mathcal{O}_0,\mathcal{R}).
\end{align*}
\end{proposition}

We will prove this proposition in this subsection. We first introduce: 

{\bf Bootstrap assumption (1)}: $\mathcal{O}\le\Delta_1$,

{\bf Bootstrap assumption (2)}:\\
\begin{align*}& \sup_{u}\|(r\nablas)^3(\etab,\chih,r\omega),(r\nablas)^2(r\mub),(r\nablas)(r^3\kappa))\|_{L^2(C_u)}\le\Delta_2,\\
&\sup_{\ub}\|(r\nablas)^{3}(r^{1/2}\eta,\chibh,\omegab),(r\nablas)^2(r^{3/2}\mu),(r\nablas)(r^2\kappab)\|_{L^2(\Cb_{\ub})}\le\Delta_2,\\
&\sup_{\ub,u}\|(r\nablas)^3(r\tr\chi,\tr\chib)\|_{L^2(S_{\ub,u})}\le\Delta_2.
\end{align*}

We will prove the following two propositions. Then Proposition \ref{connection3} follows directly by a bootstrap argument.

\begin{proposition}\label{connection1}
Assume that $\mathcal{R}, \underline{\mathcal{R}}<\infty$ and bootstrap assumptions (1)(2) hold (in fact, we use only the bootstrap assumption (1) and $\sup_{\ub}\|(r\nablas)^3(r^{1/2}\eta,\omegab)\|_{L^2(\Cb_{\ub})}\le\Delta_2$). Then for $\varepsilon$ sufficiently small, we have
\begin{align*}
\mathcal{O}[\chibh,\tr\chib,\etab,\omega,\chih,\tr\chi,\Omega\tr\chi-\overline{\Omega\tr\chi},\eta]\le C(\mathcal{O}_0,\mathcal{R}_0), \ \mathcal{O}[\omegab]\le C(\mathcal{O}_0,\mathcal{R})
\end{align*}
In particular, we can choose $\Delta_1$ sufficiently large such that $\mathcal{O}\le\frac{1}{2}\Delta_1$.
\end{proposition}
\begin{proposition}\label{connection2}
Assume that $\mathcal{R},\underline{\mathcal{R}}<\infty$ and bootstrap assumptions (1)(2) hold. Then for $\varepsilon$ sufficiently small, and $\Delta_2$ sufficiently large, the norms in bootstrap assumption (2) are bounded by $\frac{1}{2}\Delta_2$. 
\end{proposition}

\begin{proof}[Proof of Proposition \ref{connection1}]
We first establish the following lemma, which says that the geometric quantities share the same estimates up to a multiple by $\Omega$ and we need not to distinguish $\ub$ and $r$:
\begin{lemma}If $\varepsilon>0$ is sufficiently small depending on $\Delta_1,\Delta_2$, then 
\begin{align*}
 &\|(r\nablas)^{\le1}\log\Omega\|_{L^\infty(S_{\ub,u})}+ \|r^{-1/2}(r\nablas)^2\log\Omega\|_{L^4(S_{\ub,u})}+\|r^{-1}(r\nablas)^{3}\log\Omega\|_{L^2(S_{\ub,u})}\\
\le& r^{-1}C(\mathcal{O}_0,\mathcal{R}_0).
\end{align*}
In particular, we have $C(\mathcal{O}_0,\mathcal{R}_0)^{-1}\le r(\Omega-1)\le C(\mathcal{O}_0,\mathcal{R}_0)$. Also, we have $$C(\mathcal{O}_0)^{-1}(1+\ub)\le r\le C(\mathcal{O}_0)(1+\ub).$$
\end{lemma}
\begin{proof}Recall that $u|_{\Cb_{\ub_*}}$ is canonical. That is, on $S_{\ub_*,0}$, $\Omega$ satisfies the equation
\begin{align*}
\overline{\log\Omega}=0,\quad \Deltas\log\Omega=\frac{1}{2}\divs\etab+\frac{1}{2}\left(\frac{1}{2}((\chih,\chibh)-\overline{(\chih,\chibh)})-(\rho-\overline{\rho})\right).
\end{align*}
The expression on the right hand side is invariant, in particular, does not depend on $\Omega$. Therefore, by $L^2$ elliptic estimate and Sobolev inequalities (which are valid because we are working on initial null cone $C_0$), we have
\begin{align*}
 &\|(r\nablas)^{\le1}\log\Omega\|_{L^\infty(S_{\ub_*,0})}+ \|r^{-1/2}(r\nablas)^2\log\Omega\|_{L^4(S_{\ub_*,0})}+\|r^{-1}(r\nablas)^{3}\log\Omega\|_{L^2(S_{\ub_*,0})}\\
\le& r^{-1}|_{S_{\ub_*,0}}C(\mathcal{O}_0,\mathcal{R}_0).
\end{align*}
Because we do not change the foliation $\ub$ on $C_0$, $\Omega$ is extended as a constant along every null generator of $C_0$. The above estimates then hold along the whole $C_0$. 

We introduce an auxiliary bootstrap assumption: $(4C_r)^{-1}(1+\ub)\le r\le 4C_r(1+\ub)$. Then the conclusion about  $\|\log\Omega\|_{L^\infty(S_{\ub,u})}$ by the equation $\Db\log\Omega=\omegab$:
\begin{align*}
|\log\Omega|\le|\log\Omega|_{S_{\ub,0}}|+\int_0^u|\omegab|\D u'\le r^{-1}|_{S_{\ub_*,0}}C(\mathcal{O}_0,\mathcal{R}_0)+\varepsilon(\min_{\Cb_{\ub}}r)^{-1}\Delta_1,
\end{align*}
and then multiply both sides by $r$ and choose $\varepsilon$ sufficiently small.

We now go to the estimate on $r$. This follows directly from the equation $\Db r=r\overline{\Omega\tr\chib}/2$ and the estimate on $\Omega$ derived above. Therefore the auxillary bootstrap assumption above can be improved and actually holds. The equation $\Db r=r\overline{\Omega\tr\chib}/2$ follows from the variation of area:
$$\Db\int_{S_{\ub,u}}\D\mu_{\gs}=\int_{S_{\ub,u}}\Omega\tr\chib\D\mu_{\gs}.$$

We then go back to $\log\Omega$ and its derivatives. It is obvious now the conclusion about $\log\Omega$ is true because the auxillary bootstrap assumption is in fact true. We commute $\nablas$ three times to the equation $\Db\log\Omega=\omegab$ and then apply Gronwall type inequality to obtain the conclusion.
\end{proof}

The bootstrap assumption (1), the above lemma together with the bounds of $\Lambda(s)$, $\lambda(s)$ ensure the validity in the space-time of the basic inequalities, Lemma \ref{Sobolev}, Lemma \ref{poincare} and Lemma \ref{evolution} in Section \ref{basicinequalities}, which are used frequently. The constants will depend in addition on $\mathcal{O}_0$ if $\varepsilon$ is sufficiently small. We will not prove them here. They can be found in Chapter 4, 5 and 7 in \cite{Chr} or in \cite{Luk} (under a slightly stronger assumption).

Now we turn to the estimates for the connection coefficients. As the first step, we consider the structure equations for $\Db\Gamma$ where $\Gamma\in\{\chibh', \tr\chib', \etab, \omega\}$, which can be written in the following form:
\begin{align*}
\Db\Gamma_p=\underline{R}_{p'}+\sum\Gamma_{p_1}\Gamma_{p_2}
\end{align*}
where $\underline{R}\in\{\rho,\betab,\alphab\}$. Notice that in the above equations, $p\le\min\{p',p_1+p_2\}$.Then applying Gronwall type estimates, we have
\begin{align*}
\|r^p\Gamma_p\|_{L^\infty(S_{\ub,u})}&\lesssim\|r^p\Gamma_p\|_{L^\infty(S_{\ub,0})}+\int_0^u\|r^p \underline{R}_{p'}\|_{L^\infty(S_{\ub,u'})}\D u'+\int_0^u\|r^p \Gamma_{p_1}\Gamma_{p_2}\|_{L^\infty(S_{\ub,u'})}\D u'.
\end{align*}
Therefore, by applying Sobolev inequality for the curvature term,
\begin{align*}
\mathcal{O}^{0,\infty}[\chibh,\tr\chib,\etab,\omega]\lesssim\mathcal{O}_0+\varepsilon^{1/2}\underline{\mathcal{R}}+\varepsilon(\Delta_1)^2
\end{align*}
which implies that for $\varepsilon$ sufficiently small, $\mathcal{O}^{0,\infty}[\chibh,\tr\chib,\etab,\omega]\le C(\mathcal{O}_0)$.

We then commute $\nablas$ to the structure equations, for $i\le2$
\begin{align*}
\Db\nablas^i\Gamma_p=\nablas^i\underline{R}_{p'}+\sum\nablas^i(\Gamma_{p_1}\Gamma_{p_2})+\sum_{j=1}^i\nablas^j(\Omega\chib)\nablas^{i-j}\Gamma_p.
\end{align*}
The last term comes from the commutator $[\Db,\nablas^i]$, see Lemma \ref{commutator}. We estimate the nonlinear terms on the right in the following way: 
\begin{align*}
&\|r^p(r\nablas)(\Gamma_{p_1}\Gamma_{p_2})\|_{L^4(S_{\ub,u})}\\
\lesssim&\|r^{p_1}\Gamma_{p_1}\|_{L^\infty(S_{\ub,u})}\|r^{p_2}(r\nablas)\Gamma_{p_2}\|_{L^4(S_{\ub,u})}+\|r^{p_2}(r\nablas)\Gamma_{p_1}\|_{L^4(S_{\ub,u})}\|r^{p_2}\Gamma_{p_2}\|_{L^\infty(S_{\ub,u})}
\end{align*}
and
\begin{align*}
&\|r^p(r\nablas)^2(\Gamma_{p_1}\Gamma_{p_2})\|_{L^2(S_{\ub,u})}\\
\lesssim&\|r^{p_1}\Gamma_{p_1}\|_{L^\infty(S_{\ub,u})}\|r^{p_2}(r\nablas)^2\Gamma_{p_2}\|_{L^2(S_{\ub,u})}+\|r^{p_2}(r\nablas)^2\Gamma_{p_1}\|_{L^2(S_{\ub,u})}\|r^{p_2}\Gamma_{p_2}\|_{L^\infty(S_{\ub,u})}\\
&+\|r^{p_1}(r\nablas)\Gamma_{p_1}\|_{L^4(S_{\ub,u})}\|r^{p_2}(r\nablas)\Gamma_{p_2}\|_{L^4(S_{\ub,u})}.
\end{align*}
The term comes from commutator is treated in the same way. Therefore by applying evolution lemma, 
\begin{align*}
\mathcal{O}^{1,4}, \mathcal{O}^{2,2}[\chibh,\tr\chib,\etab,\omega]\lesssim\mathcal{O}_0+\varepsilon^{1/2}\underline{\mathcal{R}}
+\varepsilon(\mathcal{O}^{0,\infty}\mathcal{O}^{2,2}+(\mathcal{O}^{1,4})^2)\le\mathcal{O}_0+\varepsilon^{1/2}\underline{\mathcal{R}}+\varepsilon(\Delta_1)^2,
\end{align*}
choosing $\varepsilon$ sufficiently small yields $\mathcal{O}^{1,4}, \mathcal{O}^{2,2}[\chibh,\tr\chib,\etab,\omega]\le C(\mathcal{O}_0)$.

We then consider the structure equation for $\Db\Gamma$ for $\Gamma\in\{\Omega\tr\chi-\overline{\Omega\tr\chi},\Omega\chih,\Omega\tr\chi,\eta\}$, which can be written in the following form:
\begin{align*}
\Db\Gamma_p=R_{p'}+\nablas\Gamma''_{p''}+\sum\Gamma_{p_1}\Gamma_{p_2}
\end{align*}
where $\underline{R}\in\{\rho,\betab\}$, $\Gamma''\in\{\eta,\omegab\}$ and $p\le\min\{p',p'',p_1+p_2\}$. There are no essential differences in the estimate procedure. But notice that in order to estimate bounds of correct regularity, we need the bound $\|(r\nablas)^3(r^{1/2}\eta,\omegab)\|_{L^2(\Cb_{\ub})}$. Therefore, $\mathcal{O}[\Omega\tr\chi-\overline{\Omega\tr\chi},\Omega\chih,\Omega\tr\chi,\eta]\le C(\mathcal{O}_0)$.

\begin{remark}
We should make a remark here that the initial norms of the connection coefficients on $C_0$ are not exactly the same as in the assumptions of Theorem \ref{maintheorem}, because the optical function $u$ depends on which is the last slice and then the vectorfield $L'$ is not invariant on $C_0$ in the bootstrap argument. However, the vectorfield $L$ is invariant because the foliation on $C_0$ does not change. This ensures that the connection coefficients $\chib',\etab,\Omega\chi$ and the covariant derivative $\nablas$ do not change and $\omega$ keeps zero on $C_0$. And the difference between the second order derivatives of $\eta=-\etab+2\ds\log\Omega$ on $C_0$ and the one appears in the assumptions of Theorem \ref{maintheorem} is controlled by up to third order derivatives of $\log\Omega$ on $C_0$.
\end{remark}

Finally, we consider the structure equations for $D\omegab$:
\begin{align*}
D\omegab&=\Omega^2(2(\eta,\etab)-|\eta|^2-\rho).
\end{align*}

This equation should be integrated initiated from $\Cb_{\ub_*}$. Recall the canonical foliation equation on $\Cb_{\ub_*}$ \eqref{lastslice}. In view of the relation $\ds\log\Omega=\frac{1}{2}(\eta+\etab)$, it can be written in the following form:
\begin{align}\label{lasteta}
\divs\eta=\frac{1}{2}((\chih,\chibh)-\overline{(\chih,\chibh)})-(\rho-\overline{\rho}).
\end{align}
We commute $\Db$ with the equation for $\divs\eta$ above,
\begin{align*}
\divs(-\Omega(\chib\cdot\eta+\betab)+2\ds\omegab)-2\divs(\Omega\chibh\cdot\eta)-\Omega\tr\chib\divs\eta
=\Db(\frac{1}{2}((\chih,\chibh)-\overline{(\chih,\chibh)})-(\rho-\overline{\rho})).
\end{align*}
We plug in the equations for $\Db(\Omega\chih)$, $\Db\chibh'$ and $\Db\rho$, and denote $\check{\rho}=\rho-\frac{1}{2}(\chih,\chibh)$. We deduce the equation for $\Deltas\omegab$:
\begin{align}\label{lastomegab}
2\Deltas\omegab=2\divs(\Omega\betab)+\divs(3\Omega\chibh\cdot\eta+\frac{1}{2}\Omega\tr\chib\eta)+\Omega\tr\chib\divs\eta+(F-\overline{F}+\overline{\Omega\tr\chib\check\rho}-\overline{\Omega\tr\chib}\cdot\overline{\check{\rho}})
\end{align}
where
\begin{align*}
F=\frac{3}{2}\Omega\tr\chib\check{\rho}-(\ds\Omega,\betab)+\{(2\eta-\zeta,\betab)-\frac{1}{2}(\chibh,\nablas\tensor\eta+\eta\tensor\eta)+\frac{1}{4}\tr\chi|\chibh|^2\}.
\end{align*}
In addition, $\overline{\omegab}=-\overline{\Omega\tr\chib\log\Omega}$ by $\overline{\log\Omega}=0$.

We are going to estimate the norm $\|\nablas^2\omegab\|_{L^2(S_{\ub_*,u})}$ on the last slice in in terms of the norms on $C_0$ provided that $\varepsilon$ is sufficiently small. This is done by considering the equation \eqref{lastomegab}. We need the following lemma. 
\begin{lemma}\label{curvatureonS} For $i\le1$ and $\varepsilon$ sufficiently small,
$$\|(r\nablas)^i(r^2\rho, r^2\sigma, r\betab,rK)\|_{L^2(S_{\ub,u})}\le C(\mathcal{R}_0,\mathcal{O}_0).$$
\end{lemma}
\begin{proof}[Proof of Lemma \ref{curvatureonS}]
We consider the null Bianchi equations for $\Db\rho$, $\Db\sigma$ and $\Db\betab$ and commute $\nablas$ with them. We can estimate by evolution lemma and Sobolev inequalities
\begin{align*}
&\|(r\nablas)^i(r^2\rho, r^2\sigma, r\betab)\|_{L^2(S_{\ub_*,u})}\\\lesssim&\|(r\nablas)^i(r^2\rho,r^2\sigma, r\betab)\|_{L^2(S_{\ub_*,0})}+\mathcal{O}\int_0^u\sum_{i=0}^2\|(r\nablas)^i(r^2\rho,r^2\sigma, r\betab,\alphab)\|_{L^2(S_{\ub_*,u'})}\D u'\\
\lesssim&\mathcal{R}_0+\varepsilon^{1/2}\Delta_1\underline{\mathcal{R}}\lesssim C(\mathcal{R}_0).
\end{align*}
The estimate for $K$ comes from the Gauss equation $K+\frac{1}{4}\tr\chi\tr\chib-\frac{1}{2}(\chih,\chibh)=-\rho$.
\end{proof}

We compute
\begin{align*}
\int_{S_{\ub,u}}|\nablas\omegab|^2=-\int_{S_{\ub,u}}\Deltas\omegab(\omegab-\overline{\omegab})\lesssim&\|\Deltas\omegab\|_{L^2(S_{\ub,u})}\|\omegab-\overline{\omegab}\|_{L^2(S_{\ub,u})}\\
\lesssim&\|\Deltas\omegab\|_{L^2(S_{\ub,u})}\|\nablas\omegab\|_{L^2(S_{\ub,u})}.
\end{align*}
The last inequality holds by Poincar\'e ineqaulity. Therefore, by equation \eqref{lastomegab}
\begin{align*}
\|(r\nablas)\omegab\|_{L^2(S_{\ub_*,u})}+\|\omegab\|_{L^2(S_{\ub_*,u})}\lesssim\|r^2\Deltas\omegab\|_{L^2(S_{\ub_*,u})}+\|\overline{\omegab}\|_{L^2(S_{\ub_*,u})}\lesssim C(\mathcal{O}_0,\mathcal{R}_0).
\end{align*}

To derive estimate for $\nablas^2\omegab$, by elliptic estimate (Lemma \ref{curvatureonS} and \ref{elliptic}), we have
\begin{align*}
\|(r\nablas)^2\omegab\|_{L^2(S_{\ub_*,u})}\lesssim C(\mathcal{O}_0,\mathcal{R}_0).
\end{align*}
Then $\|r\omegab\|_{L^\infty(S_{\ub_*,u})}$ and $\|r^{1/2}(r\nablas)^2\omegab\|_{L^4(S_{\ub_*,u})}$ are bounded by Sobolev inequalities.

We apply Gronwall type estimates for $D\omegab$, we have
\begin{align*}
\|\omegab\|_{L^\infty(S_{\ub,u})}&\lesssim\|\omegab\|_{L^\infty(S_{\ub_*,u})}+\int_{0}^{\ub}\|2(\eta,\etab)-|\eta|^2-\rho\|_{L^\infty(S_{\ub',u})}\D \ub'\\
&\lesssim r^{-1}(\ub,u)C(\mathcal{O}_0,\mathcal{R}_0)+r^{-3/2}(\ub,u)C(\mathcal{O}_0,\mathcal{R}_0,\mathcal{R}).
\end{align*}
We will also commute $\nablas$ with the equation for $D\omegab$. By a similar argument, we finally obtain $\mathcal{O}[\omegab]\le C(\mathcal{O}_0,\mathcal{R})$. 
\end{proof}

\begin{proof}[Proof of Proposition \ref{connection2}]
In this proof we will appeal to the Hodge systems introduced in Lemma \ref{hodgesystems}. Notice that under the assumption of Proposition \ref{connection2}, the conclusion of Proposition \ref{connection1} holds and we will make use of it.

Now we proceed by considering each component.

\paragraph{{\bf Estimate for $\nablas^3\etab$}}We commute $\nablas$ twice with the equation for $\Db\mub$ and estimate the right hand side in suitable norm. The nonlinear term on the right hand side involving only connection coefficients (before commuting $\nablas$) are simply bounded by $\mathcal{O}$. The nonlinear term involving derivatives of connection coefficients or curvature components (before commuting $\nablas$) can be estimated by H\"older and Sobolev inequalities as
\begin{align*}
\sum_{i=0}^2\|\nablas^{2-i}(\Gamma\cdot\Psi)\|_{L^2(S_{\ub,u})}\lesssim\sum_{i\le2,j\le2}\|\nablas^i\Gamma\|_{L^2(S_{\ub,u})}\|\nablas^j\Psi\|_{L^2(S_{\ub,u})}
\end{align*}
where $\Psi$ refers to $\nablas\Gamma'$ or $R$, $\underline{R}$. The first factor is bounded by $\mathcal{O}$, and the second factor is bounded by $\mathcal{R}$, $\underline{\mathcal{R}}$ or bootstrap assumption (2) after integrating which is explained as follow.

Precisely, the right hand side of the equation $\Db\nablas^2\mub$ can be bounded with the help of the following three terms: we write
\begin{align*}
&I=r^3\|\nablas^2(-\Omega\tr\chib\mub-\frac{1}{2}\Omega\tr\chib\mu)\|_{L^2(S_{\ub,u})}\\
\lesssim&r^{-1}\mathcal{O}[\tr\chib]\sum_{i=0}^2\left(\|(r\nablas)^i(r\mub)\|_{L^2(S_{\ub,u})}+r^{-1/2}\|(r\nablas)^i(r^{3/2}\mu)\|_{L^2(S_{\ub,u})}\right)\\
&II=r^3\|\nablas^2(-\frac{1}{4}\Omega\tr\chi|\chibh|^2+\frac{1}{2}\Omega\tr\chib|\eta|^2)\|_{L^2(S_{\ub,u})}\\
\lesssim& r^{-1}C(\mathcal{O})\\
&III=r^3\|\nablas^2(\divs(2\Omega\chibh\cdot\etab-\Omega\tr\chib\eta))\|_{L^2(S_{\ub,u})}\\
\lesssim&r^{-2}\mathcal{O}[\etab]\|(r\nablas)^3(\Omega\chibh)\|_{L^2(S_{\ub,u})}+r^{-1}\mathcal{O}[\chibh]\|(r\nablas)^3\etab\|_{L^2(S_{\ub,u})}\\
&+r^{-3/2}\mathcal{O}[\tr\chib]\|(r\nablas)^3(r^{1/2}\eta)\|_{L^2(S_{\ub,u})}+r^{-2}\mathcal{O}[\eta]\|(r\nablas)^3(\Omega\tr\chib)\|_{L^2(S_{\ub,u})}+r^{-2}C(\mathcal{O}).
\end{align*}
Similarly, the commutator can be estimated
\begin{align*}
r^3\|[\Db,\nablas^2]\mub\|_{L^2(S_{\ub,u})}\lesssim r^{-1}\mathcal{O}[\chibh,\tr\chib]\sum_{i=1}^2\|(r\nablas)^i(r\mub)\|_{L^2(S_{\ub,u})}.
\end{align*}
Then we apply Gronwall type estimates to the equation for $\Db\nablas^2\mub$, 
\begin{align*}
&\|(r\nablas)^2(r\mub)\|_{L^2(S_{\ub,u})}\lesssim\|(r\nablas)^2(r\mub)\|_{L^2(S_{\ub,0})}\\
+&\int_0^u(I+II+III+r^3\|[\Db,\nablas^2]\mub\|_{L^2(S_{\ub,u'})})\D u'.
\end{align*}
We take the square of the above inequality and integrate along $\ub$ over $0$ to $\ub$. Therefore, we estimate, by H\"older inequality, 
\begin{align*}
&\int_0^{\ub}\left[\int_0^u (I+II+r^3\|[\Db,\nablas^2]\mub\|_{L^2(S_{\ub,u'})})\D u'\right]^2\D\ub'\\\lesssim&\mathcal{O}[\tr\chib,\chibh]^2\varepsilon\sum_{i=0}^2\left[\int_0^ur^{-2}\|(r\nablas)^i(r\mub)\|_{L^2(C_u')}^2\D u'+\int_0^{\ub}r^{-3}\|(r\nablas)^i(r^{3/2}\mu)\|_{L^2(\Cb_{\ub'})}^2\D\ub'\right]+\varepsilon^2r^{-1}C(\mathcal{O})\\
\lesssim&\varepsilon r^{-1}C(\mathcal{O},\mathcal{R}, \underline{\mathcal{R}},\Delta_2).
\end{align*}
$III$ is estimated in a similar way
\begin{align*}
&\int_0^{\ub}\left[\int_0^u III\D u'\right]^2\D\ub'\\\lesssim&\mathcal{O}[\etab]^2\varepsilon\int_{0}^{\ub}r^{-4}\|(r\nablas)^3(\Omega\chibh)\|_{L^2(\Cb_{\ub'})}^2\D\ub'+\mathcal{O}[\tr\chib]^2\varepsilon\int_{0}^{\ub}r^{-3}\|(r\nablas)^3(r^{1/2}\eta)\|_{L^2(\Cb_{\ub'})}^2\D\ub'\\
&+\mathcal{O}[\chibh]^2\varepsilon\int_0^ur^{-2}\|(r\nablas)^3\etab\|_{L^2(C_{u'})}^2\D u'+\mathcal{O}[\eta]^2\varepsilon\int_0^{\ub}\int_0^ur^{-4}\|(r\nablas)^3(\Omega\tr\chib)\|_{L^2(S_{\ub',u'})}^2\D u'\D\ub'\\&+\varepsilon r^{-3}C(\mathcal{O})\\
\lesssim&\varepsilon r^{-2}C(\mathcal{O},\mathcal{R}, \underline{\mathcal{R}},\Delta_2).
\end{align*}
Therefore, for $\varepsilon$ sufficiently small,
\begin{align*}
\|(r\nablas)^2(r\mub)\|_{L^2(C_u)}\lesssim C(\mathcal{O}_0, \mathcal{R}_0).
\end{align*}
By the div-curl system for $\etab$ we have
\begin{align*}
\|(r\nablas)^3\etab\|_{L^2(S_{\ub,u})}\lesssim\|(r\nablas)^2(r\rho)\|_{L^2(S_{\ub,u})}+\|(r\nablas)^2(r\mub)\|_{L^2(S_{\ub,u})}+r^{-1}C(\mathcal{O}).
\end{align*}
The last term comes from estimating the lower order terms by Proposition \ref{connection1}. Consequently
\begin{align*}
\|(r\nablas)^3\etab\|_{L^2(C_u)}\lesssim C(\mathcal{O}_0, \mathcal{R}, \mathcal{R}_0).
\end{align*}

The estimate for $\nablas^3\chib$ and $\nablas^3\omega$ are almost the same and we only sketch it. In fact, the estimates for $\nablas^3\etab$, $\nablas^3\chib$ and $\nablas^3\omegab$ can be done simultaneously.

\paragraph{{\bf Estimate for $\nablas^3\chibh$ and $\nablas^3\tr\chib$}}We commute $\nablas^3$ with the equations for $\Db(\Omega\tr\chib)$ and we estimate the right hand side as follows,
\begin{align*}
&\int_0^u \|(r\nablas)^3(-\frac{1}{2}(\Omega\tr\chib)^2-\frac{1}{2}|\Omega\chibh|^2+2\omegab(\Omega\tr\chib))+[\Db,\nablas^3](\Omega\tr\chib)\|_{L^2(S_{\ub,u'})}\D u'\\\lesssim&\mathcal{O}[\chibh,\tr\chib,\omegab]\left[\int_0^u\|(r\nablas)^3(\Omega\tr\chib)\|_{L^2(S_{\ub,u})}\D u'+\varepsilon^{1/2}(\|(r\nablas)^3(\Omega\chibh)\|_{L^2(\Cb_{\ub})}+\|(r\nablas)^3\omegab\|_{L^2(\Cb_{\ub})})\right]+\varepsilon C(\mathcal{O})\\
\lesssim&\varepsilon^{1/2} C(\mathcal{O}, \underline{\mathcal{R}},\Delta_2).
\end{align*}
Then for $\varepsilon$ sufficiently small, by Gronwall type estimates for the equation for $\Db\nablas^3(\Omega\tr\chib)$,
\begin{align*}\|(r\nablas)^3\tr\chib\|_{L^2(S_{\ub,u})}\lesssim C(\mathcal{O}_0).\end{align*}
By the $\divs$ equation for $\chibh$, we have
\begin{align*}
\|(r\nablas)^3(\Omega\chibh)\|_{L^2(S_{\ub,u})}\lesssim\|(r\nablas)^2(r\Omega\betab)\|_{L^2(S_{\ub,u})}+\|(r\nablas)^3(\Omega\tr\chib)\|_{L^2(S_{\ub,u})}+r^{-1}C(\mathcal{O}).
\end{align*}
Consequently,
\begin{align*}\|(r\nablas)^3(\Omega\chibh)\|_{L^2(\Cb_{\ub})}\lesssim C(\mathcal{O}_0,\underline{\mathcal{R}}).\end{align*}

\paragraph{{\bf Estimate for $\nablas^3\omega$}}We commute $\nablas$ with the equation for $\Db\kappa$, and estimate
\begin{align*}
&\int_0^{\ub}\left[\int_0^u r^3\|\nablas(-\Omega\tr\chib\kappa-2(\Omega\chibh,\nablas^2\omega)+m)\|_{L^2(S_{\ub',u'})}\D u'\right]^2\D\ub'\\
\lesssim&\mathcal{O}^2\varepsilon\left[\int_0^ur^{-2}[\sum_{i=0}^1\|(r\nablas)^i(r^3\kappa)\|_{L^2(C_u')}^2+r^2\|(r\nablas)^3(r\omega)\|_{L^2(C_u')}^2+r^2\|(r\nablas)^3\etab\|_{L^2(C_u')}^2+\mathcal{R}[\beta]^2]\D u'\right.\\
&\left.+\int_0^{\ub}r^{-2}[\|(r\nablas)^3(\Omega\chibh)\|_{L^2(\Cb_{\ub'})}^2+r\|(r\nablas)^3(r^{1/2}\eta)\|_{L^2(\Cb_{\ub'})}^2+\underline{\mathcal{R}}[\rho,\sigma,\betab]^2]\D\ub'\right]\\
\lesssim&\varepsilon C(\mathcal{O}, \mathcal{R},\underline{\mathcal{R}},\Delta_2)
\end{align*}
Notice that there are no terms arising from commutator $[\Db,\nablas]$ because we only commute $\nablas$ once and $\kappa$ is a function. Applying Gronwall type estimates, we have
\begin{align*}
\|(r\nablas)(r^3\kappa)\|_{L^2(C_u)}\lesssim C(\mathcal{O}_0,\mathcal{R})
\end{align*}
By the Laplacian equation for $\omega$, we have
\begin{align*}
\|(r\nablas)^3(r\omega)\|_{L^2(S_{\ub,u})}\lesssim\|(r\nablas)(r^3\kappa)\|_{L^2(S_{\ub,u})}+\|(r\nablas)^2(r^2\Omega\beta)\|_{L^2(S_{\ub,u})},
\end{align*}
Consequently, for $\varepsilon$ sufficiently small,
\begin{align}
\|(r\nablas)^3(r\omega)\|_{L^2(C_u)}\lesssim  C(\mathcal{O}_0,\mathcal{R}).
\end{align}

\begin{remark}
We also remark here that, $\chib'$, $\mub$ and $\kappa$ are the same as in the assumptions of Theorem \ref{maintheorem} in the construction, and then the initial norm of $\Omega\chib$ differs from the one in the assumptions of Theorem \ref{maintheorem} by up to third order derivatives of $\log\Omega$.
\end{remark}

\paragraph{{\bf Estimate for $\nablas^3\chih$ and $\nablas^3\tr\chi$}} The difference here to the above estimates is, we rely on Hodge systems coupled with propogation equations along $D$ direction instead of $\Db$ direction above. As in estimating $\mathcal{O}$, integrating along $L$ direction should start from last slice. Therefore, we should consider first the value of $\nablas(\Omega\tr\chi)$ on the last slice.

By the definition of canonical foliation, the equation for $\Db(\Omega\tr\chi)$ is reduced to
\begin{align*}
\Db(\Omega\tr\chi)+\frac{1}{2}\Omega\tr\chib\Omega\tr\chi&=\Omega^2(2|\eta|^2-\overline{(\chih,\chibh)}+2\overline{\rho}).
\end{align*}
The term $\divs\eta$ disappears and $\overline{\rho}$ will disappear if commuting $\nablas$ on the equation, that is,
\begin{align*}
\Db(\nablas^3(\Omega\tr\chi))=-\frac{1}{2}\nablas^3(\Omega\tr\chib\Omega\tr\chi)+2\nablas^3(\Omega^2|\eta|^2).
\end{align*}
The right hand side involves only third order derivatives of connection coefficients, and therefore can be controlled by $\mathcal{O}$ and $\Delta_2$ after integrating along $\Lb$. We have, by Gronwall type estimates,
\begin{align*}
&\|(r\nablas)^3(r\Omega\tr\chi)\|_{L^2(S_{\ub_*,u})}\\
\lesssim&\|(r\nablas)^3(r\Omega\tr\chi)\|_{L^2(S_{\ub_*,0})}\\
&+\int_0^u[\mathcal{O}[\tr\chi]\|(r\nablas)^3(\Omega\tr\chib)\|_{L^2(S_{\ub_*,u'})}+r^{-1}\mathcal{O}[\tr\chib]\|(r\nablas)^3(r\Omega\tr\chi)\|_{L^2(S_{\ub_*,u'})}\\&\quad +r^{-3/2}\mathcal{O}[\eta]\|(r\nablas)^3(r^{1/2}\eta)\|_{L^2(S_{\ub_*,u'})}]\D u'\\
\lesssim&C(\mathcal{O}_0)+\varepsilon^{1/2}\mathcal{O}(\Delta_2+C(\mathcal{O}_0))\lesssim C(\mathcal{O}_0)
\end{align*}
for $\varepsilon$ sufficiently small.

Now we commute $\nablas^3$ with the equation for $D(\Omega\tr\chi)$ and estimate the terms on the right as follows,
\begin{align*}
&\int_{\ub}^{\ub_*}r^4\|\nablas^3[-\frac{1}{2}(\Omega\tr\chi)^2-\frac{1}{2}|\Omega\chih|^2+2\omega(\Omega\tr\chi)]+\Omega\tr\chi\nablas^3(\Omega\tr\chi)+[D,\nablas^3](\Omega\tr\chi)\|_{L^2(S_{\ub',u})}\D\ub'\\
\lesssim&\mathcal{O}[\Omega\tr\chi-\overline{\Omega\tr\chi},\chih,\omega]\\&\times\int_{\ub}^{\ub_*}r^{-2}[\|(r\nablas)^3(r\Omega\tr\chi)\|_{L^2(S_{\ub',u})}+\|(r\nablas)^3(r\Omega\chih)\|_{L^2(S_{\ub',u})}+\|(r\nablas)^3(r\omega)\|_{L^2(S_{\ub',u})}+C(\mathcal{O})]\D\ub'\\
\lesssim&\mathcal{O}\int_{\ub}^{\ub_*}r^{-2}[\|(r\nablas)^3(r\Omega\tr\chi)\|_{L^2(S_{\ub',u})}+\|(r\nablas)^3(r\Omega\chih)\|_{L^2(S_{\ub',u})}]\D\ub'\\
&+\left(\int_{\ub}^{\ub_*}r^{-4}\D\ub'\right)^{1/2}\|(r\nablas)^3(r\omega)\|_{L^2(C_u)}+r^{-1}C(\mathcal{O}).
\end{align*}
By the $\divs$ equation for $\Omega\chih$, we have
\begin{align*}
\|(r\nablas)^3(\Omega\chih)\|_{L^2(S_{\ub,u})}\lesssim\|(r\nablas)^3(\Omega\tr\chi)\|_{L^2(S_{\ub,u})}+\|(r\nablas)^2(r\Omega\beta)\|_{L^2(S_{\ub,u})}+r^{-1}C(\mathcal{O}).
\end{align*}
Therefore, by applying Gronwall type estimates, combining the above inequality multiplied by $r$, we have
\begin{align*}
\|(r\nablas)^3(r\Omega\tr\chi)\|_{L^2(S_{\ub,u})}\lesssim&\|(r\nablas)^3(r\Omega\tr\chi)\|_{L^2(S_{\ub_*,u})}\\
\lesssim&\mathcal{O}\int_{\ub}^{\ub_*}r^{-2}\|(r\nablas)^3(r\Omega\tr\chi)\|_{L^2(S_{\ub',u})}\D\ub'\\
&+r^{-3/2}(\|(r\nablas)^3(r\omega)\|_{L^2(C_u)}+\|(r\nablas)^3(r\beta)\|_{L^2(C_u)})+r^{-1}C(\mathcal{O}).
\end{align*}
The first term on the right hand side can be absorbed by Gronwall's inequality, thank's to the factor $r^{-2}$. Finally, we have the estimate
\begin{align*}
\|(r\nablas)^3(r\Omega\tr\chi)\|_{L^2(S_{\ub,u})}\lesssim C(\mathcal{O}_0,\mathcal{R}),
\end{align*}
and consequently,
\begin{align*}
\|(r\nablas)^3(\Omega\chih)\|_{L^2(C_u)}\lesssim C(\mathcal{O}_0,\mathcal{R}).
\end{align*}

\paragraph{{\bf Estimate for $\nablas^3\eta$}} Notice that by definition of canonical foliation, $\mu=\overline{\mu}$ and hence $\nablas^2\mu=0$. Then we commute $\nablas^2$ with the equation for $D\mu$, and estimate the right hand side as follows,
\begin{align*}
&\int_0^u\left[\int_{\ub}^{\ub_*}r^3\|\nablas^2(-\Omega\tr\chi\mu-\frac{1}{2}\Omega\tr\chi\mub)+\Omega\tr\chi\nablas^2\mu+[D,\nablas^2]\mu\|_{L^2(S_{\ub',u'})}\D\ub'\right]^2\D u'\\
\lesssim&\int_0^u r^{-1}\mathcal{O}[\tr\chi]^2\sum_{i=0}^2\|(r\nablas)^i(r\mub)\|_{L^2(C_u')}^2\D u'\\
&+r^{-1}\mathcal{O}[\Omega\tr\chi-\overline{\Omega\tr\chi},\chih]^2\int_{\ub}^{\ub_*}r^{-1}\sum_{i=0}^1\|(r\nablas)^i(r\mu)\|_{L^2(\Cb_{\ub'})}^2\D \ub'\\
\lesssim& r^{-1}C(\mathcal{O},\mathcal{O}_0,\mathcal{R},\underline{\mathcal{R}}),\\
&\int_0^u\left[\int_{\ub}^{\ub_*}r^3\|\nablas^2(-\frac{1}{4}\Omega\tr\chib|\chih|^2+\frac{1}{2}\Omega\tr\chi|\etab|^2)\|_{L^2(S_{\ub',u'})}\D\ub'\right]^2\D u'\lesssim\varepsilon r^{-4}C(\mathcal{O}),\\
&\int_0^u\left[\int_{\ub}^{\ub_*}r^3\|\nablas^2(\divs(2\Omega\chih\cdot\eta-\Omega\tr\chi\etab))\|_{L^2(S_{\ub',u'})}\D\ub'\right]^2\D u'\\\lesssim&\mathcal{O}[\eta]^2\int_{0}^{u}r^{-3}\|(r\nablas)^3(\Omega\chih)\|_{L^2(C_{u'})}^2\D u'+\mathcal{O}[\tr\chi]^2\int_{0}^{u}r^{-1}\|(r\nablas)^3\etab\|_{L^2(C_{u'})}^2\D u'\\
&+\mathcal{O}[\chih]^2r^{-1}\int_{\ub}^{\ub_*}r^{-2}\|(r\nablas)^3\eta\|_{L^2(\Cb_{\ub'})}^2\D \ub'+\mathcal{O}[\etab]^2\int_0^{u}\left[\int_{\ub}^{\ub_*}r^{-3}\|(r\nablas)^3(r\Omega\tr\chi)\|_{L^2(S_{\ub',u'})}^2\D \ub'\right]^2\D u'\\&+\varepsilon r^{-4}C(\mathcal{O})\\
\lesssim& \mathcal{O}[\chih]^2r^{-1}\int_{\ub}^{\ub_*}r^{-2}\|(r\nablas)^3\eta\|_{L^2(\Cb_{\ub'})}^2\D \ub'+r^{-1} C(\mathcal{O}, \mathcal{O}_0, \mathcal{R}).
\end{align*}
We estimate $\|(r\nablas)^3\eta\|_{L^2(\Cb_{\ub'})}^2$ by the $\divs$-$\curls$ systems for $\eta$, 
\begin{align*}
\|(r\nablas)^3\eta\|_{L^2(\Cb_{\ub})}\lesssim\|(r\nablas)^2(r\rho)\|_{L^2(\Cb_{\ub})}+\|(r\nablas)^2(r\mu)\|_{L^2(\Cb_{\ub})}+\varepsilon^{1/2} r^{-1}C(\mathcal{O})
\end{align*}
Therefore, by applying Gronwall type estimates and integrating the square of the resulting inequality along $\Db$ direction,
\begin{align*}
\|(r\nablas)^2(r\mu)\|_{L^2(\Cb_{\ub})}^2\lesssim\mathcal{O}[\chih]^2r^{-1}\int_{\ub}^{\ub_*}r^{-2}\|(r\nablas)^3(r\mu)\|_{L^2(\Cb_{\ub'})}^2\D \ub'+r^{-1} C(\mathcal{O}, \mathcal{O}_*, \mathcal{O}_0, \mathcal{R},\underline{\mathcal{R}}).
\end{align*}
Then, by Gronwall's inequality, we have
\begin{align*}
\|(r\nablas)^2(r\mu)\|_{L^2(\Cb_{\ub})}^2\lesssim r^{-1}C(\mathcal{O}, \mathcal{O}_0, \mathcal{R},\underline{\mathcal{R}})
\end{align*}
and finally,
\begin{align*}
\|(r\nablas)^2(r^{1/2}\eta)\|_{L^2(\Cb_{\ub})}\lesssim C(\mathcal{O},  \mathcal{O}_0, \mathcal{R},\underline{\mathcal{R}}).
\end{align*} 

\paragraph{{\bf Estimate for $\nablas^3\omegab$}}At the last step, we consider $\nablas^3\omegab$. As above, we should consider first the value of $\nablas\kappa$ on the last slice. This can be directly seen from the equation \eqref{lastomegab}, which is, in terms of $\kappab$, on the last slice,
\begin{align*}
2\kappab=\divs(3\Omega\chibh\cdot\eta+\frac{1}{2}\Omega\tr\chib\eta)+\Omega\tr\chib\divs\eta+(F-\overline{F}+\overline{\Omega\tr\chib\check\rho}-\overline{\Omega\tr\chib}\cdot\overline{\check{\rho}}).
\end{align*}
It is easy to see that the right hand side involves only curvature components and first order derivatives of connection coefficients. By lemma \ref{curvatureonS} and Proposition \ref{connection1}, $\|(r\nablas)(r^2\kappab)\|_{L^2(S_{\ub_*,u})}\lesssim C(\mathcal{O}_0,\mathcal{R}_0)$.

We estimate the right hand side of the equation for $D\nablas\kappab$ which obtained by commute $\nablas$ with the equation for $D\kappab$ as follows.
\begin{align*}
&\int_0^u\left[\int_{\ub}^{\ub_*}r^2\|-\nablas(\Omega\tr\chi)\kappab-\nablas(2(\Omega\chih,\nablas^2\omegab)+\underline{m})\|_{L^2(S_{\ub',u'})}\D\ub'\right]^2\D u'\\
\lesssim&\mathcal{O}r^{-1}\int_{\ub}^{\ub^*}\left[r^{-2}(\sum_{i=0}^1\|(r\nablas)^i(r^2\kappab)\|_{L^2(\Cb_{\ub'})}^2+\|(r\nablas)^3\omegab\|_{L^2(\Cb_{\ub'})}^2\right.\\&\left.+\|(r\nablas^3)(r^{1/2})\eta\|_{L^2(S_{\ub,u})}+\underline{\mathcal{R}}[\rho,\sigma,\betab]^2)\right]\D u'\\
&+\int_0^ur^{-6}\|(r\nablas)^3\etab\|_{L^2(C_{u'})}^2\\
\lesssim&\mathcal{O}r^{-1}\int_{\ub}^{\ub^*}r^{-2}\|(r\nablas)^3\omegab\|_{L^2(\Cb_{\ub'})}^2\D\ub'+r^{-2}C(\mathcal{O},\mathcal{O}_0,\mathcal{R},\underline{\mathcal{R}}).
\end{align*}
By the Laplacian equation for $\omegab$, we have
\begin{align*}
\|(r\nablas)^3\omegab\|_{L^2(\Cb_{\ub})}\lesssim\|(r\nablas)(r^2\kappab)\|_{L^2(\Cb_{\ub})}+\|(r\nablas)^2(r\Omega\betab)\|_{L^2(\Cb_{\ub})}
\end{align*}
Applying Gronwall type estimates to the equation for $\Db\nablas\kappa$ leads to
\begin{align*}
\|r^2\nablas\kappab\|_{L^2(\Cb_{\ub})}^2\lesssim& \|r^2\nablas\kappab\|_{L^2(\Cb_{\ub_*})}^2\\&+r^{-1}\mathcal{O}\int_{\ub}^{\ub^*}r^{-2}\|(r\nablas)(r^2\kappab)\|_{L^2(\Cb_{\ub'})}^2\D\ub'+r^{-2}C(\mathcal{O},\mathcal{O}_0,\mathcal{R},\underline{\mathcal{R}}).
\end{align*}
We multiply the above inequality by $r$ and apply Gronwall's inequality for $\|r^{5/2}\nablas\kappab\|_{L^2(\Cb_{\ub})}$. Notice that the factor in the integral becomes $r^{-3/2}$ and is integrable, therefore
\begin{align*}
\|r^{5/2}\nablas\kappab\|_{L^2(\Cb_{\ub})}^2 \lesssim r^{-1}C(\mathcal{O},\mathcal{O}_0,\mathcal{R},\underline{\mathcal{R}})
\end{align*}
and therefore,
\begin{align*}
\|(r\nablas)^3\omegab\|_{L^2(\Cb_{\ub})}\lesssim C(\mathcal{O},\mathcal{O}_0,\mathcal{R},\underline{\mathcal{R}}).
\end{align*}

At last, we choose $\Delta_2$ sufficiently large to complete the proof.\end{proof}

\subsection{Estimates for Curvature}\label{SectionCurvature}

\begin{comment}
We write $\check{\rho}=\rho-\frac{1}{2}(\chih,\chibh)$ and $\check{\sigma}=\sigma+\frac{1}{2}\chih\wedge\chibh$ in order to elimilate the $\alpha$ in null Bianchi equations. We compute
\begin{align*}
D(\chih,\chibh)=-4(\Omega\chih,\chih\cdot\chibh)-2\Omega\tr\chi(\chih,\chibh)+(D\chih',\Omega\chibh)+(\chih',D(\Omega\chibh)).
\end{align*}
Plug in the two equations for $D\chih'$ and $D(\Omega\chibh)$, we can write the above equations as
\begin{align*}
D(\chih,\chibh)+\frac{3}{2}\Omega\tr\chi(\chih,\chibh)+(\Omega\chibh,\alpha)=\Omega\chih(\nablas\etab+\etab^2)+\chih^2(\Omega\chibh+\Omega\tr\chib)
\end{align*}
where the coefficients on the left hand side are precise but those on the right are not needed. Then we can deduce that
\end{comment}

We are going to complete Step 3 of the proof in this subsection. We will prove in this subsection the following proposition:
\begin{proposition}\label{curvaturecompletenullcone}
If $\varepsilon$ is sufficiently small depending on $\mathcal{O}_0$, $\mathcal{R}_0$, $\underline{\mathcal{R}}_0$, then we have
\begin{align*}
\mathcal{R}+\underline{\mathcal{R}}\le C(\mathcal{O}_0,\mathcal{R}_0,\underline{\mathcal{R}}_0).
\end{align*}
\end{proposition}
\begin{proof}
As discussed in the introduction, we adapt the renormalization in \cite{L-R1} in the current work. We write $\check{\rho}=\rho-\frac{1}{2}(\chih,\chibh)$ and $\check{\sigma}=\sigma+\frac{1}{2}\chih\wedge\chibh$. By direct computation, they satisfy the following renormalized Bianchi equations
\begin{align*}
D\check{\rho}+\frac{3}{2}\Omega\tr\chi\check{\rho}&=\Omega\{\divs\beta+(2\etab+\zeta,\beta)-\frac{1}{2}(\chih,\nablas\tensor\etab+\etab\tensor\etab)+\frac{1}{4}\tr\chib|\chih|^2\}\\
D\check{\sigma}+\frac{3}{2}\Omega\tr\chi\check{\sigma}&=\Omega\{\curls\beta+(2\etab+\zeta,^*\beta)+\frac{1}{2}\chih\wedge(\nablas\tensor\etab+\etab\tensor\etab)\}.
\end{align*}
 This caculation is similar to that in deriving the equation \eqref{lastomegab}. To couple with these two equations, the equation for $\Db\beta$ should also rewrite as
\begin{align*}
&\Db\beta+\frac{1}{2}\Omega\tr\chib\beta\\=&\Omega\chibh \cdot \beta-\omegab \beta+\Omega\{\ds \check{\rho}+{}^*\ds \check{\sigma}+3\eta\rho+3{}^*\eta\sigma+2\chih\cdot\betab+\frac{1}{2}(\ds(\chih,\chibh)-^*\ds(\chih\wedge\chibh))\}.
\end{align*}

To make the expression in a uniform way, in this subsection, we will denote
$$\underline{R}_1=\frac{1}{\sqrt{2}}\alphab,\ \underline{R}_2=-\betab,\ \underline{R}_3=(-\check{\rho},\check{\sigma})$$
and denote
$$R_2=-\betab,\ R_3=(\rho,\sigma),\ R_4=\beta.$$
Comparing to our usual convention in previous sections, the definition is essentially the same up to a constant multiple.

Now define $\mathfrak{D}$ acting on $\underline{R}_j$, say $(-\check{\rho},\check{\sigma}),-\betab,\alphab$ and their derivatives respectively to be
\begin{align*}
\mathfrak{D}\nablas^i(-\check{\rho},\check{\sigma})&=D\nablas^i(-\check{\rho},\check{\sigma})+\frac{1}{2}\Omega\tr\chi\nablas^i(-\check{\rho},\check{\sigma}),\\
\mathfrak{D}\nablas^i(-\betab)&=D\nablas^i(-\betab)-\Omega\chih\cdot\nablas^i(-\betab),\\
\mathfrak{D}\nablas^i\alphab&=\Dh\nablas^i\alphab-\frac{1}{2}\Omega\tr\chi\nablas^i\alphab.
\end{align*}
We also denote $\mathcal{D}$ to be one of the Hodge operators $\mathcal{D}_1$, $\mathcal{D}_2$, ${}^*\mathcal{D}_1$ and ${}^*\mathcal{D}_2$ (see the definition before Lemma \ref{commutator}). We also define $\mathcal{D}$ to be acting on derivatives of the corresponding tensorfields in a obvious way. Then we denote $^*\mathcal{D}$ to be the corresponding $L^2(S_{\ub,u})$ formal adjoint. \footnote{We abuse the notation here as compared to Lemma \ref{commutator}.}

We collect three groups of null Bianchi equations, say $\Db\beta$-$D\check{\rho}$-$D\check{\sigma}$, $\Db\rho$-$\Db\sigma$-$D\beta$ and $\Db\betab$-$\Dh\alphab$. By the commutation formula for $[D,\nablas^i]$, $[\Db,\nablas^i]$, $[\mathcal{D},\nablas^i]$ and $[^*\mathcal{D},\nablas^i]$ (see Lemma \ref{commutator}), we can write them in the following form. For $j=1,2,3$, denoting $\nablas^{-1}K=0$, 
\begin{align*}
E^3_{j,i}\triangleq&\Db \nablas^iR_{j+1}-\Omega\mathcal{D}\nablas^i\underline{R}_j\\
&=\sum\sum_{k=0}^i\left(\nablas^k\Gamma_p\nablas^{i-k}R_{j+1}+\nablas^k(\Omega\Gamma_{p_1})\nablas^{i-k}\underline{R}_{p_2}+\nablas^{k-1}K\nablas^{i-k}\underline{R}_j\right),\\
E^4_{j,i}\triangleq&\mathfrak{D}\nablas^i \underline{R}_j+\frac{j-1}{2}\Omega\tr\chi \nablas^i\underline{R}_j+\Omega^*\mathcal{D}\nablas^iR_{j+1}\\
&=\sum\sum_{k=0}^i\left(\nablas^k\Gamma_p\nablas^{i-k}\underline{R}_j+\nablas^k(\Omega\Gamma_{p_1})\nablas^{i-k}R_{p_2}+\nablas^i(\Omega\Gamma_{p'_1})\nablas^{i-k}\nablas\Gamma_{p'_2}\right.\\
&\qquad\left.+\nablas^k(\Omega\Gamma_{p''_1})\nablas^{i-k}(\Gamma_{p''_2}\Gamma_{p''_3})+\nablas^{k-1}K\nablas^{i-k}R_{j+1}\right).
\end{align*}

With the null Bianchi equations written in the above form, we can refer to the exact form of the null Bianchi equations by direct counting to obtain the following lemma:
\begin{lemma}\label{counting}
In the first equation, $j+1\le p_1+p_2$, and by definition of the numbers assigned, $p\ge1$ automatically. In the second equation, we have $p\ge2$, $j+2\le\min\{p_1+p_2,p'_1+p'_2+1,p''_1+p''_2+p''_3\}$.
\end{lemma}

We then compute
\begin{align*}
&\Db(r^{2(i+j-1)}|\nablas^iR_{j+1}|^2\D\mu_{\gs})\\
=&r^{2(i+j-1)}(\Gamma[\tr\chib,\chibh]|\nablas^iR_{j+1}|^2+2(\nablas^iR_{j+1},\Db\nablas^i R_{j+1}))\\
=&r^{2(i+j-1)}(\Gamma[\tr\chib,\chibh]|\nablas^iR_{j+1}|^2+2(\nablas^iR_{j+1},\Omega\mathcal{D}\nablas^i\underline{R}_j+E_{j,i}^3)).
\end{align*}
We also compute carefully
\begin{align*}
&D(r^{2(i+j-1)}|\nablas^i\underline{R}_j|^2\D\mu_{\gs})\\
=&(2(i+j-1)r^{2(i+j)-2}Dr|\nablas^i\underline{R}_j|^2-r^{2(i+j-1)}i\Omega\tr\chi|\nablas^i\underline{R}_j|^2\\
&+2r^{2(i+j-1)}(\nablas^i\underline{R}_j,\mathfrak{D}\underline{R}_i)+r^{2(i+j-1)}\Omega\chih\cdot\nablas^i\underline{R}_j\cdot\nablas^i\underline{R}_j)\D\mu_{\gs}\\
=&r^{2(i+j-1)}\left((i+j-1)(\Omega\tr\chi-\overline{\Omega\tr\chi})|\nablas^i\underline{R}_j|^2-2(\nablas^i\underline{R}_j,\Omega^*\mathcal{D}\nablas^iR_{j+1}+\Omega\chih\cdot\nablas^i R_j+E^4_{j,i})\right).
\end{align*}
Suming up the above two identities, then integrating on the spacetime manifold $M_{\ub,u}$ which is enclosed by $C_0$, $\Cb_0$, $\Cb_{\ub}$ and $C_{u}$, we have
\begin{align*}
&\int_{C_u}r^{2(i+j-1)}|\nablas^iR_{j+1}|^2+\int_{\Cb_{\ub}}r^{2(i+j-1)}|\nablas^i\underline{R}_j|^2\\
=&\int_{C_0}r^{2(i+j-1)}|\nablas^iR_{j+1}|^2+\int_{\Cb_0}r^{2(i+j-1)}|\nablas^i\underline{R}_j|^2+\iint_{M_{\ub,u}}2\Omega^{-2}\tau_j^{(i)}\\
\lesssim&\int_{C_0}r^{2(i+j-1)}|\nablas^iR_{j+1}|^2+\int_{\Cb_0}r^{2(i+j-1)}|\nablas^i\underline{R}_j|^2+\iint_M|\tau_j^{(i)}|
\end{align*}
where $M=M_{\ub_*,\varepsilon}$ and
\begin{align*}
\tau_{j}^{(i)}=&r^{2(i+j-1)}\left(\sum\sum_{k=0}^i\left(\nablas^k\Gamma_p\nablas^{i-k}R_{j+1}+\nablas^k(\Omega\Gamma_{p_1})\nablas^{i-k}\underline{R}_{p_2}+\nablas^{k-1}K\nablas^{i-k}\underline{R}_j\right),\nablas^i R_{j+1}\right)\\
&+r^{2(i+j-1)}\left(\sum\sum_{k=0}^i\left(\nablas^k\Gamma_p\nablas^{i-k}\underline{R}_j+\nablas^k(\Omega\Gamma_{p_1})\nablas^{i-k}R_{p_2}+\nablas^i(\Omega\Gamma_{p'_1})\nablas^{i-k}\nablas\Gamma_{p'_2}\right.\right.\\
&\qquad\left.\left.+\nablas^k(\Omega\Gamma_{p''_1})\nablas^{i-k}(\Gamma_{p''_2}\Gamma_{p''_3})+\nablas^{k-1}K\nablas^{i-k}R_{j+1}\right),\nablas^i\underline{R}_j\right)
\end{align*}
with the numbers assigned to every component satisfy Lemma \ref{counting}. Here we abuse the notations, that the $p$, $p_1$, $p_2$ in the first line are not the same as in the second line.

\begin{comment}
By commuting at most twice $\nablas$ with the above equations and integration by parts, 
\begin{align*}
&\int_{C_u}r^{4+2i}|\nablas^i\beta|^2+\int_{\Cb_{\ub}}r^{4+2i}(|\nablas^i\check{\rho}|^2+|\nablas^i\check{\sigma}|^2)\\
=&\int_{C_0}r^{4+2i}|\nablas^i\beta|^2+\int_{\Cb_0}r^{4+2i}(|\nablas^i\check{\rho}|^2+|\nablas^i\check{\sigma}|^2)+\iint_M\tau_3^{(i)}.
\end{align*}

We also do the same thing to the group $\Db\rho$-$\Db\sigma$-$D\betab$ of null Bianchi equations to obtain
\begin{align*}
&\int_{C_u}r^{2+2i}(|\nablas^i\rho|^2+|\nablas^i\sigma|^2)+\int_{\Cb_{\ub}}r^{2+2i}|\nablas^i\betab|^2\\
=&\int_{C_0}r^{2+2i}(|\nablas^i\rho|^2+|\nablas^i\sigma|^2)+\int_{\Cb_0}r^{2+2i}|\nablas^i\betab|^2+\iint_M\tau_2^{(i)}
\end{align*}
and to the group $\Db\beta$-$D\alphab$,
\begin{align*}
&\int_{C_u}r^{2i}|\nablas^i\betab|^2+\int_{\Cb_{\ub}}r^{2i}|\nablas^i\alphab|^2=\int_{C_0}r^{2i}|\nablas^i\betab|^2+\int_{\Cb_0}r^{2i}|\nablas^i\alphab|^2+\iint_M\tau_1^{(i)}.
\end{align*}
Notice that we do not need to renormalize $\rho$ and $\sigma$ in the second equality. 
\end{comment}

We divide terms contained in $\tau_j^{(i)}$ in different types.

\begin{itemize}

\item[\textbf{ Type 1}] Terms like $r^{2(i+j-1)}(\nablas^iR_{j+1}\nablas^{i_1}\Gamma_{p_1}\nablas^{i_2}\underline{R}_{p_2})$ where $i_1+i_2= i\le2$. We have $j+1\le p_1+p_2$, therefore we estimate by H\"older inequality,
\begin{align*}
&\iint_Mr^{2(i+j-1)}|\nablas^iR_{j+1}\nablas^{i_1}\Gamma_{p_1}\nablas^{i_2}\underline{R}_{p_2}|\\
\lesssim&\left(\int_0^u\int_{C_{u'}}r^{2(i+j-1)}|\nablas^iR_{j+1}|^2\D u'\right)^{1/2}\\
&\times\left(\int_0^{\ub}r^{-2}\int_{\Cb_{\ub'}}r^{2(i_1+i_2+p_1+p_2-1)}|\nablas^{i_1}\Gamma_{p_1}\nablas^{i_2}\underline{R}_{p_2}|^2\D \ub'\right)^{1/2}\\
\triangleq& I\cdot II\lesssim\varepsilon^{1/2}\mathcal{R}\cdot II.
\end{align*}

For $i_1=0$,
\begin{align*}
II^2\lesssim\sup_{\ub,u}(\|r^{i_1+p_1}\nablas^{i_1}\Gamma_{p_1}\|^2_{L^\infty(S_{\ub,u})})\int_0^{\ub}r^{-2}\int_{\Cb_{\ub'}}r^{2(i_2+p_2-1)}|\nablas^{i_2}\underline{R}_{p_2}|^2\lesssim C(\mathcal{O}_0)\underline{\mathcal{R}}^2.
\end{align*}

For $i_1=1$, then $i_2\le 1$, applying Sobolev inequality,
\begin{align*}
II^2\lesssim&\sup_{\ub,u}(r^{-1/2}\|r^{i_1+p_1}\nablas^{i_1}\Gamma_{p_1}\|^2_{L^4(S_{\ub,u})})\\&\times\int_0^{\ub}r^{-2}\D\ub'\int_0^u\D u'\|r^{(i_2+p_2-1)}\nablas^{i_2}\underline{R}_{p_2}\|^2_{L^4(S_{\ub',u'})}\\
\lesssim&\sup_{\ub,u}(r^{-1/2}\|r^{i_1+p_1}\nablas^{i_1}\Gamma_{p_1}\|^2_{L^4(S_{\ub,u})})\\&\times\int_0^{\ub}r^{-2}\D\ub'\sum_{k=i_2}^{i_2+1}\int_{\Cb_{\ub'}}|r^{2(k+p_2-1)}\nablas^{k}\underline{R}_{p_2}|^2\lesssim C(\mathcal{O}_0)\underline{\mathcal{R}}^2.
\end{align*}

For $i_1=2$, then $i_2=0$, we also apply H\"older and Sobolev inequalities, 
\begin{align*}
II^2\lesssim&\sup_{\ub,u}(r^{-1}\|r^{i_1+p_1}\nablas^{i_1}\Gamma_{p_1}\|^2_{L^2(S_{\ub,u})})\\
&\times\int_0^{\ub}r^{-2}\D\ub'\sum_{k=i_2}^{i_2+2}\int_{\Cb_{\ub'}}|r^{2(k+p_2-1)}\nablas^{k}\underline{R}_{p_2}|^2\lesssim C(\mathcal{O}_0)\underline{\mathcal{R}}^2.
\end{align*}

There are also terms like $r^{2(i+j-1)}(\nablas^{i_1}\Gamma_{p_1}\nablas^{i_2}R_{p_2}\nablas^i\underline{R}_j)$. We estimate
\begin{align*}
&\iint_Mr^{2(i+j-1)}|\nablas^{i_1}\Gamma_{p_1}\nablas^{i_2}R_{p_2}\nablas^i\underline{R}_j|\\
\lesssim&\left(\int_0^{\ub}r^{-2}\int_{\Cb_{\ub'}}r^{2(i+j-1)}|\nablas^i\underline{R}_{j}|^2\D \ub'\right)^{1/2}\\
&\times\left(\int_0^{u}\int_{C_{u'}}r^{2(i_1+i_2+p_1+p_2-2)}|\nablas^{i_1}\Gamma_{p_1}\nablas^{i_2}R_{p_2}|^2\D u'\right)^{1/2}\\
\lesssim&\varepsilon^{1/2}C(\mathcal{O}_0)\mathcal{R}\underline{\mathcal{R}}.
\end{align*}
because we have $j+2\le p_1+p_2$.

\item[\textbf{Type 2}] Terms like $r^{2(i+j-1)}(\nablas^{i_1}\Gamma_p\nablas^{i_2}R_{j+1}\nablas^iR_{j+1})$ and $r^{2(i+j-1)}(\nablas^{i_1}\Gamma_p\nablas^{i_2}\underline{R}_{j}\nablas^i\underline{R}_{j})$ with $i_1+i_2=i\le2$. Because $p\ge0$ in the first case, we estimate
\begin{align*}
&\iint_Mr^{2(i+j-1)}|\nablas^{i_1}\Gamma_p\nablas^{i_2}R_{j+1}\nablas^iR_{j+1}|\\
\lesssim&\sup_{\ub,u}(r^{-2/q}\|(r\nablas)^{i_q}\Gamma\|_{L^q(S_{\ub,u})})\sum_{k=0}^2\int_0^u\int_{C_{u'}}r^{2(k+j-1)}|\nablas^k\beta|^2\D u'\\
\lesssim&\varepsilon C(\mathcal{O}_0,\mathcal{R})\mathcal{R}^2
\end{align*}
where $i_q=2,1,0$ for $q=2,4,\infty$ respectively. Notice that $\omegab$ involves in this case and $\mathcal{O}$ depends on $\mathcal{R}$.

For the second case, only $\chih$, $\Omega\tr\chi-\overline{\Omega\tr\chi}$ and $\omega$ involve as a $\Gamma$, Therefore, we estimate
\begin{align*}
&\iint_Mr^{2(i+j-1)}|(\nablas^{i_1}\Gamma_p\nablas^{i_2}\underline{R}_{j}\nablas^i\underline{R}_{j})|\\
\lesssim&\sup_{\ub,u}(r^{-2/q}\|r^2(r\nablas)^{i_q}\Gamma\|_{L^q(S_{\ub,u})})\sum_{k=0}^2\int_0^{\ub'}r^{-2}\int_{\Cb_{\ub'}}r^{2(k+j-1)}|\nablas^k\underline{R}_j|^2\D\ub'\\
\lesssim&C(\mathcal{O}_0)\sum_{k=0}^2\int_0^{\ub'}r^{-2}\int_{\Cb_{\ub'}}r^{2(k+j-1)}|\nablas^k\underline{R}_j|^2\D\ub'.
\end{align*}

\item[\textbf{Type 3}]

Terms like $r^{2(i+j-1)} (\nablas^{i_1}(\Gamma_{p'_1})\nablas^{i_2}\nablas\Gamma_{p'_2}\nablas^i\underline{R}_j)$ and $r^{2(i+j-1)} (\nablas^{i_1}(\Gamma_{p''_1})\nablas^{i_2}(\Gamma_{p''_2}\Gamma_{p''_3})\nablas^i\underline{R}_j)$ for $i_1+i_2=i\le2$. This type which arises from renormalization only appears when $j=3$. For the first case, we have $j+2\le p'_1+p'_2+1$ and $\Gamma_{p'_2}$ contains only $\etab$. Therefore the estimate can be done as the second case of Type 1. For the second case, notice that the connection coefficients $\Gamma$ and $\Gamma^2$ have the same regularity. That is, $\Gamma^2$ can be estimated in $L^\infty(S)$, $\nablas(\Gamma^2)=\Gamma\nablas\Gamma$ can be estimated in $L^4(S)$, and $\nablas^2(\Gamma^2)=\Gamma\nablas^2\Gamma+(\nablas\Gamma)^2$ can be estimated in $L^2(S)$. Therefore, $\nablas^{i_2}(\Gamma_{p''_2}\Gamma_{p''_3})$ can always estimated in different norms on $S_{\ub,u}$ and get addition $\varepsilon^{1/2}$ factor after integrating on $M$.

\item[\textbf{Type 4}] Terms like $r^{2(i+j-1)}(\nablas^{i_1}K\nablas^{i_2}\underline{R}_j\nablas^iR_{j+1})$ and $r^{2(i+j-1)}(\nablas^{i_1}K\nablas^{i_2}R_{j+1}\nablas^i\underline{R}_{j})$ for $i_1+i_2=i-1\le1$. By Lemma \ref{curvatureonS}, $\|(r\nablas)^{i_1}(rK)\|_{L^2(S_{\ub,u})}\le C(\mathcal{R}_0,\mathcal{O}_0)$ and hence $\|r^{3/2}K\|_{L^4(S_{\ub,u})}\le C(\mathcal{R}_0,\mathcal{O}_0)$. Therefore,
\begin{align*}
&\iint_Mr^{2(i+j-1)}|\nablas^{i_1}K\nablas^{i_2}\underline{R}_j\nablas^iR_{j+1}|\\
\lesssim&\left(\int_0^u\int_{C_{u'}}r^{2(i+j-1)}|\nablas^iR_{j+1}|^2\D u'\right)^{1/2}\\
&\times\left(\int_0^{\ub}r^{-2}\int_{\Cb_{\ub'}}r^{2(i_1+i_2+j+1)}|\nablas^{i_1}K\nablas^{i_2}\underline{R}_{j}|^2\D \ub'\right)^{1/2}\\
\lesssim&\varepsilon^{1/2}C(\mathcal{R}_0,\mathcal{O}_0)\underline{\mathcal{R}}\mathcal{R}
\end{align*}
and the second case can be treated similarly.

\end{itemize}

Then, by choosing $\varepsilon$ sufficiently small,
\begin{align*}
&\sum_{i=0}^2\left(\int_{C_u}r^{2(i+j-1)}|\nablas^iR_{j+1}|^2+\int_{\Cb_{\ub}}r^{2(i+j-1)}|\nablas^i\underline{R}_j|^2\right)\\
\le&\sum_{i=0}^2\left(\int_{C_0}r^{2(i+j-1)}|\nablas^iR_{j+1}|^2+\int_{\Cb_0}r^{2(i+j-1)}|\nablas^i\underline{R}_j|^2\right)\\&+C(\mathcal{O}_0)\sum_{k=0}^2\int_0^{\ub'}r^{-2}\int_{\Cb_{\ub'}}r^{2(k+j-1)}|\nablas^k\underline{R}_j|^2\D\ub'+1.
\end{align*}
Because the factor $r^{-2}$ is integrable and $\mathcal{O}_0$ is independent of $\underline{\mathcal{R}}$, then we can apply Gronwall inequality to prove the proposition.
\end{proof}
\begin{remark}
The function $u$ restricted on $\Cb_0$ does not coincide with $\underline{s}$. Therefore, in the above estimate, the terms
\begin{align*}
\int_{\Cb_{0}}r^{2(i+j-1)}|\nablas^i\underline{R}_j|^2
\end{align*}
are not exactly the same as in $\underline{\mathcal{R}}_0$. But the differences are controlled by $\mathcal{O}_0,\mathcal{R}_0$ for $\varepsilon$ sufficiently small. Also, the terms
$$\int_{C_0}r^{2(i+j-1)}|\nablas^iR_{j+1}|^2$$
differ from the norms in $\mathcal{R}_0$ by up to second order derivatives of $\Omega$.
\end{remark}

\subsection{Canonical Last Slice}\label{SectionLastSlice}

We will carry out the last step of the proof in this subsection. As introduced above, we first extend the solution to a larger region $M_\delta$ which corresponds to $0\le\ub\le\ub_*+\delta$, $0\le u\le\varepsilon+\delta'$. First of all, $\delta$ and $\delta'$ should be chosen such that all bounds derived in Section \ref{SectionConnection} and Section \ref{SectionCurvature} are bounded by twice of their bounds on $M$. This is ensured simply by continuity because $\delta$ and $\delta'$ are allowed to depend on $\varepsilon$ and $\ub_*$. In the remaining part of this subsection, we will fix $\delta'$ and then choose $\delta$ sufficiently small such that we can construct a new canonical function $u_\delta$ varying from $0$ to $\varepsilon$.

At this point, we hope to apply Proposition \ref{canonicalslice} to the new last slice $\Cb_{\ub_*+\delta}$. Different from the initial slice $\Cb_0$, we should notice that, first, we must estimate the bounds for  $\underline{\mathcal{R}}[\Db\alphab,\Db^2\alphab],\mathcal{O}^2[\Db\omegab],\mathcal{O}[\Db^2\omegab]$ which are not estimated yet, second, on the null cone  $\Cb_{\ub_*+\delta}$, we have less room as compared to $\Cb_0$, to solve the canonical slice equation, because the background foliation $u$ which is canonical on $\Cb_{\ub_*}$ only varies from $[0,\varepsilon+\delta']$, and this is exactly the reason why we need the $\delta'$ to be positive.

So, we will prove the following as the first step:
\begin{proposition}\label{DbDbalphab}If $\varepsilon$ is sufficiently small depending on $\mathcal{O}_0$, $\mathcal{R}_0$, $\underline{\mathcal{R}}_0$, then we have
\begin{align*}
\underline{\mathcal{R}}[\Db\alphab,\Db^2\alphab],\mathcal{O}^2[\Db\omegab],\mathcal{O}[\Db^2\omegab]\le C(\mathcal{O}_0,\mathcal{R}_0,\underline{\mathcal{R}}_0).
\end{align*}
\end{proposition}
\begin{remark}
This proposition is proven in the origin region $M_{\ub_*,\varepsilon}$. We can then choose $\delta$ and $\delta'$ sufficiently small such that this proposition still holds in the region $M_{\ub_*+\delta,\varepsilon+\delta'}$ with a twice large bound.
\end{remark}

\begin{proof}For the norms about $\alphab$, we rely on the null Bianchi equations for $\Db\betab$ and $\Dh\alphab$. We first commute $\Db$ with the Bianchi equations:
\begin{align*}
&\Dh\Dbh\alphab-\frac{1}{2}\Omega\tr\chi\Dbh\alphab+\Omega\nablas\tensor\Db\betab\\
=&\frac{1}{2}\Db(\Omega\tr\chi)\cdot\alphab-2\omega\Dbh\alphab-2\Db\omega\cdot\alphab+[\Dh,\Dbh]\alphab-[\Dbh,\nablas\tensor]\betab\\ &-\Omega\omegab\nablas\tensor\betab-\Db\left(\Omega\{(4\etab-\zeta)\tensor \betab+3\chibh \rho-3{}^*\chibh \sigma\}\right),\\
&\Db\Db\betab-\Omega\chibh\cdot \Db\betab+\Omega\divs\Dbh\alphab\\=&\Db(\Omega\chibh)\cdot\betab-[\Dbh,\divs]\alphab-\Omega\omegab\divs\alphab-\Db\left(\frac{3}{2}\Omega\tr\chib\betab-\omegab\betab+\Omega\{(\eta-2\zeta)\cdot\alphab\}\right).\\
\end{align*}
As what we have done in energy estimates, we can establish the following inequality:
\begin{align*}
\int_{C_u}|\Db\betab|^2+\int_{\Cb_{\ub}}|\Db\alphab|^2\lesssim\int_{C_0}|\Db\betab|^2+\int_{\Cb_0}|\Db\alphab|^2+\int_M|\tau_1^4|.
\end{align*}
The term $\tau_1^4$ comes from $\Db\alphab$ times the right hand side of the first equation, plus $\Db\betab$ times the right hand side of the second equation. By commuting one more $\nablas$, we also have
\begin{align*}
\int_{C_u}|(r\nablas)\Db\betab|^2+\int_{\Cb_{\ub}}|(r\nablas)\Db\alphab|^2\lesssim\int_{C_0}|(r\nablas)\Db\betab|^2+\int_{\Cb_0}|(r\nablas)\Db\alphab|^2+\int_M|\tau_1^{4(1)}|.
\end{align*}

The terms which are $\Db$ applying to connection coefficients are expressed directly using null structure equations (except $\Db\omegab$). Therefore, using the bounds for $\mathcal{O}$, $\mathcal{R}$, $\underline{\mathcal{R}}$, Lemma \ref{curvatureonS}, H\"older and Sobolev inequalities, for $\varepsilon>0$ sufficiently small depending on $\mathcal{O}_0$, $\mathcal{R}_0$ and $\underline{\mathcal{R}}_0$, we have
\begin{align*}
&\int_M|\tau_1^4|+\int_M|\tau_1^{4(1)}|\\
\lesssim&\varepsilon^{1/2}C(\mathcal{O}_0,\mathcal{R}_0,\underline{\mathcal{R}}_0)+\int_0^{\ub}r^{-2}\mathcal{O}[\omega,\Omega\tr\chi-\overline{\Omega\tr\chi}]\int_{\Cb_{\ub'}}\sum_{i=0}^1|(r\nablas)^i\Db\alphab|^2\D \ub'\\
&+C(\mathcal{O}_0,\mathcal{R}_0)\left(\int_0^{\ub}r^{-2}\int_{\Cb_{\ub'}}\sum_{i=0}^2|(r\nablas)^i\alphab|^2\D\ub'\right)^{1/2}\left(\int_0^{\ub}r^{-2}\int_{\Cb_{\ub'}}\sum_{i=0}^1|(r\nablas)^i\Db\alphab|^2\D\ub'\right)^{1/2}\\
&+\mathcal{O}[\tr\chib,\omegab,\chibh]\int_0^{u}\int_{C_{u'}}\sum_{i=0}^1|(r\nablas)^i\Db\betab|^2\D u'\\
&+\mathcal{O}[\eta,\etab]\left(\int_0^{u}\int_{C_{u'}}\sum_{i=0}^2|(r\nablas)^i\Db\betab|^2\D u'\right)^{1/2}\left(\int_0^{\ub}r^{-2}\int_{\Cb_{\ub'}}\sum_{i=0}^1|(r\nablas)^i\Db\alphab|^2\D\ub'\right)^{1/2},\\
&+O^2[\Db\omegab]\left(\int_0^{u}\int_{C_{u'}}\sum_{i=0}^2|(r\nablas)^i\betab|^2\D u'\right)^{1/2}\left(\int_0^{u}\int_{C_{u'}}\sum_{i=0}^1|(r\nablas)^i\Db\betab|^2\D u'\right)^{1/2},
\end{align*}
where the terms which are already estimated are collected in the first term on the right hand side. By Gronwall's inequality, we have
\begin{equation}\label{energyDbalphab}
\begin{split}
\mathcal{R}[\Db\betab]^2+\underline{\mathcal{R}}[\Db\alphab]^2\lesssim&\mathcal{R}_0[\Db\betab]^2+\underline{\mathcal{R}}_0[\alphab]^2+C(\mathcal{O}_0,\mathcal{R}_0,\underline{\mathcal{R}}_0)\left(\underline{\mathcal{R}}[\Db\alphab]+\varepsilon^{1/2}+\varepsilon^{1/2}\mathcal{R}[\Db\betab]\underline{\mathcal{R}}[\Db\alphab]\right)\\
&+\varepsilon O^2[\Db\omegab]\mathcal{R}[\betab]\mathcal{R}[\Db\betab].
\end{split}
\end{equation}

Therefore we should first estimate $O^2[\Db\omegab]$ . We commute $\nablas^i\Db$ with the structure equation for $D\omegab$ for $i\le 1$:
\begin{align*}
D\nablas^i\Db\omegab=\nablas^i\Db(\Omega^2(2(\eta,\etab)-|\eta|^2-\rho))+[D,\nablas^i\Db]\omegab.
\end{align*}
We should integrate this equation from the last slice $\Cb_{\ub_*}$. Commuting $\Db$ with the equation \eqref{lastomegab} which is written on $\Cb_{\ub_*}$:
\begin{align*}
2\Deltas\Db\omegab=&\Db\left(2\divs(\Omega\betab)+\divs(3\Omega\chibh\cdot\eta+\frac{1}{2}\Omega\tr\chib\eta)+\Omega\tr\chib\divs\eta+(F-\overline{F}+\overline{\Omega\tr\chib\check\rho}-\overline{\Omega\tr\chib}\cdot\overline{\check{\rho}})\right)\\
&+2[\Deltas,\Db]\omegab,\\
\overline{\Db\omegab}=&\overline{\Omega\tr\chib}\overline{\omegab}-\overline{\Omega\tr\chib\omegab}-\Db\overline{\Omega\tr\chib\log\Omega}.
\end{align*}

We should pay special attention to the right hand side. Notice that, by null Bianchi equations for $\Db\betab$,  up to lower order terms, $\Db\divs(\Omega\betab)\sim\divs\divs\alphab$, which is not in $L^2(S_{\ub,u})$ but only in $L^2(\Cb_{\ub})$. But notice that we can control $\|(r\nablas)^i\alphab\|_{L^4(S_{\ub,u})}$ for $i\le 1$ using an analogue of Proposition 10.2 in \cite{Chr} or a special case for $Y=u$ in \eqref{SobolevCb}, as:
\begin{align*}
\|(r\nablas)^i\alphab\|_{L^4(S_{\ub,u})}\lesssim\mathcal{R}_0+C(\underline{\mathcal{R}}[\alpha], \underline{\mathcal{R}}[\Db\alpha]).
\end{align*}
Therefore, we should be able to control $\|(r\nablas)^i\Db\omegab\|_{L^4(S_{\ub,u})}$. In this paper, we only control $\|(r\nablas)^i\Db\omegab\|_{L^2(S_{\ub,u})}$ for simplicity, which is enough for our work.

To do this, we denote $\Deltas\Db\omegab=-\Omega\divs\divs\alphab+G_{(3)}$ on the last slice $\Cb_{\ub_*}$. It is direct to verify that $\|r^2G_{(3)}\|_{L^2(S_{\ub,u})}\lesssim C(\mathcal{O}_0,\mathcal{R}_0)$. Then we compute
\begin{align*}
\int_{S_{\ub_*,u}}|\nablas\Db\omegab|^2=&-\int_{S_{\ub_*,u}}\Deltas\Db\omegab(\Db\omegab-\overline{\Db\omegab})=\int_{S_{\ub_*,u}}(-\divs\divs\alphab+G_{(3)})(\Db\omegab-\overline{\Db\omegab})\\
\lesssim&\|G_{(3)}\|_{L^2(S_{\ub_*,u})}\|\nablas\Db\omegab\|_{L^2(S_{\ub_*,u})}+\int_{S_{\ub_*,u}}\divs\alphab\cdot\nablas\Db\omegab.
\end{align*}
Therefore,
\begin{align*}
\sum_{i=0}^1\|(r\nablas)^i\Db\omegab\|_{L^2(S_{\ub_*,u})}\lesssim\|r^2G_{(3)}\|_{L^2(S_{\ub_*,u})}+\|(r\nablas)\alphab\|_{L^2(S_{\ub_*,u})}\lesssim C(\mathcal{O}_0,\mathcal{R}_0,\underline{\mathcal{R}}_0,\underline{\mathcal{R}}[\Db\alphab]).
\end{align*}
Then we apply Gronwall type estimates to the equations for $D\nablas^i\Db\omegab$ for $i\le1$, and conclude that the above estimate holds for all $\ub$.

Substituting this estimate back to \eqref{energyDbalphab}, we have
\begin{align*}
\mathcal{R}[\Db\betab]+\underline{\mathcal{R}}[\Db\alphab]\lesssim C(\mathcal{O}_0,\mathcal{R}_0,\underline{\mathcal{R}}_0).
\end{align*}
This in turn gives the desired estimate of $\|(r\nablas)^i\Db\omegab\|_{L^2(S_{\ub,u})}$ for $i\le1$. By Sobolev inequality, we can also bound $\|r^{1/2}\Db\omegab\|_{L^4(S_{\ub,u})}$.

In order to estimate $\underline{\mathcal{R}}[\Db^2\alphab]$, we commute $\Db^2$ with the null Bianchi equations for $\Dh\alphab$ and $\Db\betab$. We do the estimate exactly as above, if we have the estimate for $\mathcal{O}[\Db^2\omegab]$. This can be done as above estimates for $\mathcal{O}^2[\Db\omegab]$, by commuting $\Db^2$ with the equation for $\Db\omegab$ and the equation \eqref{lastomegab}. Notice that in this case, when considering the Laplacian equation for $\Db^2\omegab$ on the last slice $\Cb_{\ub_*}$, the highest order term of the  is $\divs\divs(\Db\alphab)$, which has no estimates. But we have controlled $\divs(\Db\alphab)$ in $L^2(\Cb_{\ub})$. Therefore, by a similar argument, we can first control $\|\Db^2\omegab\|_{L^2(\Cb_{\ub_*})}$ and then integrate the equation for $D\Db^2\omegab$ to obtain the estimate of $\|\Db^2\omegab\|_{L^2(\Cb_{\ub})}$ for all $\ub$.\end{proof}

Now we consider the function space $\mathcal{K}=\mathcal{K}_{\ub_*+\delta,\varepsilon+\delta'}\subset C([0,\varepsilon],H^2(S_{\ub_*+\delta,0}))$ such that
\begin{equation}\label{functionspacelast}
\begin{cases}W(0,\theta)=0,\\0\le W(s,\theta)\le\varepsilon+\delta',\\\sup_s\|(r\nablas)^{\le2}(W(s,\cdot)-s)\|_{L^2(S_{\ub_*+\delta,0})}\le\varepsilon_\mathcal{K},\end{cases}
\end{equation}
for some small number $\varepsilon_\mathcal{K} $ to be fixed. $\mathcal{K}$ is a close set in $C([0,\varepsilon],H^2(S_{\ub_*+\delta,0}))$. We remark that the second condition ensure that there is enough room such that the new foliation ${}^{(W)}u$ given by $W\in\mathcal{K}$ also varies from $0$ to $\varepsilon$. This is exactly the point that why we need $\delta$ is sufficiently small depending on $\delta'$. Obviously we can first choose $\varepsilon_\mathcal{K}$ small enough such that the second condition is ensured. Also notice that by our definition of $\mathcal{K}$, $W$ need not to represent a foliation, but only a family of sections parameterized by $s$. As in the first step, we will prove
\begin{proposition}\label{fixedpoint}
For $\varepsilon$ sufficiently small depending on $\mathcal{O}_0$, $\mathcal{R}_0$, $\underline{\mathcal{R}}_0$, and $\delta$ sufficiently small (may depending on $\varepsilon$ and $\delta'$), $\mathcal{A}(\mathcal{K})\subset\mathcal{K}$ and $\mathcal{A}$ is a contraction in $\mathcal{K}\subset C([0,\varepsilon],H^2(S_{\ub_*+\delta,0}))$.
\end{proposition}

\begin{proof}It is not hard to see the assumptions of Proposition \ref{canonicalslice} hold for $\varepsilon$ sufficiently small depending on  $\mathcal{O}_0$, $\mathcal{R}_0$, $\underline{\mathcal{R}}_0$. We use \eqref{contraction} for $W_1=W$ and $W_2=s$. We recall that $\Cb_{\ub_*}$ is a canonical last slice with the foliation given by the origin $u$, then for any $\varepsilon_*>0$, by continuity, we can choose $\delta$ sufficiently small, such that
\begin{align*}
&\|(r\nablas)^{\le2}(\log{}^{((W)_{\mathcal{A}})}\Omega(W(s,\cdot),\cdot)-\log\Omega(s,\cdot))\|_{L^2(S_{\ub_*+\delta,0})}\\
\lesssim& C(\mathcal{O}_0,\mathcal{R}_0,\underline{\mathcal{R}}_0)\|(r\nablas)^{\le2}(W(s,\cdot)-s)\|_{L^2(S_{\ub_*+\delta,0})}+\varepsilon_*.
\end{align*}
Then we integrate $s$ from $0$ to $1$,
\begin{align*}
&\|(r\nablas)^{\le2}(W(s,\cdot)-s)\|_{L^2(S_{\ub_*+\delta,0})}\\
\lesssim&\int_0^s\|(r\nablas)^{\le2}{}^{((W)_{\mathcal{A}})}({}^{(W_\mathcal{A})}\Omega^2(W(s',\cdot),\cdot)\Omega^{-2}(W(s',\cdot),\cdot)-1)\|_{L^2(S_{\ub_*+\delta,0})}\D s\\
\lesssim& \varepsilon C(\mathcal{O}_0,\mathcal{R}_0,\underline{\mathcal{R}}_0)\sup_s\|(r\nablas)^{\le2}(W(s,\cdot)-s)\|_{L^2(S_{\ub_*+\delta,0})}+\varepsilon\varepsilon_*.
\end{align*}
The inequality \eqref{SobolevCb} should also be used here because we compare the bound of background lapse $\Omega$ on the sphere $S_{\ub_*+\delta,s}$ to the bound on $S_{W(s)}$. We choose $\varepsilon$ small depending only on $\mathcal{O}_0$, $\mathcal{R}_0$, $\underline{\mathcal{R}}_0$, then choose $\varepsilon_*$ sufficiently small depending on $\varepsilon_\mathcal{K}$ and $\varepsilon$, such that
\begin{align*}
\sup_s\|(r\nablas)^{\le2}(W(s,\cdot)-s)\|_{L^2(S_{\ub_*+\delta,0})}\le\varepsilon_\mathcal{K}.
\end{align*}
This implies $\mathcal{A}$ maps $\mathcal{K}$ into itself.

Also, we can use \eqref{contraction} to conclude that,
\begin{align*}
&\|(r\nablas)^{\le2}(W_1(s,\cdot)-W_2(s,\cdot))\|_{L^2(S_{\ub_*+\delta,0})}\\
\lesssim&\int_0^s\|(r\nablas)^{\le2}({}^{(W_\mathcal{A})}\Omega^2(W_1(s',\cdot))\Omega^{-2}(W_1(s',\cdot),\cdot)\\
&-{}^{(W_\mathcal{A})}\Omega((W_2,\cdot),\cdot)\Omega^{-2}(W_2(s,\cdot),\cdot)\|_{L^2(S_{\ub_*+\delta,0})}\D s\\
\lesssim& \varepsilon C(\mathcal{O}_0,\mathcal{R}_0,\underline{\mathcal{R}}_0)\sup_s\|(r\nablas)^{\le2}(W(s,\cdot)-s)\|_{L^2(S_{\ub_*+\delta,0})}.
\end{align*}
By choosing $\varepsilon$ sufficiently small depending only on $\mathcal{O}_0$, $\mathcal{R}_0$, $\underline{\mathcal{R}}_0$, then $\mathcal{A}:\mathcal{K}\to\mathcal{K}$ is a contraction and this completes the proof. \end{proof}

To proceed further, we are going to estimate the difference between two optical functions $u_1$ and $u_2$ which are constructed from canonical foliations on $\Cb_{\ub_1}$ and $\Cb_{\ub_2}$ respectively, and $\ub_1\le\ub_2$. The corresponding geometric quantities are labelled by lower index $1$ or $2$. In the past of $\Cb_{\ub_1}$, The null frames are related by
\begin{align*}
 \Lb'_2&=\Lb'_1,\\
L_2&=L_1+ (-\frac{1}{2}g(L'_1,L_2))\Lb_1+2 (-\frac{1}{2}g(L_1,L_2))^{1/2}\sigma^A (E_1)_A,\\
(E_2)_A&=(E_1)_A+\Omega_1^{-1}(-\frac{1}{2}g(L'_1,L_2))^{1/2}\sigma_A\Lb_1.
\end{align*}
Here $\sigma^1$, $\sigma^2$ satisfy $\sigma^A\sigma^Bg((E_1)_A,(E_1)_B)=1$. This relation is exactly the relation \eqref{nullframerelation}. 

Along the null generators of $(C_2)_{u_2}$, using the above relation, we compute
\begin{align}\label{D2u1u2}
D_2(u_1-u_2)=[L_1+ (-\frac{1}{2}g(L'_1,L_2))\Lb_1+2 (-\frac{1}{2}g(L_1,L_2))^{1/2}\sigma^A (E_1)_A](u_1)=-\frac{1}{2}g(L'_1,L_2),
\end{align}
and
\begin{align}\label{Db2u1u2}
\Db_2(u_1-u_2)=\Omega_2^2\Omega_1^{-2}-1.
\end{align}
We should consider $g(L'_1,L_2)$ to estimate $u_1-u_2$. Then we compute
\begin{equation}\label{D2L1L2}
\begin{split}
&D_2g(L'_1,L_2)\\
=&\omega_2 g(L'_1,L_2)+g(\nabla_{L_2}L'_1,L_2)\\
=&\omega_2 g(L'_1,L_2)-2g(L'_1,L_2)\sigma^A\sigma^B\Omega_1(\chih_1)_{AB}-g(L'_1,L_2)\Omega_1(\tr\chi)_1\\
&-4g(L'_1,L_2)(-\frac{1}{2}g(L_1,L_2))^{1/2}\sigma^A(\eta_1)_A-g(L'_1,L_2)^2\omegab_1,
\end{split}
\end{equation}
and
\begin{align}\label{Db2L1L2}\Db_2g(L'_1,L_2)=-2\Omega_1^{-2}\Omega_2^2\omegab_1g(L'_1,L_2)+4\Omega_1^{-2}\Omega_2^2(-g(L_1,L_2))^{1/2}(\sigma_1^A(\eta_1)_A+\sigma_2^A(\eta_2)_A).
\end{align}
Here $\sigma_1$ and $\sigma_2$ satisfy the same property as $\sigma$.

Now we go back to the canonical foliation given by $u_\delta$ on the new last slice $\Cb_{\ub_*+\delta}$, which we have constructed above. The origin optical function $u$ which is canonical on $\Cb_{\ub_*}$ is also extended up to $\Cb_{\ub_*+\delta}$ and the difference between $u$ and $u_\delta$ on $\Cb_{\ub_*+\delta}$ is controlled by $\varepsilon_\mathcal{K}$, see \eqref{functionspacelast}. Notice that we can choose $\varepsilon_\mathcal{K}$ arbitrarily small by setting $\delta$ sufficiently small. 

Our goal is to extend $u_\delta$ as a optical function back to the initial null cone $\Cb_0$. This means that both $u_\delta$ and $u$ satisfy the above system of equations \eqref{D2u1u2}-\eqref{Db2L1L2} with $u_1=u$ and $u_2=u_\delta$. With the ``initial data'' on $\Cb_{\ub_*+\delta}$ which are close to each other, we can conclude that the extension of $u_\delta$ to $\Cb_0$ can be done. This is essentially the continuity of the equations \eqref{D2u1u2} and \eqref{D2L1L2}. Again by continuity, the curvature norms $\mathcal{R}_\delta$, $\underline{\mathcal{R}}_\delta$ expressed in the new foliation $(\ub,u_\delta)$ is bounded by twice the norms $\mathcal{R},\underline{\mathcal{R}}$ in the old foliation $(\ub,u)$. Therefore, we complete Step 4 of the proof.

\subsection{Global Optical Function and Recovering the Full Decay} Finally, we will construct the gloabl retarted time function $u$ to complete the proof. We also appeal to the equations \eqref{D2u1u2}-\eqref{Db2L1L2}. Let $u_1$ and $u_2$ be optical functions from canonical foliations on $\Cb_{\ub_1}$ and $\Cb_{\ub_2}$ respectively, and $\ub_1\le\ub_2$. The space-time is constructed by sending $\ub_1$ to $+\infty$ and we will consider the convergence of the function $\ub_1$. Therefore, we need to compare $\ub_1$ to $\ub_2$ for any $\ub_2\ge\ub_1$.

We first consider \eqref{Db2L1L2} on $\Cb_{\ub_1}$. Because $L_1=L_2$ on $C_0$, then $g(L'_1,L_2)=0$ on $C_0$. We assume on $\Cb_{\ub_1}$, $|r_2^4(\ub_1)g(L'_1,L_2)|\lesssim \Delta_4$\footnote{$r_2$ should also depend on $u_2$ but we only emphasize the dependence on $\ub_1$ here.}. Then  by integrating \eqref{Db2L1L2}, we have
\begin{align*}
|r_2^4(\ub_1)g(L'_1,L_2)|\lesssim\varepsilon r_2^{-1}(\ub_1) C(\mathcal{O}_0,\mathcal{R}_0)\Delta_4+\varepsilon C(\mathcal{O}_0,\mathcal{R}_0)\Delta_4^{1/2}.
\end{align*}
By choosing $\varepsilon$ sufficiently small, we have $|r_2^4(\ub_1)g(L'_1,L_2)|\lesssim \frac{1}{2}\Delta_4$. Therefore, we have deduce that $|r_2^4(\ub_1)g(L'_1,L_2)|\lesssim C(\mathcal{O}_0,\mathcal{R}_0)$ on $\Cb_{\ub_1}$.

Now consider the equation \eqref{D2L1L2} in the past of $\Cb_{\ub_1}$. We are going to apply Gronwall type estimates to integrate this equation from $\Cb_{\ub_1}$ to its past. The term $\omega_2 g(L'_1,L_2)-2g(L'_1,L_2)\sigma^A\sigma^B\Omega_1(\chih_1)_{AB}$ can be absorbed by Gronwall's inequality because both $\omega_2$ and $\chih_1$ decay\footnote{Recall that both $r_1$ and $r_2$ are equivalent to $(1+\ub)$.} not slower than $r_2^{-2}$. On the other hand, by the structure equation for $\Db\tr\chi$, we can estimate, in the foliations given by $u_1,u_2$, 
\begin{align*}
|\Omega_i\tr\chi_i-\Omega_i\tr\chi_i|_{S_{\ub,0}}|\lesssim\varepsilon^{1/2}C(\mathcal{O}_0,\mathcal{R}_0)\frac{1}{r_i^{-2}}
\end{align*}
for $i=1,2$. By construction $\Omega_1\tr\chi_1|_{S_{\ub,0}}=\Omega_2\tr\chi_2|_{S_{\ub,0}}$, we can estimate
\begin{align*}
|\Omega_1\tr\chi_1-\Omega_2\tr\chi_2|\lesssim C(\mathcal{O}_0,\mathcal{R}_0)\frac{1}{r_2^{-2}}.
\end{align*}
Therefore, replacing the term $\Omega_1\tr\chi_1$ by $\Omega_2\tr\chi_2$ will add terms decaying no slower than $r_2^{-2}$, which can be absorbed by Gronwall's inequality. We are in the position to apply Gronwall type estimates,
\begin{align*}
|r_2^2(\ub)g(L'_1,L_2)|\lesssim&|r_2^2(\ub_1)g(L'_1,L_2)|_{\Cb_{\ub_1}}\\
&+\int_{\ub}^{\ub_1}r_2^2(\ub')|(-g(L'_1,L_2))^{3/2}\sigma^A(\eta_1)_A-g(L'_1,L_2)^2\omegab_1|\D\ub'.
\end{align*}
We argue by bootstrap again. We assume $|r_2^2(\ub)g(L'_1,L_2)|\le\Delta_5 r_2^{-2}(\ub_1)$. Taking into account that $|r_2^4(\ub_1)g(L'_1,L_2)|\lesssim C(\mathcal{O}_0,\mathcal{R}_0)$ on $\Cb_{\ub_1}$, we have
\begin{align*}
|r_2^2(\ub)g(L'_1,L_2)|\lesssim&(r_2^{-2}(\ub_1)+(\Delta_5^{3/2}r_2^{-3}(\ub_1)+\Delta_5^2r_2^{-4}(\ub_1))r_2^{-2}(0))C(\mathcal{O}_0,\mathcal{R}_0).
\end{align*}
Now we can choose $r_2^{-2}(\ub_1)$ sufficiently small, i.e., $\ub_1$ sufficiently large, and $\Delta_5$ sufficiently large, such that $|r_2^2(\ub)g(L'_1,L_2)|\lesssim\frac{1}{2}\Delta_5 r_2^{-2}(\ub_1)$. Then we conclude that $$|r_2^2(\ub)g(L'_1,L_2)|\lesssim C(\mathcal{O}_0,\mathcal{R}_0)r_2^{-2}(\ub_1).$$

We go back to the equation \eqref{Db2u1u2} and \eqref{D2u1u2}. By $\eqref{Db2u1u2}$ on $\Cb_{\ub_1}$, and $u_1=u_2$ on $C_0$, we conclude that $|u_1-u_2|\lesssim C(\mathcal{O}_0,\mathcal{R}_0)r_2^{-1}(\ub_1)$. Then by $\eqref{D2u1u2}$ in the past of $\Cb_{\ub_1}$, we have
\begin{align*}
|u_1-u_2|\lesssim&|u_1-u_2|_{\Cb_{\ub_1}}+\int_{\ub}^{\ub_1}|g(L'_1,L_2)|\D\ub'\\
\lesssim& C(\mathcal{O}_0,\mathcal{R}_0)r_2^{-1}(\ub_1),
\end{align*}
which means that, as $\ub_1$ tends to infinity, $u_1$ will converge to a global function $u$.

It is direct to check the convergence of the derivatives of $u_1$. Now consider $\nabla(u_1-u_2)=-\frac{1}{2}(L'_1-L'_2)$. Because $u_1-u_2\to 0$, we have $\Omega_1-\Omega_2\to0$. Also, the above computation shows that $g(L'_1,L_2)\to0$. By the relation between $L_1$ and $L_2$, we have $L'_1-L'_2\to0$ and we conclude $L'_1=-2\nabla u_1\to L'=-2\nabla u$ as $\ub_1\to+\infty$. In particular, $|\nabla u|=0$ and $u$ is a global optical function.

To verify that, under the global double null foliation $(\ub,u)$, the curvature norm $\mathcal{R}, \underline{\mathcal{R}}\le C(\mathcal{O}_0,\mathcal{R}_0,\underline{\mathcal{R}}_0)$, we need to show that the norms written in $(\ub,u_1)$, $\mathcal{R}_1$, $\underline{\mathcal{R}}_1$ converge to $\mathcal{R}$ and $\underline{\mathcal{R}}$. To do this, it is sufficient to show that $\nablas_2^ig(L'_1,L_2)\to0$ as $\ub_1\to+\infty$ for $i\le2$ in suitable norm. The case $i=0$ is proved just above. For the case $i=1$ we compute
\begin{align*}
(E_2)_Ag(L'_1,L_2)=g(\nabla_{{(E_2)}_A}L'_1,L_2)+g(L'_1,\nabla_{{(E_2)}_A}L_2).
\end{align*}
For the first term on the right hand side, we express $(E_2)_A$ and $L_2$ in terms of $(E_1)_A$, $\Lb_1$ and $L_1$. Then the first term on the right hand side equals to the sum of the form $\Gamma_1\cdot g(L'_1,L_2)^{\nu}$ with the power $\nu\ge1$. For the second term on the right hand side, we express $\nabla_{{(E_2)}_A}L_2=\Omega_2(\chi_2)_{A}^C(E_2)_C+\Omega_2(\etab_2)_A L_2$, and then express the $E_2$ and $L_2$ in terms of $E_1$, $L_1$ and $\Lb_1$. Therefore, the second term on the right hand side equals to the sum of $\Gamma_2\cdot g(L'_1,L_2)^{\nu}$ with $\nu\ge 1$. We already have $g(L'_1,L_2)\to 0$ and $E_2\to E$ where $E$ is tangent to the sphere $S_{\ub,u}$, therefore $\nablas_2g(L'_1,L_2)\to0$ as $\ub_1\to+\infty$. The second order derivaitves are done in a similar way, involving one order derivative of connection coefficients. Then $\nablas_2^2g(L'_1,L_2)\to0$ in $L^4(S_{\ub,u})$.

At last, if we denote
\begin{align*}\mathcal{R}[\alpha]=\sup_u\|(r\nablas)^{\le2}(r^2\alpha)\|_{L^2(C_u)},&\ \underline{\mathcal{R}}[\beta]=\sup_{\ub}\|(r\nablas)^{\le2}(r^2\beta)\|_{L^2(\Cb_{\ub})},\\
\mathcal{R}[D\alpha]=\sup_u\|(r\nablas)^{\le1}(r^3D\alpha)\|_{L^2(C_u)},&\ \underline{\mathcal{R}}[D\beta]=\sup_{\ub}\|(r\nablas)^{\le1}(r^3D\beta)\|_{L^2(\Cb_{\ub})},\\
\mathcal{R}[D^2\alpha]=\sup_u\|r^4D^2\alpha\|_{L^2(C_u)},&\ \underline{\mathcal{R}}[D^2\beta]=\sup_{\ub}\|r^4D^2\beta\|_{L^2(\Cb_{\ub})},
\end{align*}
We denote $\mathcal{R}_0[\alpha,D\alpha,D^2\alpha]$ to be the corresponding norms taken on $C_0$ and $\underline{\mathcal{R}}_0[\beta,D\beta,D^2\beta]$ to be the corresponding norms taken on $\Cb_0$.
We can show that if in addition,
$$\mathcal{R}_0[\alpha,D\alpha,D^2\alpha]+\underline{\mathcal{R}}_0[\beta,D\beta,D^2\beta]<\infty,$$ 
then $\mathcal{R}[\alpha,D\alpha,D^2\alpha],\underline{\mathcal{R}}[\beta,D\beta,D^2\beta]\le C(\mathcal{O}_0,\mathcal{R}_0,\underline{\mathcal{R}}_0)$, where we have included the above two groups of initial norms into the definition of $\mathcal{R}_0,\underline{\mathcal{R}}_0$. Such a decaying condition on $\alpha$ comes from the work of Christodoulou-Klainerman \cite{Ch-K} and Klainerman-Nicol\`o \cite{K-N}, which deal with the so-called strongly asymptotically flat Cauchy data.

We firstly consider the case for $\alpha$ itself. For angular derivatives $\nablas$, the proof is similar. We use the null Bianchi equation for $\Db\alpha$ and $D\beta$ to obtain
\begin{align*}
\int_{C_u}|r^2\alpha|^2+\int_{\Cb_{\ub}}2|r^2\beta|^2\lesssim\int_{C_0}|r^2\alpha|^2+\int_{\Cb_0}2|r^2\beta|^2+\int_Mr^4|\tau_4^{(0)}|.
\end{align*}
We estimate
\begin{align*}
&\int_Mr^4|\tau_4^{(0)}|\lesssim\int_Mr^4|\Gamma[\tr\chib,\omegab]\cdot\alpha\cdot\alpha+\Gamma[\eta,\etab]\cdot\beta\cdot\alpha+\chih\cdot R[\rho,\sigma]\cdot\alpha+\Gamma[\tr\chi,\omega]\cdot\beta\cdot\beta|.
\end{align*}
To estimate this term, we put all the curvature in suitably weighted $L^2(C_u)$ norm and obtain
\begin{align*}
\int_Mr^4|\tau_4^{(0)}|
\lesssim\varepsilon\mathcal{O}(\mathcal{R}[\alpha]^2+\mathcal{R}[\beta,\rho,\sigma]\mathcal{R}[\alpha]+\mathcal{R}[\beta]^2).
\end{align*}
Choosing $\varepsilon$ sufficiently small we can derive the desired bound for $\mathcal{R}[\alpha]$ and then $\underline{\mathcal{R}}[\beta]$.

\begin{remark}\label{gaindecay}
Though the proof is quite direct we also make some remarks on this estimate. We find that the decay rates of $\alpha$ and $\beta$ are the same, and the decay rate implied by the norm on incoming null cone $\underline{\mathcal{R}}[\beta]$ is weaker than that implied by the norm on outgoing null cone $\mathcal{R}[\beta]$, which is not the case for other components $\rho,\sigma,\betab$. Therefore, unlike the proof of Proposition \ref{curvaturecompletenullcone}, though $\tr\chi$ do appear in the error terms as $\tr\chi|\beta|^2$, we can also estimate this term by using $\mathcal{R}[\beta]$ instead of $\underline{\mathcal{R}}[\beta]$.
\end{remark}

We then go to $D\alpha$. We commute $D$ with the null Bianchi equations for $\Dbh\alpha$-$D\beta$, and obtain
\begin{align*}
&\Dbh\Dh\alpha-\Omega\nablas\tensor D\beta\\
=&[\Dbh,\Dh]\alpha+[\Dh,\nablas\tensor]\beta+\Omega\omega\nablas\tensor\beta\\&-\Dh\left(\Omega\{-\frac{1}{2}\tr\chib\alpha-2\omegab\alpha-(4\eta+\zeta)\tensor \beta+3\chih \rho+3{}^*\chih \sigma\}\right),\\
&DD\beta-\Omega\chih\cdot D\beta-\Omega\divs\Dh\alpha\\=&-\frac{3}{2}\Omega\tr\chi D\beta-\frac{3}{2}D(\Omega\tr\chi)\beta+D(\Omega\chih)\cdot\beta+[\Dh,\divs]\alpha+\Omega\omega\divs\alpha+D\left(\omega\beta+\Omega\{(\etab+2\zeta)\cdot\alpha\}\right).\\
\end{align*}
\begin{comment}
The term $\frac{3}{2}D(\Omega\tr\chi)\beta$ should be handled more carefully. We write
\begin{align*}
\frac{3}{2}D(\Omega\tr\chi)\beta=\frac{3}{2}\left(-\frac{1}{2}(\Omega\tr\chi)^2-|\Omega\chih|^2+2\omega(\Omega\tr\chi)\right)\beta,
\end{align*}
and substitute
\begin{align*}
\frac{3}{2}\Omega\tr\chi\beta=-D\beta-\left(-\Omega\chih\cdot\beta-\omega\beta-\Omega\{\divs\alpha+(\etab+2\zeta)\cdot\alpha\}\right).
\end{align*}
We then have
\begin{align*}
\frac{3}{2}D(\Omega\tr\chi)\beta=&\frac{1}{2}\Omega\tr\chi D\beta+\frac{3}{2}\left(-|\Omega\chih|^2+2\omega(\Omega\tr\chi)\right)\beta\\&-\frac{1}{2}\Omega\tr\chi\left(-\Omega\chih\cdot\beta-\omega\beta-\Omega\{\divs\alpha+(\etab+2\zeta)\cdot\alpha\}\right).\\\triangleq&\frac{3}{4}\Omega\tr\chi D\beta+F_1
\end{align*}
Then the equation for $DD\beta$ is written in the form
\begin{align*}
&DD\beta-\Omega\chih\cdot D\beta-\Omega\divs\Dbh\alphab\\=&2\Omega\tr\chi D\beta+F_1-D(\Omega\chih)\cdot\beta+[\Dh,\divs]\alpha+\Omega\omega\divs\alpha+D\left(-\omega\beta-\Omega\{(\etab+2\zeta)\cdot\alpha\}\right).\\
\end{align*}
Therefore, using the weight function $r^3$ can eliminate the appearance of $\Omega\tr\chi$, and only $\Omega\tr\chi-\overline{\Omega\tr\chi}$ appears. 
\end{comment}

The terms which are $D$ applying to connection coefficients (except $D\omega$) are expressed directly by the null structure equations, and $D\omega$ is estimated by commuting $D$ with the equation for $\Db\omega$\footnote{Unlike in the proof of Proposition \ref{DbDbalphab}, the estimate for $D\omega$ is not coupled with $D\alpha$ because we simply integrate the equation from initial null cone $C_0$.}. We can then complete the estimates by energy methods as above. We should notice that, the first term on the right hand side of the second equation $-\frac{3}{2}\Omega\tr\chi D\beta$ suggests the weight function $r^3$, and $\beta$ which appears in the second term $-\frac{3}{2}D(\Omega\tr\chi)\beta$ should be estimated by using $\mathcal{R}[\beta]$ instead $\underline{\mathcal{R}}[\beta]$ to gain enough decay, see Remark \ref{gaindecay}. There is another approach to the term $\frac{3}{2}\Omega\tr\chi D\beta$. According to the null Bianchi equation, $D\beta\sim\divs\alpha$ and can be estimated by using $\mathcal{R}[\alpha]$ up to lower order terms. This gains enough decay in the energy estimates. The estimate for $\nablas D\alpha$-$\nablas D\beta$ and then $D^2\alpha$-$D^2\beta$ are similar.

\end{document}